\documentclass[reqno]{elsarticle}

\usepackage{latexsym, amsmath, amssymb}
\usepackage{srcltx}
\usepackage{qtree}
\usepackage{float}
\usepackage{graphicx}
\usepackage{tikz}
\usepackage{tikz-qtree}
\usepackage{pifont}

\newtheorem{theorem}{Theorem}[section]
\newtheorem{lemma}[theorem]{Lemma}
\newtheorem{proposition}[theorem]{Proposition}
\newtheorem{corollary}[theorem]{Corollary}
\newdefinition{example}[theorem]{Example}
\newdefinition{definition}[theorem]{Definition}
\newdefinition{remark}[theorem]{Remark}

\newenvironment{proof}[1][Proof:]{\begin{trivlist}
\item[\hskip \labelsep {\bfseries #1}]}{\end{trivlist}}

\newcommand{\ifte}[3]{\mathsf{If\ }#1\mathsf{\ Then\ }#2\mathsf{\ Else\ }#3}

\newcommand{\Vd}{\mathcal{D}}
\newcommand{\Vu}{\mathcal{U}}

\newcommand{\eval}{e}
\newcommand{\val}{w}

\newcommand{\F}[1]{F{:}#1}
\newcommand{\T}[1]{T{:}#1}
\newcommand{\X}[2]{X_{#1}{:}#2}
\newcommand{\Y}[2]{Y_{#1}{:}#2}
\newcommand{\XX}[3]{X_{#1}^{#2}{:}#3}
\newcommand{\cut}{\mathcal{R}(\textsl{CUT\/})}

\newcommand{\reg}[1]{\mathcal{R}(#1)}
\newcommand{\abs}{\mathcal{R}(\textsl{ABS})}

\newcommand{\dom}{\textsl{dom}}

\newcommand{\Fs}{\F{}}
\newcommand{\Ts}{\T{}}

\newcommand{\C}{\theta}

\newcommand{\lthree}{\ensuremath{\textit{\L}_3}}

\usepackage{hyperref}

\begin{document}

\title{Bivalent semantics, generalized compositionality \\
       and analytic classic-like tableaux \\
       for finite-valued logics\tnoteref{t1}}

\tnotetext[t1]{The research reported in this paper falls within the scope of the EU FP7 Marie Curie PIRSES-GA-2012-318986 project \href{http://sqig.math.ist.utl.pt/GeTFun/}{GeTFun}: \textit{Generalizing Truth-Functionality}. The first author further acknowledges the support of FCT and EU FEDER via the project PEst-OE/EEI/LA0008/2013 of Instituto de Telecomunica\c c\~ oes.  The second author acknowledges partial support of CNPq.}

\author[lx]{Carlos Caleiro}
\ead{ccal@math.ist.utl.pt}
\address[lx]{SQIG-Instituto de Telecomunica\c c\~ oes and Dept.\ Mathematics, IST, U Lisboa, Portugal}

\author[rn]{Jo\~ao Marcos}
\ead{jmarcos@dimap.ufrn.br}
\address[rn]{LoLITA and Dept.\ of Informatics and Applied Mathematics, UFRN, Brazil}

\author[ve]{Marco Volpe}
\ead{marco.volpe@univr.it}
\address[ve]{Dipartimento di Informatica, Universit\`a di Verona, Italy}


\begin{abstract}
\noindent 
The paper is a contribution both to the theoretical foundations and to the actual construction of efficient automatizable proof procedures for non-classical logics.  We focus here on the case of finite-valued logics, and exhibit: (i)~a mechanism for producing a classic-like description of them in terms of an effective variety of bivalent semantics; (ii)~a mechanism for extracting, from the bivalent semantics so obtained, uniform (classically-labeled) cut-free standard analytic tableaux with possibly branching invertible rules and paired with proof strategies designed to guarantee termination of the associated proof procedure; (iii)~a mechanism to also provide, for the same logics, uniform cut-based tableau systems with linear rules.  The latter tableau systems are shown to be adequate even when restricted to analytic cuts, and they are also shown to polynomially simulate truth-tables, a feature that is not enjoyed by the former standard type of tableau systems (not even in the 2-valued case).  The results are based on useful generalizations of the notions of analyticity and compositionality, and illustrate a theory that applies to many other classes of non-classical logics.
\end{abstract}


\maketitle

\vspace{-4mm}
\section{Introduction} \label{aims}

\noindent
Our paper is a contribution to the modern study of deduction in many-valued logics, in line with the research from standard references such as~\cite{DBLP:books/el/RV01/BaazFS01,hah:AMVL:HPL}, and consummating the track of publications surveyed in~\cite{ccal:mar:survey12}.
The present paper deals with finite-valued logics --- logics whose connectives are semantically characterizable by truth-tables with a finite number of `algebraic truth-values'.  
We first recall that such logics may be alternatively characterized by way of bivalent semantics --- semantics with only two `logical values' (cf.~\cite{sus:75,ccal:car:con:mar:03a}).  Going beyond that, we show that such bivalent characterizations, based on a generalized notion of compositionality, can be produced in a constructive way, for any finite-valued logic.  Several technical problems that appear underway are shown to be circumventable.
Providing further evidence on how model-theoretic and proof-theoretic analyses have strong impact on each other, from our bivalent characterizations of finite-valued logics we show, in each case, how to extract adequate analytic classic-like tableau systems.  Analyticity, in these systems, is based on appropriate generalized versions of the subformula property and on the adoption, in each case, of convenient proof strategies.
While analytic tableaux for propositional logic are expected to yield decidability, there is no general reason to expect the associated decision procedure to be computationally feasible.  
In order to secure a measurable gain in proof complexity we show also how to extract, from our bivalent characterizations of finite-valued logics, alternative tableau systems that control the combinatorial explosion caused by intrinsic redundancies of usual analytic tableau methods. We show that these alternative systems can polynomially simulate truth-tables, the former thus not being `worse' than the latter.  Such cut-based tableaux generalize the so-called `KE system' for Classical Logic (cf.~\cite{dag:handbook:99}), in which all tableau rules are linear except for the (non-eliminable yet analytic) cut rule.

In Section \ref{TFvsBiv} we list the basic syntactical definitions about logics in general and the basic semantic definitions about finite-valued logics in particular, and contrast truth-functional semantics with classic-like (bivalent) semantics.  
Several well-known examples of truth-functional logics are introduced.  
Many-valued logics in general, and truth-functional logics in particular, are shown to be (non-constructively) reducible to bivalent semantics, alongside the lines of the so-called `Suszko's Thesis'.
To render such bivalent reduction constructive, for a given finite-valued truth-functional logic, a fine analysis of its expressiveness is due: in turn, we show how one may algorithmically check for sufficient expressiveness, and how one may generate upon demand a sufficiently expressive conservative extension of the given logic.
We then show how to produce an adequate classic-like characterization of any given finite-valued logic.  We also show that this characterization is rather robust: the collection of boolean statements that determines it, in the metalanguage, may be replaced by equivalent (and possibly more economical) collections of similar statements.
Our reductive mechanism gives rise to an effective variety of bivalent semantics, based on a generalization of the syntactical notion of subformula and a related broader take on the Principle of Compositionality of Meaning.

In Section \ref{TableauExtraction} we show how to deal with partial information and syntactical subtleties describing unobtainable semantic scenarios, that will lead to nonstandard additional closure rules in tableau systems.
Adequate classic-like tableaux are then shown to be extractible for each sufficiently expressive finite-valued logic.  
An extended notion of analyticity is guaranteed by a proof strategy to be coupled with a given proof system, based on an extended notion of formula complexity.

On the one hand, it is well known that proofs involving the cut rule (or, equivalently, modus ponens) can be dramatically shorter than the shortest cut-free proof of the same assertion (see, e.g.,~\cite{boolos:CUT}, and the discussion in \cite[Section 3.8]{dag:handbook:99}, where the introduction of cut-based KE tableau systems is motivated). On the other hand it is obvious that unrestricted use of cuts may lead to infinitary branching in proof search. Taking
those facts into account, restricted forms of cut have been investigated that imply gains in minimal proof size without rendering proof search unwieldy.  In particular, cut-based tableaux are based on a goal-directed form of employing analytic cuts, that is, cuts involving what we call generalized subformulas of formulas already to be found in a given branch.  Cut-based tableaux for Classical Logic are studied in~\cite{dag:mon:taming}.  In Section~\ref{CutBased} of the present paper we show how such systems may be uplifted to the realm of finite-valued logics.  Moreover, for Classical Logic it has been proved (cf.~\cite{dag:TabTT}) that (propositional) cut-based tableaux polynomially simulate the truth-table procedure while 
for some classes of formulas the shortest standard analytic tableaux may be exponentially larger than the truth-tables.
We extend these findings about proof complexity to finite-valued logics in general.

\section{Exploring the bivalence behind truth-functionality} \label{TFvsBiv}

\noindent
In what follows we propose a mechanism for producing a classic-like description of an arbitrary finite-valued logic in terms of an effective variety of bivalent semantics.  To accomplish such goal, we show how one may exploit the linguistic resources of a given logic, automatically checking for its sufficient expressiveness, and minimally extending it, in a conservative way, when necessary.

\subsection{Finite-valued logics}\label{sec:fvl}

\noindent 
Consider an alphabet consisting of a denumerable set~$\mathcal{A}$ of
\textsl{atomic variables} and a finite collection~$\Sigma$ of \textsl{connectives} (or \textsl{constructors}).  
By $\Sigma_k\subseteq\Sigma$ we will denote the collection of $k$-ary constructors in~$\Sigma$; the 0-ary connectives are also called \textsl{sentential constants}.
The set~$\mathcal{S}$ of \textsl{formulas}, as usual, is the carrier of the free $\Sigma$-algebra~$\mathbb{S}$ generated by~$\mathcal{A}$.
By $\varphi(q_1,\dots,q_k)$ we will denote a \textsl{statement-form} $\varphi\in\mathcal{S}$
written in the variables $q_1,\dots,q_k\in\mathcal{A}$; if $\psi=\varphi(\psi_1,\dots,\psi_k)$, for given $\psi_1,\dots,\psi_k\in\mathcal{S}$, we say that~$\psi$ is an \textsl{instance of}~$\varphi$.  By $\mathcal{S}(\varphi)$ we denote the set of all instances of~$\varphi$.
If~$\varphi\in\mathcal{S}$ contains some $k$-ary connective, for $k>1$, we call this formula \textsl{composite}; otherwise, that is, in case~$\varphi$ is either an atomic variable or a sentential constant, we call it \textsl{noncomposite}.  
The outermost constructor of a composite formula is called its \textsl{head} connective.  Formulas containing no atomic variables are called \textsl{ground}.  Given $\varphi={\odot}(\varphi_1,\varphi_2,\ldots,\varphi_k)$ in~$\mathcal{S}$, with ${\odot}\in\Sigma_k$, we call $\varphi_1,\ldots,\varphi_k\in\mathcal{S}$ the \textsl{immediate subformulas} of~$\varphi$.  
The set $\mathsf{sbf}(\varphi)$ of \textsl{subformulas of~$\varphi$} is obtained by closing~$\{\varphi\}$ under immediate subformulas, that is, it is the smallest set containing~$\varphi$ and the immediate subformulas of each element of~$\mathsf{sbf}(\varphi)$.  A \textsl{proper subformula of~$\varphi$} is any element of $\mathsf{sbf}(\varphi)\setminus\{\varphi\}$.  
These notions are extended from formulas to sets of formulas in the usual way.
A~canonical way of measuring the complexity of a given formula is by counting the nested occurrences of $k$-ary constructors in it, for $k>1$, that is, by inductively defining a mapping 
$\mathsf{dpth}:\mathcal{S}\longrightarrow\mathbb{N}$ such that:
$$\mathsf{dpth}(\varphi)=\left\{\begin{array}{ll}
				0 & \textrm{if }\varphi\textrm{ is noncomposite}\\[2mm]
				1{+}\mathop{\mathsf{Max}}\limits_{1\leq i\leq k}\mathsf{dpth}(\varphi_i) & \textrm{if }\varphi={\odot}(\varphi_1,\dots,\varphi_k)\\[-1mm]
				& \textrm{\quad for }{\odot}\in\Sigma_k\textrm{, }k\neq 0\textrm{ and }
				\varphi_1,\dots,\varphi_k\in\mathcal{S}
			\end{array}\right.
$$

In the present study, by
$\mathcal{V}_n=\{\frac{i}{n-1}:0\leq i< n\}$, where $n\in\mathbb{N}$, 
we will denote a set of
\textsl{truth-values}, partitioned into a set $\Vd_{m,n}\subseteq\mathcal{V}_n$ 
of \textsl{designated} values and a set $\Vu_{m,n}=\mathcal{V}_n\setminus\Vd_{m,n}$ of 
\textsl{undesignated} values.  
In what follows, we will often 
refer to~$0$ as~$F$ and to~$1$ as~$T$. 
In general, an (\textsl{$n$-valued}) \textsl{assignment} of truth-values to
the atomic variables is any mapping $\eval:\mathcal{A}\longrightarrow\mathcal{V}_n$, 
and a(n \textsl{$n$-valued}) \textsl{valuation} is any extension
$\val^\eval:\mathcal{S}\longrightarrow\mathcal{V}_n$ of such an assignment 
to the set of all formulas.  
Given some $\mathcal{R}\subsetneq\mathcal{S}$ and some valuation~$\val$, the restriction $\left.\val\right|_\mathcal{R}$ will be called a \textsl{partial valuation over $\mathcal{R}$}.
An \textsl{$n$-valent semantics} for~$\mathcal{S}$ based on~$\mathcal{V}_n$, 
then, is simply an arbitrary collection of $n$-valued valuations.  
In particular, we will call \textsl{bivalent} any semantics over
$\mathcal{V}_2=\{F,T\}$, and say it is \textsl{classic-like} in case $\Vd_{1,2}=\{T\}$; the
corresponding valuations are called \textsl{bivaluations}.
The canonical notion of \textsl{entailment}
$\models_\mathsf{Sem}\;\subseteq\mathsf{Pow}(\mathcal{S})\times\mathcal{S}$
associated to an $n$-valent semantics $\mathsf{Sem}$ and characterizing a 
\textsl{logic}~$\langle\mathcal{S},\models_\mathsf{Sem}\rangle$ is defined by setting, 
for arbitrary $\Gamma\cup\{\alpha\}\subseteq\mathcal{S}$, ($\Gamma\models_\mathsf{Sem}\alpha$) iff 
($\val[\Gamma]\subseteq\Vd_{m,n}$ implies $\val(\alpha)\in\Vd_{m,n}$, 
for every $\val:\mathcal{S}\longrightarrow\Vd_{m,n}\!\cup\Vu_{m,n}$ in $\mathsf{Sem}$), where $\val[\Gamma]=\{\val(\gamma):\gamma\in\Gamma\}$.
If $\varnothing\models_\mathsf{Sem}\alpha$, we say that~$\alpha$ is a \textsl{valid} formula.
Subscripts in the sets of truth-values will be dropped whenever there is no risk of ambiguity.

Now, in any $n$-valent semantics one can clearly notice a shade of bivalence resting upon the opposition between designated and undesignated truth-values. 
It is not difficult to see that we can take advantage of this in order to transform such an $n$-valent semantics into a classic-like bivalent semantics which is undistinguishable from the former semantics from the viewpoint of the associated notions of entailment. To see that, consider the total mapping $t^{m,n}:\mathcal{V}_n\longrightarrow\mathcal{V}_2$ such that $t^{m,n}(v)=T$ iff $v\in\Vd_{m,n}$.  Then:

\begin{definition}\label{dyadic}
Let $\mathcal{L}=\langle\mathcal{S},\models_\mathsf{Sem}\rangle$ be an $n$-valent logic 
with a semantics $\mathsf{Sem}$. 
For each valuation $\val:\mathcal{S}\longrightarrow\Vd_{m,n}\!\cup\Vu_{m,n}$, 
consider the bivaluation $b_\val=t^{m,n}\circ \val$.  
We call $\mathsf{Sem}_2=\{b_\val:\val\in\mathsf{Sem}\}$ the \textsl{$\mathsf{S}$-reduction} of $\mathsf{Sem}$.
\end{definition}

\begin{proposition}\label{SredPre}
Any $n$-valent logic $\mathcal{L}=\langle\mathcal{S},\models_\mathsf{Sem}\rangle$ can be characterized by its $\mathsf{S}$-reduction, 
in other words, $\mathsf{Sem}$ and $\mathsf{Sem}_2$ characterize the same logic~$\mathcal{L}$.
\end{proposition}
\begin{proof}
It is straightforward to check that $\models_\mathsf{Sem}{=}\models_{\mathsf{Sem}_2}$, as a consequence of the fact that $b_\val(\varphi)\in\Vd_{1,2}$ iff $b_\val(\varphi)=t^{m,n}(\val(\varphi))=T$ iff $\val(\varphi)\in\Vd_{m,n}$, 
for any $n$-valent valuation $\val\in\mathsf{Sem}$.\qed
\end{proof}

A particularly interesting case of $n$-valent semantics, hereupon dubbed \textsl{truth-func\-tion\-al},
obtains when the semantics is presented by way of an appropriate $\Sigma$-algebra~$\mathbb{V}$ with carrier~$\mathcal{V}$, when we associate to each connective ${\odot}\in\Sigma_k$ an
operator $\widehat{{\odot}}:\mathcal{V}^k\longrightarrow\mathcal{V}$ in~$\mathbb{V}$,
and we collect in $\mathsf{Sem}$ the set of all homomorphisms
$\val:\mathbb{S}\longrightarrow\mathbb{V}$. 
Any such homomorphism may be construed as the free extension of some as\-sign\-ment $\eval:\mathcal{A}\longrightarrow\mathcal{V}$
into a valuation $\val^\eval:\mathbb{S}\longrightarrow\mathbb{V}$ by 
imposing that
$\val({\odot}(\varphi_1,\ldots,\varphi_k))=\widehat{{\odot}}(\val(\varphi_1),\ldots,\val(\varphi_k))$. Given a formula $\varphi(q_1,\dots,q_k)$ we will write
$\widehat{\varphi}(x_1,\dots,x_k)$ to denote the value
$\val(\varphi)$ assigned to $\varphi$ by any homomorphism
$\val:\mathbb{S}\longrightarrow\mathbb{V}$ 
such that $\val(q_r)=x_r$ for every $1\leq r\leq k$. 
One might say that such a truth-functional semantics is `compositional' in that the meaning it attributes to a composite formula depends (functionally) on the meaning of its immediate subformulas. 

\begin{definition}\label{nval}
A logic characterized by truth-functional
means, for a given~$\mathcal{V}_n$, is
called~$n$-\textsl{valued}. An $n$-valued
logic~$\mathcal{L}$ with an entailment relation
$\models_\mathsf{Sem}$ is said to be \textsl{genuinely}
$n$-\textsl{valued} in case there is no $n^\prime<n$ such
that~$\models_\mathsf{Sem}$ can be canonically obtained by way of
an $n^\prime$-valued truth-functional semantics.
\end{definition}

\begin{example}[Some well-known truth-functional logics]\label{exTFlogics}
Recall the set of truth-values~$\mathcal{V}_n=\{0,\frac{1}{n-1},\ldots,\frac{n-2}{n-1},1\}$, and consider initially the collection of binary connectives $\Sigma^a_2=\{\land,\lor\}$, 
that will be interpreted by setting~$\widehat{\land}=\lambda x\, y.\textsf{Min}(x,y)$ and $\widehat{\lor}=\lambda x\, y.\textsf{Max}(x,y)$.
Let's introduce a single unary connective through the collection $\Sigma^a_1=\{\neg\}$ and interpret this connective by setting $\widehat{\neg}=\lambda x.(1-x)$.
Assume $\Sigma^a=\Sigma^a_1\cup\Sigma^a_2$.
Classical Logic is obtained now if we fix $n=2$ and $\Vd_{1,n}=\{1\}$.  
Kleene Logic is obtained by fixing $n=3$ and $\Vd_{1,n}=\{1\}$, 
and Asenjo-Priest Logic again fixes $n=3$ but differs from Kleene in fixing $\Vu_{2,n}=\{0\}$.
Consider now an extra binary connective given by $\Sigma^b_2=\{\supset\}$, and assume that $\Sigma^b=\Sigma^a\cup\Sigma^b_2$.
For all logics below, we will consider $\varphi\equiv\psi$ as simply an abbreviation for $(\varphi\supset\psi)\land(\psi\supset\varphi)$.
To define the hierarchy of $n$-valued logics~\L$_n$, proposed by \L ukasiewicz, for $n\geq 2$, 
interpret~$\land$, $\lor$ and~$\neg$ as above, and interpret 
$\widehat{\supset}=\lambda x\, y.\textsf{Min}(1,1-x{+}y)$.  
To define the hierarchy of $n$-valued logics~G$_n$, proposed by G\"odel, for $n\geq 2$, interpret~$\land$ and~$\lor$ again as above, 
but now interpret~$\widehat{\neg}=\lambda x.(\ifte{x=0}{1}{0})$, 
and $\widehat{\supset}=\lambda x\,y.(\ifte{x\leq y}{1}{y})$.
For each \L$_n$ and~G$_n$ we fix~$\Vd_{1,n}=\{1\}$.  
Note that the interpretation of~$\supset$ for~\L$_n$ and for~G$_n$ coincide if $n=2$, and may be defined in that case by setting $\varphi\supset\psi$ as an abbreviation for $(\neg\varphi)\lor\psi$; this coincidence no longer obtains if $n>2$.  
Further, in the \L$_n$ it is enough to take $\Sigma^a_1\cup\Sigma^b_2$ as a choice of primitive connectives, given that $x\widehat{\lor}y=(x\widehat{\supset}y)\widehat{\supset}y$, and that $x\widehat{\land}y=\widehat{\neg}(\widehat{\neg}x\widehat{\lor}\widehat{\neg}y)$.
\end{example}

It shoud be clear that:

\begin{proposition}\label{Sred}
Any truth-functional logic can be characterized by a classic-like bivalent semantics.
\end{proposition}
\begin{proof}
This follows in fact as a corollary of Prop.~\ref{SredPre}, where we now start with a semantics $\mathsf{Sem}$ written in terms of $n$-valued homomorphisms~$\val:\mathbb{S}\longrightarrow\mathbb{V}$, where~$\Vd_{m,n}$ is fixed as the set of designated values in the carrier of~$\mathbb{V}$.
\qed
\end{proof}

\begin{remark}
\label{suszkothesis}
The idea that any semantics can be converted~/ reduced to a bivalent semantics is known as \textsl{Suszko's Thesis} (cf.~\cite{cal:car:con:mar:humbug:05,Shr:Wan:2011}).  
For the truth-func\-tion\-al case, the underlying intuition is that the `algebraic truth-values' from the carrier of~$\mathbb{V}$ should be distinguished from the `logical values' (namely, $F$ and~$T$: `the False' and `the True', according to Roman Suszko).  
The paradoxicality of such a reduction would seem to reside in regarding a logic at times as truth-functional and at other times simply as bivalent (cf.~\cite{mal:93,ccal:car:con:mar:03a}).  
One should be wary not to confuse though, on the one hand, a logic as a structure in which a set of formulas is endowed with a consequence relation enjoying certain properties, and, on the other hand, the variegated forms in which such a consequence relation may be characterized by semantical means (cf.~\cite{mar:09a:full}).
\end{remark}

The bivalent semantics produced by the instructions laid out in Def.~\ref{dyadic} is obviously classic-like, and if the input logic is genuinely $n$-valued, for $n>2$, the output semantics cannot be truth-functional.  
In particular, while a truth-functional characterization seems obviously attractive for its good behavior, it is not clear that the same can be said about the bivalent characterization thereby originated.  Is the latter set of bivaluations at least describable recursively, without resource to the original set of $n$-valued valuations?  
Can the reduction from $\mathsf{Sem}$ to $\mathsf{Sem}_2$ at least be done \textit{constructively}, in the truth-functional case?
Furthermore, must such reduction throw away for good the fundamental feature of \textit{compositionality}, together with truth-functionality? 
Will the meaning of a formula no longer be related to the meaning of its subformulas? 
The answer to the first two questions will be affirmative if we find a way of appropriately exploiting the original linguistic resources of the given logic, or else extend such resources conveniently in order to make the logic sufficiently expressive.  
The answer to the final two questions will be negative if we find a way of being more generous about the very meaning of compositionality.  
We will next discuss these issues, and show how they can be satisfactorily resolved to our benefit.

\subsection{Separation of truth-values}\label{sec:separation}

\noindent
An $n$-valued logic~$\mathcal{L}$ is said to be \textsl{functionally complete} if any operation $\widehat{\varphi}:(\mathcal{V}_n)^k\longrightarrow\mathcal{V}_n$ is the interpretation of some statement form $\varphi(q_1,\dots,q_k)$ expressible in the language of~$\mathcal{L}$.
Besides being 2-valued (thus, bivalent), Classical Logic is the only logic in Ex.~\ref{exTFlogics} that enjoys functional completeness.
The so-called Post Logics~P$_n^m$ are functionally complete genuinely $n$-valued logics with~$m$ designated values, for $n>2$ and $0<m<n$, and they may be defined by adding unary permutation operators to the~\L$_n$ and setting $\Vu=\{0,\frac{1}{n-1}\ldots,\frac{n-m-1}{n-1}\}$.  Of course, P$_2^1$ coincides with Classical Logic.  
Functional completeness is a rare property, enjoyed only by extremely expressive logics.

In producing an algorithmic version of the $\mathsf{S}$-reduction, which identifies every designated value as a `true' value and identifies every undesignated value as a `false' value, the challenge is to still be able to somehow recover information about the original `algebraic' values even after the classic-like reduction is produced.
In all cases, the idea will be to check whether a given logic is \textit{expressive enough} so as to allow for its original truth-values to be uniquely described by way of its original linguistic resources.
To that effect, we will look for a way of distinguishing each pair of
values of a genuinely $n$-valued logic~$\mathcal{L}$. %

\begin{definition}\label{separable}
Given $x, y\in \mathcal{V}_n$, we write $x\;\#\;y$ and
say that~$x$ and~$y$ are \textsl{separated} 
in case one value is designated and the other undesignated, that is,
in case $t(x)\neq t(y)$. 
We say that a one-variable formula $\theta^{xy}(p)$ of~$\mathcal{L}$
\textsl{distinguishes}~two truth-values~$x$ and~$y$ if
$\widehat{\theta^{xy}}(x)\;\#\;\widehat{\theta^{xy}}(y)$.
In that case we will also say that the values~$x$ and~$y$ 
of~$\mathcal{L}$ are \textsl{distinguishable}, as they may be separated using just the
linguistic resources of~$\mathcal{L}$. 
Finally, a logic~$\mathcal{L}$ is called \textsl{separable} in case its truth-values are pairwise 
distinguishable, that is, in case for any pair of distinct values 
$\langle x,y\rangle\in\mathcal{V}\times\mathcal{V}$ 
there exists a one-variable \textsl{separator} formula~$\theta^{xy}(p)$ that distinguishes~$x$ and~$y$. 
\end{definition}

Obviously, in functionally complete genuinely $n$-valued logics, by design, any pair of values is distinguishable --- thus, any such logic is separable.
For other logics, when the separation of all truth-values is at all possible, we will often assume some appropriate collection of one-variable separators to have been listed as a finite sequence~$\theta_{1},\ldots,\theta_{s}$, and we will further use~$\theta_0$ to denote the identity mapping $\textsf{id}=\lambda p.p$ 
(notice indeed that $\theta_0(p)$ by itself suffices to distinguish any pair of values $\langle x,y\rangle\in(\Vd\times\Vu)\cup(\Vu\times\Vd)$).
From here on, any such $\overline{\theta}=\langle\theta_r\rangle_{r=0}^s$ will be dubbed a \textsl{separating sequence} for the given logic. 

\begin{remark}
\label{sepbounds}
It is worth remarking that the length $s+1$ of the separating sequence $\langle\theta_r\rangle_{r=0}^s$ must be such that $\log_2(n)\leq s+1< n$ for any $n$-valued logic.  In fact, considering that the same suitably designed separator formula could be used to distinguish some pair of designated values and simultaneously also to distinguish some pair of undesignated values, it will be sufficient in the best scenario to have precisely $\log_2(\max(|\Vd|,|\Vu|))+1$ separator formulas in the separating sequence of a given $|\mathcal{V}|$-valued logic, where $\mathcal{V}=\Vd\cup\Vu$.
\end{remark}

\begin{definition}\label{binprint}
Fixed a separating sequence $\overline{\theta}=\langle\theta_r\rangle_{r=0}^s$ for a given $n$-valued logic~$\mathcal{L}$, the \textsl{binary print} of a value $z\in\mathcal{V}_n$
is the sequence
$\overline{\theta}(z)=\langle b_\val(\theta_{r}(p))\rangle_{r=0}^s$, where $\val(p)=z$.
We dub $\overline{\theta}[\mathcal{V}_n]$ the set of \textsl{obtainable} binary prints; intuitively, they are the binary prints that correspond and uniquely describe some actual truth-value from the given $n$-valued semantics.
\end{definition}

Notice that $\overline{\theta}(z)=\langle t({\widehat{\theta_r}}(z))\rangle_{r=0}^s$. 
More importantly, for each pair of distinct values 
$\langle x,y\rangle\in\mathcal{V}_n\times \mathcal{V}_n$ 
it is now obviously the case that
$\overline{\theta}({x})\neq \overline{\theta}({y})$.

\begin{example}[Some separable logics] \label{lukas}\label{ex:separatingTF}
Recall Ex.~\ref{exTFlogics}.  It should be clear that the two values of Classical Logic are separated by $\theta_0(p)=p$.  Also, the two undesignated values of Kleene Logic and the two designated values in Asenjo-Priest Logic are separated by adding $\theta_1(p)=\neg p$ to the separating sequence: indeed, such $\theta_1(p)$ helps in distinguishing the binary prints of~$0$ and~$\frac{1}{2}$ in Kleene, and in distinguishing the binary prints of~$\frac{1}{2}$ and~$1$ in Asenjo-Priest.
Consider now \L ukasiewicz logic~\L$_n$, with $n>2$.  In that case we have to devise a way of pairwise separating each of its $n-1$ undesignated values.  
For that purpose one may consider a collection of 
operators~$j^m_{\geq}$, for $0< m< n-1$, 
such that $\widehat{j^m_{\geq}}=\lambda z.(\ifte{z\geq \frac{m}{n-1}}{1}{0})$ 
---
it is worth noticing that such~$j^m_{\geq}$ operators may be defined as abbreviations using solely the connectives in~$\Sigma^b$ (cf.~\cite{ros:tur:52}).
Clearly a~$j^m_{\geq}$ operator separates the undesignated value~$\frac{m}{n-1}$ of~\L$_n$ from all the lower values.  Indeed, an appropriate separating sequence~$\overline{\theta}=\langle\theta_r\rangle_{r=0}^{n-2}$ for~\L$_n$ may be defined by setting $\theta_r=j_\geq^{(n-1)-r}$.  For such choice, we see that 
$\overline{\theta}\left(\frac{m}{n-1}\right)$, the $(n-1)$-long binary print of~$\frac{m}{n-1}$, will be an $(n-1-m)$-long sequence of~$F$s, followed by an $m$-long sequence of~$T$s.
\end{example}

A word is due here with respect to the general problem of distinguishing truth-values. 
It turns out that not every $n$-valued logic~$\mathcal{L}$ is separable, even if~$\mathcal{L}$ is genuinely $n$-valued, as illustrated below in Ex.~\ref{godel} --- the original language of the logic~$\mathcal{L}$ may simply fail to be \textit{sufficiently expressive}. 
This fact would seem to pose a limitation to the methods proposed in the present paper. However, this is by no means a serious limitation.
Indeed, as we will show in what follows, it is not difficult to see that a clever search may be employed to efficiently decide the separability of any given finite-valued logic, and a simple procedure may be devised to output a separating sequence in case it exists.

\begin{remark}
\label{cleversearch}
Let $\mathcal{F}_n$ denote the set of all unary operations on~$\mathcal{V}_n$. 
To decide whether a given $n$-valued logic $\mathcal{L}$ is separable, it suffices to compute the set of all unary functions $f:\mathcal{V}_n\longrightarrow\mathcal{V}_n$ that are definable by the connectives in~$\Sigma$. Since this set must be finite (there are only $n^n$ functions in $\mathcal{F}_n$), one can then test each of the definable functions on the pairs of values that demand separation. 

Note that the definable unary functions are precisely those that can be expressed by $\widehat{\alpha}$ where $\alpha(p)\in\mathcal{S}$ is a formula written with at most one variable. This gives us a simple way of computing the set of all definable unary functions, using Kleene's fixed-point theorem~\cite{kle:lfp}, as the least fixed-point of the operator $\mu:2^{\mathcal{F}_n}\longrightarrow 2^{\mathcal{F}_n}$ defined by
$\mu(H)=H\cup\{\mathsf{id}\}\cup\{\widehat{c}:c\in\Sigma_0\}\cup
                   \{\lambda x.\widehat{{\odot}}(h_1(x),\dots,h_k(x)):{\odot}\in\Sigma_k\mbox{, for }k\neq 0\mbox{, and }h_1,\dots,h_k\in H\}.$
This operator is clearly Scott-continuous as any function in $\mu(H)$ depends only on finitely many elements of $H$, and thus the set of definable unary functions is given in particular by the least~$m$ such that $\mu^m(\varnothing)=\mu^{m+1}(\varnothing)$.  Obviously, $m\leq n^n$.
\end{remark}

We will consider in what follows the extension of an $n$-valued logic by the addition of connectives with $n$-valued interpretations.
In other words, let~$\mathcal{L}$ be an $n$-valued logic given by means of the collection of all homomorphisms from the algebra of formulas $\mathbb{S}$ to an $n$-valued $\Sigma$-algebra~$\mathbb{V}$, and consider an extension $\mathcal{L}^{+}$ of~$\mathcal{L}$, defined over the extended algebra of formulas~$\mathbb{S}^{+}$ obtained from the extended set of connectives~$\Sigma^{+}$, and characterized by the collection $\mathsf{Sem}^{+}$ of all homomorphisms to a properly extended $n$-valued algebra $\mathbb{V}^{+}$.  It is clear that~$\mathcal{L}^{+}$ is always a \textsl{conservative extension} of~$\mathcal{L}$ in the sense that $\Gamma\models_\mathsf{Sem}\alpha$ if and only if $\Gamma\models_{\mathsf{Sem}^{+}}\alpha$ for every pair $\Gamma\cup\{\alpha\}\subseteq\mathcal{S}$. 
We will see next how one such conservative extension can be built in order to upgrade a nonseparable logic~$\mathcal{L}$ into a separable logic~$\mathcal{L}^{+}$.

When $\mathcal{L}$ is genuinely $n$-valued, and determined by a set $\Vd_{m,n}\subseteq\mathcal{V}_n$ of designated values, the structure $\langle\mathbb{V},\Vd_{m,n}\rangle$ is often dubbed a \textsl{logical matrix} (cf.~\cite{Wojcicki88}).
It should be clear that the \textsl{Leibniz congruence} (in the sense of~\cite{BP89}) of any such logical matrix is the identity. 
That is to say that the matrix is \textsl{simple}, meaning that every non-trivial congruence of the algebra~$\mathbb{V}$ must equate designated with undesignated values. 
Indeed, if that were not the case, then one could use any such non-trivial congruence to quotient~$\mathbb{V}$ (and~$\Vd_{m,n}$) and obtain a truth-functional semantics for the logic~$\mathcal{L}$ with less than $n$ truth-values. It turns out that such property can be used to compute a convenient extension of the primitive collection of connectives of~$\mathcal{L}$ whenever this logic is not separable. 

\begin{proposition}\label{prop:genuine}
Every genuinely $n$-valued logic has a separable genuinely $n$-valued conservative extension.
\end{proposition}
\begin{proof}
In the following, we shall employ~$\overrightarrow{v}^{(m)}$ as notation for a list $\underbrace{v,\dots,v}_{\textrm{$m$ times}}$.

Let $x$ and $y$ be two truth-values that cannot be distinguished by formulas of~$\mathcal{L}$. Then, by identifying the two, one generates a non-trivial congruence of the algebra~$\mathbb{V}$ that must therefore also identify a designated with an undesignated value.
So, there must exist a formula $\varphi(q_1,\dots,q_k)\in\mathbb{S}$ in $k>1$ variables and values $l_1,\dots,l_k,r_1,\dots,r_k\in\mathcal{V}_n$
such that:
\begin{equation}
\label{eq:eq1}
\widehat{\varphi}(l_1,\dots,l_k)\;\#\;\widehat{\varphi}(r_1,\dots,r_k)
\end{equation}
where, for each $1\leq i\leq k$, either 
$(a)$ 
$l_i=r_i$, or 
$(b)$ 
$\{l_i,r_i\}=\{x,y\}$.
Of course option $(b)$ must be satisfied at least once for $1\leq i\leq k$. Let us assume, without loss of generality, that option $(a)$ is satisfied for $1\leq i\leq j$ for some $j<k$. We have, then
\begin{equation}
\label{eq:eq2}
\widehat{\varphi}(r_1,\dots,r_j,l_{j{+}1},\dots,l_k)\;\#\;\widehat{\varphi}(r_1,\dots,r_j,r_{j{+}1},\dots,r_k).
\end{equation}
Now, if $l_{j{+}1}=\dots=l_k$ then $r_{j{+}1}=\dots=r_k$. Assuming, without loss of generality, that $l_{j{+}i}=x$ and $r_{j{+}i}=y$, for $0<i<k-j$, 
we thus have
\begin{equation}
\label{eq:eq21}
\widehat{\varphi}(r_1,\dots,r_j,\overrightarrow{x}^{(k-j)})\;\#\;\widehat{\varphi}(r_1,\dots,r_j,\overrightarrow{y}^{(k-j)}).\tag{2.1}
\end{equation}
In this case, one may distinguish $x$ and $y$ by introducing a unary connective~$\theta$ such that $\widehat{\theta}=\lambda z.\widehat{\varphi}(r_1,\dots,r_j,\overrightarrow{z}^{(k-j)})$. Alternatively, one could introduce (only the necessary) sentential constants $a_1,\dots,a_j$ such that $\widehat{a_i}=r_i$ for $1\leq i\leq j$, defining the separator $\theta=\lambda q.\varphi(a_1,\dots,a_j,\overrightarrow{q}^{(k-j)})$.

Otherwise, assume, again without loss of generality, that $l_{j{+}1}=\dots=l_p=x$ and $l_{p{+}1}=\dots=l_k=y$ for some $j<p<k$. Of course, 
one will then have $r_{j{+}1}=\dots=r_p=y$ and $r_{p{+}1}=\dots=r_k=x$. The situation is described by
\begin{equation}
\label{eq:eq22}
\widehat{\varphi}(r_1,\dots,r_j,\overrightarrow{x}^{(p-j)},\overrightarrow{y}^{(k-p)})\;\#\;\widehat{\varphi}(r_1,\dots,r_j,\overrightarrow{y}^{(p-j)},\overrightarrow{x}^{(k-p)}).\tag{2.2}
\end{equation}
Take the expression $\widehat{\varphi}(r_1,\dots,r_j,\overrightarrow{y}^{(k-j)})$. Clearly, its value must be separated from one of the two expressions in (\ref{eq:eq22}). Assume, yet again without loss of generality, that it is separated from the first expression, i.e., 
\begin{equation}
\widehat{\varphi}(r_1,\dots,r_j,\overrightarrow{x}^{(p-j)},\overrightarrow{y}^{(k-p)})\;\#\;\widehat{\varphi}(r_1,\dots,r_j,\overrightarrow{y}^{(p-j)},\overrightarrow{y}^{(k-p)}).\tag{2.3}
\end{equation}

In this case, we can separate $x$ and $y$ by introducing a unary connective~$\theta$ such that $\widehat{\theta}=\lambda z.\widehat{\varphi}(r_1,\dots,r_j,\overrightarrow{z}^{(p-j)},\overrightarrow{y}^{(k-p)})$. Alternatively, one could introduce the sentential constants $a_1,\dots,a_j,a_y$ such that $\widehat{a_i}=r_i$ for $1\leq i\leq j$ and $\widehat{a_y}=y$, and define the separator $\theta=\lambda q.\varphi(a_1,\dots,a_j,\overrightarrow{q}^{(p-j)},\overrightarrow{a_y}^{(k-p)})$.\qed
\end{proof}

\begin{example}[Separating with the help of a conservative extension]\label{godel}
In Re\-mark~\ref{cleversearch} we have seen a fixed-point procedure that may be used now to show that G\"odel logics (introduced in Ex.~\ref{exTFlogics}) are not separable when they involve more than three truth-values.  
For instance, it is easy to see that in~G$_4$ there are precisely six different definable unary operations, and none of them distinguishes the undesignated values~$\frac{1}{3}$ and~$\frac{2}{3}$ from one another.
Thus, we here will directly follow one of the two strategies employed in the proof of Prop.~\ref{prop:genuine} and consider the conservative extension of each logic~G$_n$ obtained by the addition to~$\Sigma^b$ of the family of sentential constants $\Sigma_0=\{a_m\}_{0<m<n-1}$, to be interpreted by setting $\widehat{a_m}=\frac{m}{n-1}$. In the extended logic~G$^+_n$, we may now introduce a family of unary operators~$\Sigma_1=\{k_m\}_{0<m<n-1}$ interpreted to such an effect that 
$\widehat{k_m}=\lambda z.\left(\ifte{z=\frac{m}{n-1}}{1}{(\ifte{z>\frac{m}{n-1}}{\frac{m}{n-1}}{z})}\right)$. 
Such operators $k_m$ are easily definable, e.g., by $\lambda p.(a_m\equiv p)$ or $\lambda p.(p\equiv a_m)$.
Obviously, given that $t(\widehat{k_m}(z))=T$ iff $z=\frac{m}{n-1}$, these unary operators may be used to produce an appropriate separating sequence for~G$^+_n$: just define~$\overline{\theta}=\langle\theta_r\rangle_{r=0}^{n-2}$ by setting $\theta_r=k_{r}$ for $r>0$. The resulting binary print $\overline{\theta}(0)$ will consist only of $F$s, while $\overline{\theta}(1)$ will have exactly one $T$, in the first position. Each $\overline{\theta}\left(\frac{m}{n-1}\right)$, for $0<m<n-1$, will also have exactly one $T$, in position $m+1$.
\end{example}

As we shall see in Remark~\ref{primseps}, a suitable conservative extension may be useful even if the logic at hand is already separable from the start.

\subsection{Classic-like characterization of finite-valued logics}\label{sec:classiclike}

\noindent 
Assuming, henceforth, that we are dealing with a separable $n$-valued truth-func\-tion\-al logic $\mathcal{L}$ characterized by a semantics~$\mathsf{Sem}$, let us proceed toward providing a constructive description of its $\mathsf{S}$-reduction $\mathsf{Sem}_2$ produced by Def.~\ref{dyadic}. For that purpose we will adopt a classic metalanguage: we shall use $\&$ to represent conjunction, $\mid\mid$ to represent disjunction, $\Longrightarrow$ to represent implication, $\top$ to represent truth, and $\divideontimes$ to represent an absurd. We shall also consider \textsl{labeled formulas} of the form $\X{}\varphi$ where $X\in\{F,T\}$ and $\varphi$ is a formula of $\mathcal{L}$. When convenient, we shall write $X^c$ to denote the \textsl{conjugate} of~$X$, defined by setting $F^c=T$ and $T^c=F$. We shall say that a bivaluation $b$ \textsl{satisfies} $\X{}\varphi$ if $b(\varphi)=X$. Analogously, we shall say that an $n$-valued valuation~$\val$ \textsl{satisfies} a labeled formula $\X{}\varphi$
if the corresponding bivaluation~$b_\val$ does, that is, if $b_\val(\varphi)=t(\val(\varphi))=X$. The extension of both notions of satisfaction to statements of the classical metalanguage is straightforward (that is, we assume $\&$, $\mid\mid$, $\Longrightarrow$, $\top$ and~$\divideontimes$ have the expected boolean interpretations).

We will describe the bivalent non-truth-functional semantics $\mathsf{Sem}_2$ by taking advantage of the truth-value separation apparatus developed above. Let us assume that $\overline{\theta}=\langle\theta_r\rangle_{r=0}^s$ is a separating sequence for~$\mathcal{L}$. As we have seen, $\overline{\theta}$~associates a different binary print to each of the~$n$ truth-values in $\mathcal{V}_n$. We can use an appropriate meta-linguistic statement to capture the fact that, in a given situation, the value of a formula~$\varphi$ corresponds to a certain binary print $\overline{X}=\langle X_r\rangle_{r=0}^s$:
\begin{equation}
{\X{0}\varphi} {\;\;\&\;\;} {\X{1}\theta_1(\varphi)} {\;\;\&\;\;} \dots {\;\;\&\;\;} {\X{s}\theta_s(\varphi)}. \tag{${\textsl{V}(\varphi\,;\overline{X})}$}\label{eq:value}
\end{equation}

\noindent
In general, given binary prints $\overline{X}_1,\dots,\overline{X}_k$ and formulas $\varphi_1,\dots,\varphi_k$, we write:
\begin{equation}
\textsl{V}(\varphi_1\,;\overline{X}_1){\;\;\&\;\;}\dots{\;\;\&\;\;}\textsl{V}(\varphi_k\,;\overline{X}_k). 
\tag{${\textsl{V}(\varphi_1,\dots,\varphi_k\,;\overline{X}_1,\dots,\overline{X}_k)}$}
\label{eq:tuples}
\end{equation}

\noindent
Obviously, given $z\in \mathcal{V}_n$, the statement ${\textsl{V}(\varphi\,;\overline{\theta}(z))}$ will capture the fact that the value of $\varphi$ is precisely $z$.
This means in particular that we can characterize the $2^{s{+}1}-n$ sequences of $F$ and $T$ of length $s{+}1$ that are unobtainable. 
This fact can be captured, for each such sequence $\overline{X}\notin\overline{\theta}[\mathcal{V}_n]$, by the following meta-linguistic statement over an arbitrary~$\varphi\in\mathcal{S}$:
\begin{equation}
{\textsl{V}(\varphi\,;\overline{X})} {\quad\Longrightarrow\quad} \divideontimes. \tag{$\textsl{U}\overline{X}$}\label{eq:useless}
\end{equation}

Recall that for each connective
${\odot}\in\Sigma_k$ there is an associated operator
$\widehat{{\odot}}:(\mathcal{V}_n)^k\longrightarrow\mathcal{V}_n$ in the algebra of truth-values. Given $X\in\{F,T\}$ and a separating formula~$\theta_r$ with $0\leq r\leq s$, let $R_X^{\theta_r{\odot}}$ be the set $\{\overline{x}\in(\mathcal{V}_n)^k: t(\widehat{\theta_r}(\widehat{{\odot}}(\overline{x})))=X\}$, that is, the set of all tuples of values in $\mathcal{V}_n$ that the subformulas $\varphi_1,\dots,\varphi_k$ may be assigned in order to guarantee that the bivalent value of the composite formula $\theta_r({\odot}(\varphi_1,\dots,\varphi_k))$ is~$X$. Each such tuple $\overline{x}=\langle{x_1,\dots,x_k}\rangle\in R_X^{\theta_r{\odot}}$ is characterized by the statement ${\textsl{V}(\varphi_1,\dots,\varphi_k\,;\overline{\theta}(x_1),\dots,\overline{\theta}(x_k))}$.
%
%
Thus, the complete behavior of the formula $\theta_r({\odot}(\varphi_1,\dots,\varphi_k))$ is captured by meta-linguistic statements of the form:
\begin{equation}
{\X{}\theta_r({\odot}(\varphi_1,\dots,\varphi_k))} {\;\Longrightarrow\;} ({\mid\mid}_{\overline{x}\in R_X^{\theta_r{\odot}}}
{\textsl{V}(\varphi_1,\dots,\varphi_k\,;\overline{\theta}(x_1),\dots,\overline{\theta}(x_k))}). \tag{$\textsl{B}_X^{\theta_r{\odot}}$}\label{eq:behavior}
\end{equation}

\begin{remark}
\label{complementaryBHV}
It should be clear that $R_F^{\theta_r{\odot}}$ and $R_T^{\theta_r{\odot}}$ are such that $R_F^{\theta_r{\odot}}\cup R_T^{\theta_r{\odot}}=(\mathcal{V}_n)^k$ and 
$R_F^{\theta_r{\odot}}\cap R_T^{\theta_r{\odot}}=\varnothing$. 
Hence, the right-hand sides of $\textsl{B}_F^{\theta_r{\odot}}$ and $\textsl{B}_T^{\theta_r{\odot}}$ are complementary, taking into account the unobtainable binary prints, i.e., a bivaluation satisfying statements $\textsl{U}\overline{X}$ for all $\overline{X}\in(\{F,T\}^{s{+}1}\setminus\overline{\theta}[\mathcal{V}_n])$ will satisfy the right-hand side of $R_F^{\theta_r{\odot}}$ if and only if it does not satisfy the right-hand side of $R_T^{\theta_r{\odot}}$. This means also that the meta-linguistic implication in each $\textsl{B}_X^{\theta_r{\odot}}$ statement is actually an equivalence.
Note that it may occur that $R_X^{\theta_r{\odot}}=\varnothing$, for some $X\in\{F,T\}$, and thus $R_{X^c}^{\theta_r{\odot}}=(\mathcal{V}_n)^k$. 
That happens, for instance, when ${\odot}\in\Sigma_0$, or when $\widehat{{\odot}}[(\mathcal{V}_n)^k]\subseteq\Vd_{m,n}$ or $\widehat{{\odot}}[(\mathcal{V}_n)^k]\subseteq\Vu_{m,n}$, for ${\odot}\in\Sigma_k$ and $k>0$. 
In such circumstances, the right-hand side of one of the two $\textsl{B}_X^{\theta_r{\odot}}$ statements with $X\in\{F,T\}$ will be tautological, and the other will be absurd.
\end{remark}


\begin{definition}\label{bivstatements}
The set $\mathcal{B}(\mathcal{L},\overline{\theta})$ of \textsl{bivalent statements} associated to a given separable finite-valued logic $\mathcal{L}$, fixed a separating sequence $\overline{\theta}=\langle\theta_r\rangle_{r=0}^s$, is formed by all instances of:
\begin{itemize}
\item $\textsl{U}\overline{X}$, for each $\overline{X}\in(\{F,T\}^{s{+}1}\setminus\overline{\theta}[\mathcal{V}_n])$, and 
\hfill{($\textsl{U}$-statements)}
\item $\textsl{B}_X^{\theta_r{\odot}}$, for each $X\in\{F,T\}$, $0\leq r\leq s$ and ${\odot}\in\Sigma$.
\hfill{($\textsl{B}$-statements)}
\end{itemize}
\end{definition}

\begin{remark}
\label{DNF}
Note that all the bivalent statements employed to characterize a finite-valued logic have left-hand sides that are conjunctions of labeled formulas, and right-hand sides that are in disjunctive normal form.
\end{remark}


\noindent
The following result guarantees the adequacy of our bivalent characterization.

\begin{proposition}\label{oksem}
$\mathsf{Sem}_2$ is the set of all bivaluations that satisfy $\mathcal{B}(\mathcal{L},\overline{\theta})$.
\end{proposition}
\begin{proof}
First, observe that if $\val:\mathbb{S}\longrightarrow\mathbb{V}\in\mathsf{Sem}$ then~$b_\val$ (and~$\val$) satisfies the bivalent statements associated to $\mathcal{L}$ almost by construction.  Indeed, given $\varphi\in\mathcal{S}$ and an unobtainable binary sequence~$\overline{X}$, then 
of course $\overline{\theta}(\val(\varphi))\neq\overline{X}$.  This implies that $b_\val(\theta_r(\varphi))\neq X_r$ for some $0\leq r\leq s$, thus~$b_\val$ fails to satisfy ${\X{r}\theta_r(\varphi)}$ and by consequence it fails to satisfy the meta-linguistic conjunction on the left-hand side of $\textsl{U}\overline{X}$. Given $X\in\{F,T\}$, $0\leq r\leq s$ and ${\odot}\in\Sigma$, assume that $b_\val$ satisfies the left-hand side of $\textsl{B}_X^{\theta_r{\odot}}$, that is, assume~$b_\val$ satisfies $\X{}\theta_r({\odot}(\varphi_1,\dots,\varphi_k))$. Such assumption means that 
$b_\val(\theta_r({\odot}(\varphi_1,\dots,\varphi_k)))=t(\widehat{\theta_r}(\widehat{{\odot}}(\val(\varphi_1),\dots,\val(\varphi_k))))=X$ and therefore 
$\langle \val(\varphi_1),\dots,\val(\varphi_k)\rangle\in R_X^{\theta_r{\odot}}$. Thus, $b_\val$ satisfies $\textsl{V}(\varphi_1,\dots,\varphi_k\,;\overline{\theta}(\val(\varphi_1)),\ldots,\overline{\theta}(\val(\varphi_k)))$, consequently satisfying the disjunction on the right-hand side of $\textsl{B}_X^{\theta_r{\odot}}$.

Conversely, suppose that a bivaluation $b:\mathcal{S}\longrightarrow\mathcal{V}_2$ satisfies all the bivalent statements associated to $\mathcal{L}$. For each $\varphi\in\mathcal{S}$, due to the fact that $b$ satisfies all the statements $\textsl{U}\overline{X}$ for unobtainable $\overline{X}$, it is clear that the sequence $\overline{X}_\varphi=\langle b(\varphi),b(\theta_1(\varphi)),\dots,b(\theta_s(\varphi))\rangle$ must be obtainable. Thus, we can define an $n$-valuation $\val_b:\mathcal{S}\longrightarrow\mathcal{V}_n$ by setting $\val_b(\varphi)$ to be the unique truth-value in~$\mathcal{V}_n$ whose binary print $\overline{\theta}(\val_b(\varphi))$ is $\overline{X}_\varphi$. Clearly, $b_{\val_b}=t^{m,n}\circ\val_b=b$ and then we are just left with proving that $w_b$ is a homomorphism between the $\Sigma$-algebras~$\mathbb{S}$ and~$\mathbb{V}$. 
Let ${\odot}\in\Sigma$ be an arbitrary connective and $\varphi_1,\dots,\varphi_k\in\mathcal{S}$, and recall that $b$ satisfies the statements $\textsl{B}_X^{\theta_r{\odot}}$ for each $X\in\{F,T\}$ and $0\leq r\leq s$. For $\overline{x}=\langle w_b(\varphi_1),\dots,w_b(\varphi_k)\rangle$, it must be the case that either we have both $\overline{x}\in R_F^{\theta_r{\odot}}$ and $\overline{x}\notin R_T^{\theta_r{\odot}}$, or else we have both $\overline{x}\notin R_F^{\theta_r{\odot}}$ and $\overline{x}\in R_T^{\theta_r{\odot}}$. If $\overline{x}\in R_X^{\theta_r{\odot}}$ then $t(\widehat{\theta_r}(\widehat{{\odot}}(\overline{x})))=X$. Moreover, $b$ cannot satisfy the disjunction on the right-hand side of $\textsl{B}_{X^c}^{\theta_r{\odot}}$, 
and thus it also does not satisfy its left-hand side. Hence, it must be the case that $b$ satisfies the left-hand side of $\textsl{B}_{X}^{\theta_r{\odot}}$, that is, $b(\theta_r({\odot}(\varphi_1,\dots,\varphi_k)))=X$ as well. But this means that the binary print $\overline{\theta}(\widehat{{\odot}}(\overline{x}))$ coincides with $\overline{X}_{{\odot}(\varphi_1,\dots,\varphi_k)}$, and thus $\widehat{{\odot}}(\overline{x})$ is the unique value whose binary print is $\overline{X}_{{\odot}(\varphi_1,\dots,\varphi_k)}$. We conclude that $\val_b({\odot}(\varphi_1,\dots,\varphi_k))=\widehat{{\odot}}(\overline{x})=\widehat{{\odot}}(w_b(\varphi_1),\dots,w_b(\varphi_k))$, thus $w_b\in\mathsf{Sem}$.\qed
\end{proof}

\begin{example}[Bivalent characterization of \L$_3$]\label{lukas2}
Let us return to the example of \L$_3$, separated by $\overline{\theta}=\langle \mathsf{id},\theta\rangle$, where $\theta=\lambda p.(\neg p\supset p)$ is a possible definition of the unary operator $j^1_\geq$ (i.e., the separator~$\theta_1$ mentioned in Ex.~\ref{ex:separatingTF}, whose subscript we drop here). Note that the binary print $\langle T,F\rangle$ is unobtainable, whereas $\overline{\theta}(0)=\langle F,F\rangle$, $\overline{\theta}(\frac{1}{2})=\langle F,T\rangle$ and $\overline{\theta}(1)=\langle T,T\rangle$.
The bivalent statements in $\mathcal{B}(\textsl{\L$_3$},\langle p,\theta(p)\rangle)$ are shown in Table~\ref{biv:L3}.
\begin{table}[htbp]\small
	\begin{center}
\scalebox{.70}{
\begin{tabular}{|c|lcl|}
\hline
 ($\textsl{U}\langle T,F\rangle$) & $({\T\varphi} {\;\;\&\;\;} {\F\theta(\varphi)})$ & $\Longrightarrow$ & $\divideontimes$\\
\hline
 ($\textsl{B}_F^{\,\neg}$) & ${\F\neg\varphi}$ & $\Longrightarrow$ & $({\F\varphi} {\;\;\&\;\;} {\T\theta(\varphi)}) \mid\mid ({\T\varphi} {\;\;\&\;\;} {\T\theta(\varphi)})$ \\
\hline
 ($\textsl{B}_T^{\,\neg}$) & ${\T\neg\varphi}$ & $\Longrightarrow$ & $({\F\varphi} {\;\;\&\;\;} {\F\theta(\varphi)})$ \\
\hline
 ($\textsl{B}_F^{\,\theta\neg}$) & ${\F\theta(\neg\varphi)}$ & $\Longrightarrow$ & $({\T\varphi} {\;\;\&\;\;} {\T\theta(\varphi)})$ \\
\hline
 ($\textsl{B}_T^{\,\theta\neg}$) & ${\T\theta(\neg\varphi)}$ & $\Longrightarrow$ & $({\F\varphi} {\;\;\&\;\;} {\F\theta(\varphi)}) \mid\mid ({\F\varphi} {\;\;\&\;\;} {\T\theta(\varphi)})$ \\
\hline
 ($\textsl{B}_F^{\,\supset}$) & ${\F\varphi\supset\psi}$ & $\Longrightarrow$ & $({\F\varphi} {\;\;\&\;\;} {\T\theta(\varphi)} {\;\;\&\;\;} {\F\psi} {\;\;\&\;\;} {\F\theta(\psi)}) \mid\mid ({\T\varphi} {\;\;\&\;\;} {\T\theta(\varphi)} {\;\;\&\;\;} {\F\psi} {\;\;\&\;\;} {\F\theta(\psi)})\mid\mid$\\
  & & & $({\T\varphi} {\;\;\&\;\;} {\T\theta(\varphi)} {\;\;\&\;\;} {\F\psi} {\;\;\&\;\;} {\T\theta(\psi)})$\\
\hline
 ($\textsl{B}_T^{\,\supset}$) & ${\T\varphi\supset\psi}$ & $\Longrightarrow$ & $({\F\varphi} {\;\;\&\;\;} {\F\theta(\varphi)} {\;\;\&\;\;} {\F\psi} {\;\;\&\;\;} {\F\theta(\psi)}) \mid\mid ({\F\varphi} {\;\;\&\;\;} {\F\theta(\varphi)} {\;\;\&\;\;} {\F\psi} {\;\;\&\;\;} {\T\theta(\psi)})\mid\mid$\\
  & & & $({\F\varphi} {\;\;\&\;\;} {\F\theta(\varphi)} {\;\;\&\;\;} {\T\psi} {\;\;\&\;\;} {\T\theta(\psi)})\mid\mid ({\F\varphi} {\;\;\&\;\;} {\T\theta(\varphi)} {\;\;\&\;\;} {\F\psi} {\;\;\&\;\;} {\T\theta(\psi)})\mid\mid$\\
  & & & $({\F\varphi} {\;\;\&\;\;} {\T\theta(\varphi)} {\;\;\&\;\;} {\T\psi} {\;\;\&\;\;} {\T\theta(\psi)})\mid\mid ({\T\varphi} {\;\;\&\;\;} {\T\theta(\varphi)} {\;\;\&\;\;} {\T\psi} {\;\;\&\;\;} {\T\theta(\psi)})$\\
\hline
($\textsl{B}_F^{\,\theta\supset}$) & ${\F\theta(\varphi\supset\psi)}$ & $\Longrightarrow$ & $({\T\varphi} {\;\;\&\;\;} {\T\theta(\varphi)} {\;\;\&\;\;} {\F\psi} {\;\;\&\;\;} {\F\theta(\psi)})$ \\
\hline
($\textsl{B}_T^{\,\theta\supset}$) & ${\T\theta(\varphi\supset\psi)}$ & $\Longrightarrow$ & $({\F\varphi} {\;\;\&\;\;} {\F\theta(\varphi)} {\;\;\&\;\;} {\F\psi} {\;\;\&\;\;} {\F\theta(\psi)}) \mid\mid ({\F\varphi} {\;\;\&\;\;} {\F\theta(\varphi)} {\;\;\&\;\;} {\F\psi} {\;\;\&\;\;} {\T\theta(\psi)})\mid\mid$\\
  & & & $({\F\varphi} {\;\;\&\;\;} {\F\theta(\varphi)} {\;\;\&\;\;} {\T\psi} {\;\;\&\;\;} {\T\theta(\psi)})\mid\mid ({\F\varphi} {\;\;\&\;\;} {\T\theta(\varphi)} {\;\;\&\;\;} {\F\psi} {\;\;\&\;\;} {\F\theta(\psi)})\mid\mid$\\
  & & & $({\F\varphi} {\;\;\&\;\;} {\T\theta(\varphi)} {\;\;\&\;\;} {\F\psi} {\;\;\&\;\;} {\T\theta(\psi)})\mid\mid ({\F\varphi} {\;\;\&\;\;} {\T\theta(\varphi)} {\;\;\&\;\;} {\T\psi} {\;\;\&\;\;} {\T\theta(\psi)})\mid\mid$\\
  & & & $({\T\varphi} {\;\;\&\;\;} {\T\theta(\varphi)} {\;\;\&\;\;} {\F\psi} {\;\;\&\;\;} {\T\theta(\psi)})\mid\mid ({\T\varphi} {\;\;\&\;\;} {\T\theta(\varphi)} {\;\;\&\;\;} {\T\psi} {\;\;\&\;\;} {\T\theta(\psi)})$\\
\hline
($\textsl{B}_F^{\,\lor}$) & ${\F\varphi\lor\psi}$ & $\Longrightarrow$ & $({\F\varphi} {\;\;\&\;\;} {\F\theta(\varphi)} {\;\;\&\;\;} {\F\psi} {\;\;\&\;\;} {\F\theta(\psi)})\mid\mid ({\F\varphi} {\;\;\&\;\;} {\F\theta(\varphi)} {\;\;\&\;\;} {\F\psi} {\;\;\&\;\;} {\T\theta(\psi)})\mid\mid$\\
  & & & $({\F\varphi} {\;\;\&\;\;} {\T\theta(\varphi)} {\;\;\&\;\;} {\F\psi} {\;\;\&\;\;} {\F\theta(\psi)})\mid\mid ({\F\varphi} {\;\;\&\;\;} {\T\theta(\varphi)} {\;\;\&\;\;} {\F\psi} {\;\;\&\;\;} {\T\theta(\psi)})$\\
\hline
($\textsl{B}_T^{\,\lor}$) & ${\T\varphi\lor\psi}$ & $\Longrightarrow$ & $({\F\varphi} {\;\;\&\;\;} {\F\theta(\varphi)} {\;\;\&\;\;} {\T\psi} {\;\;\&\;\;} {\T\theta(\psi)})\mid\mid ({\F\varphi} {\;\;\&\;\;} {\T\theta(\varphi)} {\;\;\&\;\;} {\T\psi} {\;\;\&\;\;} {\T\theta(\psi)})\mid\mid$\\
  & & & $({\T\varphi} {\;\;\&\;\;} {\T\theta(\varphi)} {\;\;\&\;\;} {\F\psi} {\;\;\&\;\;} {\F\theta(\psi)})\mid\mid ({\T\varphi} {\;\;\&\;\;} {\T\theta(\varphi)} {\;\;\&\;\;} {\F\psi} {\;\;\&\;\;} {\T\theta(\psi)})\mid\mid$\\
  & & & $({\T\varphi} {\;\;\&\;\;} {\T\theta(\varphi)} {\;\;\&\;\;} {\T\psi} {\;\;\&\;\;} {\T\theta(\psi)})$\\
\hline
($\textsl{B}_F^{\,\theta\lor}$) & ${\F\theta(\varphi\lor\psi)}$ & $\Longrightarrow$ & $({\F\varphi} {\;\;\&\;\;} {\F\theta(\varphi)} {\;\;\&\;\;} {\F\psi} {\;\;\&\;\;} {\F\theta(\psi)})$\\
\hline
($\textsl{B}_T^{\,\theta\lor}$) & ${\T\theta(\varphi\lor\psi)}$ & $\Longrightarrow$ & $({\F\varphi} {\;\;\&\;\;} {\F\theta(\varphi)} {\;\;\&\;\;} {\F\psi} {\;\;\&\;\;} {\T\theta(\psi)}) \mid\mid ({\F\varphi} {\;\;\&\;\;} {\F\theta(\varphi)} {\;\;\&\;\;} {\T\psi} {\;\;\&\;\;} {\T\theta(\psi)})\mid\mid$\\
  & & & $({\F\varphi} {\;\;\&\;\;} {\T\theta(\varphi)} {\;\;\&\;\;} {\F\psi} {\;\;\&\;\;} {\F\theta(\psi)})\mid\mid ({\F\varphi} {\;\;\&\;\;} {\T\theta(\varphi)} {\;\;\&\;\;} {\F\psi} {\;\;\&\;\;} {\T\theta(\psi)})\mid\mid$\\
  & & & $({\F\varphi} {\;\;\&\;\;} {\T\theta(\varphi)} {\;\;\&\;\;} {\T\psi} {\;\;\&\;\;} {\T\theta(\psi)})\mid\mid ({\T\varphi} {\;\;\&\;\;} {\T\theta(\varphi)} {\;\;\&\;\;} {\F\psi} {\;\;\&\;\;} {\F\theta(\psi)})\mid\mid$\\
  & & & $({\T\varphi} {\;\;\&\;\;} {\T\theta(\varphi)} {\;\;\&\;\;} {\F\psi} {\;\;\&\;\;} {\T\theta(\psi)})\mid\mid ({\T\varphi} {\;\;\&\;\;} {\T\theta(\varphi)} {\;\;\&\;\;} {\T\psi} {\;\;\&\;\;} {\T\theta(\psi)})$\\
\hline
($\textsl{B}_F^{\,\land}$) & ${\F\varphi\land\psi}$ & $\Longrightarrow$ & $({\F\varphi} {\;\;\&\;\;} {\F\theta(\varphi)} {\;\;\&\;\;} {\F\psi} {\;\;\&\;\;} {\F\theta(\psi)}) \mid\mid ({\F\varphi} {\;\;\&\;\;} {\F\theta(\varphi)} {\;\;\&\;\;} {\F\psi} {\;\;\&\;\;} {\T\theta(\psi)})\mid\mid$\\
  & & & $({\F\varphi} {\;\;\&\;\;} {\F\theta(\varphi)} {\;\;\&\;\;} {\T\psi} {\;\;\&\;\;} {\T\theta(\psi)})\mid\mid ({\F\varphi} {\;\;\&\;\;} {\T\theta(\varphi)} {\;\;\&\;\;} {\F\psi} {\;\;\&\;\;} {\F\theta(\psi)})\mid\mid$\\
  & & & $({\F\varphi} {\;\;\&\;\;} {\T\theta(\varphi)} {\;\;\&\;\;} {\F\psi} {\;\;\&\;\;} {\T\theta(\psi)})\mid\mid ({\F\varphi} {\;\;\&\;\;} {\T\theta(\varphi)} {\;\;\&\;\;} {\T\psi} {\;\;\&\;\;} {\T\theta(\psi)})\mid\mid$\\
  & & & $({\T\varphi} {\;\;\&\;\;} {\T\theta(\varphi)} {\;\;\&\;\;} {\F\psi} {\;\;\&\;\;} {\F\theta(\psi)})\mid\mid ({\T\varphi} {\;\;\&\;\;} {\T\theta(\varphi)} {\;\;\&\;\;} {\F\psi} {\;\;\&\;\;} {\T\theta(\psi)})$\\
\hline
($\textsl{B}_T^{\,\land}$) & ${\T\varphi\land\psi}$ & $\Longrightarrow$ & $({\T\varphi} {\;\;\&\;\;} {\T\theta(\varphi)} {\;\;\&\;\;} {\T\psi} {\;\;\&\;\;} {\T\theta(\psi)})$\\
\hline
($\textsl{B}_F^{\,\theta\land}$) & ${\F\theta(\varphi\land\psi)}$ & $\Longrightarrow$ & $({\F\varphi} {\;\;\&\;\;} {\F\theta(\varphi)} {\;\;\&\;\;} {\F\psi} {\;\;\&\;\;} {\F\theta(\psi)})\mid\mid ({\F\varphi} {\;\;\&\;\;} {\F\theta(\varphi)} {\;\;\&\;\;} {\F\psi} {\;\;\&\;\;} {\T\theta(\psi)})\mid\mid$\\
  & & & $({\F\varphi} {\;\;\&\;\;} {\F\theta(\varphi)} {\;\;\&\;\;} {\T\psi} {\;\;\&\;\;} {\T\theta(\psi)})\mid\mid ({\F\varphi} {\;\;\&\;\;} {\T\theta(\varphi)} {\;\;\&\;\;} {\F\psi} {\;\;\&\;\;} {\F\theta(\psi)})\mid\mid$\\
  & & & $({\T\varphi} {\;\;\&\;\;} {\T\theta(\varphi)} {\;\;\&\;\;} {\F\psi} {\;\;\&\;\;} {\F\theta(\psi)})$\\
\hline
($\textsl{B}_T^{\,\theta\land}$) & ${\T\theta(\varphi\land\psi)}$ & $\Longrightarrow$ & $({\F\varphi} {\;\;\&\;\;} {\T\theta(\varphi)} {\;\;\&\;\;} {\F\psi} {\;\;\&\;\;} {\T\theta(\psi)})\mid\mid ({\F\varphi} {\;\;\&\;\;} {\T\theta(\varphi)} {\;\;\&\;\;} {\T\psi} {\;\;\&\;\;} {\T\theta(\psi)})\mid\mid$\\
  & & & $({\T\varphi} {\;\;\&\;\;} {\T\theta(\varphi)} {\;\;\&\;\;} {\F\psi} {\;\;\&\;\;} {\T\theta(\psi)})\mid\mid ({\T\varphi} {\;\;\&\;\;} {\T\theta(\varphi)} {\;\;\&\;\;} {\T\psi} {\;\;\&\;\;} {\T\theta(\psi)})$\\
\hline
\end{tabular}
}
\end{center}
\caption{$\mathcal{B}(\textsl{\L$_3$},\langle \mathsf{id},\lambda p.(\neg p\supset p)\rangle)$: the bivalent characterization of \L$_3$ separated by $\langle \mathsf{id},\theta\rangle$.}
\label{biv:L3}
\end{table}
\end{example}
\begin{remark}\label{simplification}
The set of bivalent statements associated to a logic can often be simplified, without any danger of spoiling the result of Prop.~\ref{oksem}, nor any of the subsequent results. 
To start with, it may happen that a statement is simply tautological, as already explained in Remark~\ref{complementaryBHV}, in which case it can be simply omitted. 
The example of \L$_3$ above does not contain statements of that kind, but several such statements appear in connection with G$_4$ (see Ex.~\ref{systems-and-tableaux} below).
Still, even a nontautological statement can often be substantially shortened.
Consider for instance $\textsl{B}_F^{\,\supset}$ from the example above, namely, 

\begin{center}
\scalebox{.9}{
\begin{tabular}{c}
${\F\varphi\supset\psi}\Longrightarrow({\F\varphi} {\;\;\&\;\;} {\T\theta(\varphi)} {\;\;\&\;\;} {\F\psi} {\;\;\&\;\;} {\F\theta(\psi)}) \mid\mid$\\ 
$\hspace*{2.3cm}({\T\varphi} {\;\;\&\;\;} {\T\theta(\varphi)} {\;\;\&\;\;} {\F\psi} {\;\;\&\;\;} {\F\theta(\psi)})\mid\mid$\\
$\hspace*{2.1cm}({\T\varphi} {\;\;\&\;\;} {\T\theta(\varphi)} {\;\;\&\;\;} {\F\psi} {\;\;\&\;\;} {\T\theta(\psi)}).$
\end{tabular}
}
\end{center}

\noindent
Clearly, the first two disjuncts on the right-hand side are classically equivalent to 
$({\T\theta(\varphi)} {\;\;\&\;\;} {\F\psi} {\;\;\&\;\;} {\F\theta(\psi)})$, as either $\F\varphi$ or $\T\varphi$ must be satisfied by any given bivaluation. Similarly, the last two disjuncts are equivalent to $({\T\varphi} {\;\;\&\;\;}{\T\theta(\varphi)} {\;\;\&\;\;} {\F\psi})$. 
Each of these new expressions can be further simplified by taking into account the statement
$\textsl{U}\langle T,F\rangle$. As the binary print $\langle T,F\rangle$ is unobtainable, we thus conclude, in the former case, that ${\F\theta(\psi)}$ must imply ${\F\psi}$. Analogously, in the latter case, we conclude that ${\T\varphi}$ must imply ${\T\theta(\varphi)}$. Thus, $\textsl{B}_F^{\,\supset}$ may be equivalently stated as

\begin{center}
\scalebox{.9}{
\begin{tabular}{c}
	${\F\varphi\supset\psi}\Longrightarrow({\T\varphi} {\;\;\&\;\;} {\F\psi}) \mid\mid
	({\T\theta(\varphi)} {\;\;\&\;\;} {\F\theta(\psi)}).$
\end{tabular}
}
\end{center}


Such a simplification strategy may be applied, using boolean reasoning and the unobtainable binary prints as premises, to reach a streamlined version of the right-hand side of each statement. The only general restriction that we must impose is that the right-hand sides of our statements remain in disjunctive normal form and use only labeled formulas already occurring on the right-hand side of the original statement. Note that this streamlining procedure can be systematized by means of Karnaugh maps, and even automated by using the Quine-McCluskey algorithm, or the Espresso heuristic~\cite{Brayton:1984}.

Notation-wise, we will not distinguish a statement from a convenient simplification. 
Note, at any rate, that none of the results in this paper depend on (or is affected by) performing such a simplification.
\end{remark}

\begin{example}[A streamlined bivalent characterization of \L$_3$]\label{lukas3}
Sim\-pli\-fy\-ing the bivalent statements from Ex.~\ref{lukas2} we may obtain the equivalent list of statements, also dubbed $\mathcal{B}(\textsl{\L$_3$},\langle p,\theta(p)\rangle)$, shown in Table~\ref{stream:L3}.
\begin{table}[htbp]\small
\begin{center}
\scalebox{.7}{
\begin{tabular}{|r|lcl|}
\hline
 ($\textsl{U}\langle T,F\rangle$) & $({\T\varphi} {\;\;\&\;\;} {\F\theta(\varphi)})$ & $\Longrightarrow$ & $\divideontimes$\\
\hline
 ($\textsl{B}_F^{\,\neg}$)& ${\F\neg\varphi}$ & $\Longrightarrow$ & ${\T\theta(\varphi)}$ \\
\hline
 ($\textsl{B}_T^{\,\neg}$) & ${\T\neg\varphi}$ & $\Longrightarrow$ & ${\F\theta(\varphi)}$ \\
\hline
 ($\textsl{B}_F^{\,\theta\neg}$) & ${\F\theta(\neg\varphi)}$ & $\Longrightarrow$ & ${\T\varphi}$ \\
\hline
 ($\textsl{B}_T^{\,\theta\neg}$) & ${\T\theta(\neg\varphi)}$ & $\Longrightarrow$ & ${\F\varphi}$ \\
\hline
 ($\textsl{B}_F^{\,\supset}$) & ${\F\varphi\supset\psi}$ & $\Longrightarrow$ & $({\T\varphi} {\;\;\&\;\;} {\F\psi}) \mid\mid ({\T\theta(\varphi)} {\;\;\&\;\;} {\F\theta(\psi)})$\\
\hline
 ($\textsl{B}_T^{\,\supset}$) & ${\T\varphi\supset\psi}$ & $\Longrightarrow$ & $({\F\varphi} {\;\;\&\;\;} {\T\theta(\psi)}) \mid\mid {\F\theta(\varphi)} \mid\mid {\T\psi}$\\
\hline
($\textsl{B}_F^{\,\theta\supset}$) & ${\F\theta(\varphi\supset\psi)}$ & $\Longrightarrow$ & $({\T\varphi} {\;\;\&\;\;} {\F\theta(\psi)})$ \\
\hline
($\textsl{B}_T^{\,\theta\supset}$) & ${\T\theta(\varphi\supset\psi)}$ & $\Longrightarrow$ & ${\F\varphi} \mid\mid {\T\theta(\psi)}$\\
\hline
($\textsl{B}_F^{\,\lor}$) & ${\F\varphi\lor\psi}$ & $\Longrightarrow$ & $({\F\varphi} {\;\;\&\;\;}  {\F\psi})$\\
\hline
($\textsl{B}_T^{\,\lor}$) & ${\T\varphi\lor\psi}$ & $\Longrightarrow$ & ${\T\varphi}\mid\mid {\T\psi}$\\
\hline
($\textsl{B}_F^{\,\theta\lor}$) & ${\F\theta(\varphi\lor\psi)}$ & $\Longrightarrow$ & $({\F\theta(\varphi)} {\;\;\&\;\;} {\F\theta(\psi)})$\\
\hline
($\textsl{B}_T^{\,\theta\lor}$) & ${\T\theta(\varphi\lor\psi)}$ & $\Longrightarrow$ & ${\T\theta(\varphi)}\mid\mid {\T\theta(\psi)}$\\
\hline
 ($\textsl{B}_F^{\,\land}$) & ${\F\varphi\land\psi}$ & $\Longrightarrow$ & ${\F\varphi} \mid\mid {\F\psi}$\\
\hline
($\textsl{B}_T^{\,\land}$) & ${\T\varphi\land\psi}$ & $\Longrightarrow$ & $({\T\varphi} {\;\;\&\;\;} {\T\psi})$\\
\hline
($\textsl{B}_F^{\,\theta\land}$) & ${\F\theta(\varphi\land\psi)}$ & $\Longrightarrow$ & ${\F\theta(\varphi)}\mid\mid {\F\theta(\psi)}$\\
\hline
($\textsl{B}_T^{\,\theta\land}$) & ${\T\theta(\varphi\land\psi)}$ & $\Longrightarrow$ & $({\T\theta(\varphi)} {\;\;\&\;\;} {\T\theta(\psi)})$\\
\hline
\end{tabular}
}
\end{center}
\caption{Streamlined $\mathcal{B}(\textsl{\L$_3$},\langle \mathsf{id},\lambda p.(\neg p\supset p)\rangle)$.}
\label{stream:L3}
\end{table}
\end{example}

\subsection{Compositionality generalized}\label{sec:effectiveness}

\noindent
Separators play a crucial role in our development. Now, while the original truth-functional semantics of our finite-valued logics was based, as it has been already mentioned, on a straightforward notion of `compositionality', one might contend that the bivalent semantics which we can associate to the same logics are in fact based on a \textit{generalized notion of compositionality} according to which the value of a formula is to be (uniquely) determined from the values of separators applied to its immediate subformulas. In order to burst life into this idea we must first understand how to adequately explore the structure of formulas.

Let's start by upgrading some terminology from Section~\ref{sec:fvl} to take the separators from $\overline{\theta}=\langle\theta_r\rangle_{r=0}^s$ into account.  We call~$\varphi\in\mathcal{S}$ a \textsl{basic} formula if $\varphi=\theta_r(\psi)$ for some noncomposite formula~$\psi$ and some $0\leq r\leq s$; other formulas are called \textsl{nonbasic}.  To be sure, basic formulas are precisely those that may be obtained by applying a separator to either an atomic variable or a sentential constant; note in particular that atomic variables are indeed basic formulas, given our convention to set $\theta_0=\mathsf{id}$. 
Given $0\leq r\leq s$ and a connective ${\odot}\in\Sigma_k$ with $k\neq 0$, recall from Section~\ref{sec:fvl} that by $\mathcal{S}(\theta_r({\odot}(q_1,\dots,q_k)))$ we denoted the set of all instances of the statement-form $\theta_r({\odot}(q_1,\dots,q_k))$ --- here we will denote this more simply as $\mathcal{S}(\theta_r{\odot})$. Instead of $\mathcal{S}(\theta_0{\odot})$ we will also simply write $\mathcal{S}({\odot})$, given that $\theta_0(p)=p$.
Clearly, the family $\{\mathcal{S}({\odot})\}_{{\odot}\in\Sigma\setminus\Sigma_0}$ constitutes a partition of the set of composite formulas.

Let $\varphi$ be a nonbasic formula. 
Given $0\leq r\leq s$, whenever $\varphi\in\mathcal{S}(\theta_r{\odot})$ for some ${\odot}\in\Sigma$ we shall say that
$\theta_r{\odot}$ is a \textsl{fit for}~$\varphi$. 
As it happens, there may be $r_1\neq r_2$ and ${\odot}_1\neq{\odot}_2$ such that both $\theta_{r_1}{\odot}$ and $\theta_{r_2}{\odot}$ are fit for~$\varphi$.  In general:

\begin{lemma}\label{lemma:intersect}
Let $0\leq r_1,r_2\leq s$ and ${\odot}_1,{\odot}_2\in\Sigma$. In case $\mathcal{S}(\theta_{r_1}{\odot}_1)\cap\mathcal{S}(\theta_{r_2}{\odot}_2)\neq\varnothing$ then
exactly one of the following three situations must occur:
\begin{enumerate}[(1)]
\item $r_1=r_2$ and ${\odot}_1={\odot}_2$; or


\item $\mathcal{S}(\theta_{r_1}{\odot}_1)\subsetneq\mathcal{S}(\theta_{r_2}{\odot}_2)$ (or the other way around,  $\mathcal{S}(\theta_{r_2}{\odot}_2)\subsetneq\mathcal{S}(\theta_{r_1}{\odot}_1)$); or

\item $\mathcal{S}(\theta_{r_1}{\odot}_1)\cap\mathcal{S}(\theta_{r_2}{\odot}_2)$ is a singleton set, whose sole formula, \\
dubbed $\iota(\theta_{r_1}{\odot}_1,\theta_{r_2}{\odot}_2)$, is a ground formula.
\end{enumerate}
\end{lemma}
\begin{proof}
We first show that if $\mathcal{S}(\theta_{r_1})\cap\mathcal{S}(\theta_{r_2})\neq\varnothing$ then exactly one of the following three situations must occur:
\begin{enumerate}[(a)]
\item $r_1=r_2$; or
\item $\mathcal{S}(\theta_{r_1})\subsetneq\mathcal{S}(\theta_{r_2})$ (or the other way around, $\mathcal{S}(\theta_{r_2})\subsetneq\mathcal{S}(\theta_{r_1})$); or
\item $\mathcal{S}(\theta_{r_1})\cap\mathcal{S}(\theta_{r_2})$ is a singleton set, whose sole formula is a ground formula.
\end{enumerate}
Let $r_1\neq r_2$. We use Robinson's unification algorithm~\cite{rob:resolution:65} on the pair formed by $\theta_{r_1}(p)$ and $\theta_{r_2}(q)$ with $p,q\in\mathcal{A}$, $p\neq q$. Clearly, the pair is unifiable, and the algorithm outputs a most general unifier that will convey either a substitution of~$p$ by some $\delta(q)\in\mathcal{S}$ (or the other way around, a substitution of $q$ by some $\delta(p)\in\mathcal{S}$), or else a substitution of $p$ and $q$ by some ground formulas $\alpha_p,\alpha_q\in\mathcal{S}$. 
In the latter case, we can conclude that $\theta_{r_1}(p)$ and $\theta_{r_2}(q)$ have exactly one common instance $\theta_{r_1}(\alpha_p)=\theta_{r_2}(\alpha_q)$, which is ground.
In the former cases, assuming without loss of generality that the most general unifier conveys the substitution of $q$ by $\delta(p)$, we then have that 
$\theta_{r_1}(p)=\theta_{r_2}(\delta(p))$ and thus $\mathcal{S}(\theta_{r_1})\subseteq\mathcal{S}(\theta_{r_2})$. The inclusion is proper as $\delta(p)\neq p$, or else we would have $r_1=r_2$.

Our main result follows easily, now. Let $\varphi\in\mathcal{S}(\theta_{r_1}{\odot}_1)\cap\mathcal{S}(\theta_{r_2}{\odot}_2)$. Clearly, one also has $\varphi\in\mathcal{S}(\theta_{r_1})\cap\mathcal{S}(\theta_{r_2})\neq\varnothing$. If (a) is the case then $r_1=r_2$ and it is immediate that also ${\odot}_1={\odot}_2$, and we are in situation~(1). If $r_1\neq r_2$, then either (b) or (c) must be the case. If (c) is the case then $\varphi$ must be the sole formula in the intersection, and is therefore a ground formula, so we are in situation~(3). Otherwise, (b) must be the case, and $\mathcal{S}(\theta_{r_1})\subsetneq\mathcal{S}(\theta_{r_2})$. Thus, we have $\theta_{r_1}(p)=\theta_{r_2}(\delta(p))$, where $\delta(p)=\otimes(\delta_1(p),\dots,\delta_k(p))$ for some connective $\otimes\in\Sigma_k$ with $k\neq 0$ and one-variable formulas $\delta_1,\dots,\delta_k\in\mathcal{S}$.
Hence, we have  $\varphi=\theta_{r_1}({\odot}_1(\varphi_1,\dots,\varphi_{k_1}))=
\theta_{r_2}(\otimes(\delta_1({\odot}_1(\varphi_1,\dots,\varphi_{k_1})),\dots,\delta_k({\odot}_1(\varphi_1,\dots,\varphi_{k_1}))))=
\theta_{r_2}({\odot}_2(\psi_1,\dots,\psi_{k_2}))$. In particular, this implies that $\otimes={\odot}_2$ and $k=k_2$. Thus, $\theta_{r_1}({\odot}_1(q_1,\dots,q_{k_1}))=
\theta_{r_2}({\odot}_2(\delta_1({\odot}_1(q_1,\dots,q_{k_1})),\dots,\delta_k({\odot}_1(q_1,\dots,q_{k_1}))))$ and therefore $\mathcal{S}(\theta_{r_1}{\odot}_1)\subsetneq\mathcal{S}(\theta_{r_2}{\odot}_2)$, so we are in situation~(2).
\qed	
\end{proof}

Suppose that for a nonbasic formula $\varphi\in\mathcal{S}$ we have $\varphi=\theta_{r_1}({\odot}_1(\varphi_1,\dots,\varphi_{k_1}))$ and $\varphi=\theta_{r_2}({\odot}_2(\psi_1,\dots,\psi_{k_2}))$ with $r_1\neq r_2$, that is, suppose that
$\theta_{r_1}{\odot}$ and $\theta_{r_2}{\odot}$ are both fit for~$\varphi$ despite the fact that~$\theta_{r_1}$ are~$\theta_{r_2}$ are distinct separators.
Then, $\varphi\in\mathcal{S}(\theta_{r_1}{\odot}_1)\cap\mathcal{S}(\theta_{r_2}{\odot}_2)\neq\varnothing$ and we can examine the situation in the 
light of Lemma~\ref{lemma:intersect}. 
If $\mathcal{S}(\theta_{r_1}{\odot}_1)\cap\mathcal{S}(\theta_{r_2}{\odot}_2)$ is a singleton then $\varphi=\iota(\theta_{r_1}{\odot}_1,\theta_{r_2}{\odot}_2)$, and we dub $\varphi$ an \textsl{intersection formula}. 
Otherwise, without loss of generality, we have that $\mathcal{S}(\theta_{r_1}{\odot}_1)\subsetneq\mathcal{S}(\theta_{r_2}{\odot}_2)$.
In this case, we say that ${\theta}_{r_1}{\odot}_1$ is \textsl{more concrete} than ${\theta}_{r_2}{\odot}_2$, or that ${\theta}_{r_2}{\odot}_2$ is \textsl{more general} than ${\theta}_{r_1}{\odot}_1$. 
\begin{example}[Formula fitting and intersections]
\label{synt-coincidence-1}
Let us return to the case of \L$_3$, in the streamlined form given in Ex.~\ref{lukas3}, Table~\ref{stream:L3}. Recall that we proposed for this logic the separating sequence $\overline{\theta}=\langle \mathsf{id},\theta\rangle$ with $\theta=\lambda p.(\neg p\supset p)$. It is easy to see that both $\theta\land$ and $\supset$ are fit for a formula of the form $\theta(\varphi\land\psi)=\neg (\varphi\land\psi)\supset (\varphi\land\psi)$, as~$\theta$~itself has~$\supset$ as head connective. Clearly, $\theta\land$ is more concrete and~$\supset$ is more general. Of course, a formula of the form $\varphi\lor\psi$ has a unique fit~$\lor$.

Consider now G$^+_4$ as introduced in Ex.~\ref{godel}. Recall that in Ex.~\ref{ex:partialunobtainable} we used the separating sequence $\overline{\theta}=\langle \mathsf{id},\theta_1,\theta_2\rangle$ with $\theta_1=\lambda p.(a_1\equiv p)$ and $\theta_2=\lambda p.(p\equiv a_2)$. The asymmetric definitions of~$\theta_1$ and~$\theta_2$ were not without a purpose, as they allow us to illustrate at this point the existence of an intersection formula, namely $\iota(\theta_1 a_2,\theta_2 a_1)=(a_1\equiv a_2)$.
\end{example}


It is easy to see that the concreteness/generality order allows us to define the \textsl{most concrete} combination of separator and connective that fits any given composite formula that is not an intersection formula: for each such formula $\varphi$, the set of all fits for $\varphi$ (obviously non-empty and finite) is totally ordered by the concreteness/generality order, and therefore a minimum and a maximum exist. In particular, if $\theta_r{\odot}$ is the minimum (most concrete) fit for~$\varphi$ then we say that~$\varphi$ is a \textsl{proper $\theta_r{\odot}$-formula}.

Given a nonbasic proper $\theta_r{\odot}$-formula $\varphi=\theta_r{\odot}(\varphi_1,\ldots,\varphi_k)$, where ${\odot}\in\Sigma_k$, we call each $\theta_t\varphi_1,\ldots,\theta_t\varphi_k\in\mathcal{S}$, with $0\leq t\leq r$, a \textsl{generalized immediate subformula} of~$\varphi$.  The set $\mathsf{gsbf}(\varphi)$ of \textsl{generalized subformulas of~$\varphi$} is obtained by closing~$\{\varphi\}$ under generalized immediate subformulas, and the \textsl{proper generalized subformulas of~$\varphi$}, $\mathsf{pgsbf}(\varphi)$, are the elements of $\mathsf{gsbf}(\varphi)\setminus\{\varphi\}$.
The generalized notion of compositionality that will be presupposed in what follows demands a measure of formula complexity that is more fine-grained than the canonical measure given by \textsf{dpth}, and that takes into account both proper $\theta_r{\odot}$-formulas and generalized subformulas.

\begin{definition}\label{cplx}
The \textsl{generalized notion of formula complexity} is given by the function $\mathsf{cplx}:\mathcal{S}\longrightarrow\mathbb{N}$ such that:
$$\mathsf{cplx}(\varphi)=\left\{\begin{array}{ll}
				0 & \textrm{if }\varphi\textrm{ is basic or an intersection formula}\\[2mm]
				1+\mathop{\mathsf{Max}}\limits_{0\leq t\leq s,1\leq i\leq k}\!\!\!\mathsf{cplx}(\theta_t(\varphi_i))\hspace{-2mm} & \textrm{if $\varphi=\theta_r({\odot}(\varphi_1,\dots,\varphi_k))$ is a }\\[-1mm]
				& \textrm{proper $\theta_r{\odot}$-formula, }\\ 
				& \textrm{\quad for }{\odot}\in\Sigma_k\textrm{, $k\neq 0$, } 0\leq r\leq s,\\
				& \textrm{\quad and }\varphi_1,\dots,\varphi_k\in\mathcal{S}\\
			\end{array}\right.
$$			
\end{definition}

Note that this complexity function is well-defined precisely because $\mathsf{cplx}(\varphi)$ is completely determined from the values of $\mathsf{cplx}(\theta_r\psi)$ with $0\leq r\leq s$ and $\psi$ is an immediate subformula of $\varphi$, and not only we have finitely many separators
but also the number of immediate subformulas of $\varphi$ is finite.
With respect to this generalized notion of complexity, formulas of complexity~$0$ will be called \textsl{simple}, and formulas of positive complexity will be called \textsl{analyzable}.  
While many usual deductive formalisms capitalize on the so-called `subformula property', based on the truth-functional principle according to which the behavior of a formula is to be uniquely determined 
from the values of its less complex components, the same idea will later be explored in terms of the generalized notion of formula complexity and its associated generalized subformula property that takes separators into account as prefixes that help in internalizing `without additional cost' important semantic information at the syntactical level.

\begin{example}[Generalized complexity]
\label{synt-coincidence-2}
Returning to the example of \L$_3$ separated by $\overline{\theta}=\langle \mathsf{id},\theta\rangle$ with $\theta=\lambda p.(\neg p\supset p)$ from Ex.~\ref{synt-coincidence-1}, we see that formulas like~$p$ or~$\theta(p)$ are simple. Consider now $\theta(\varphi\land\psi)$ and $\varphi\land\psi$. 
Despite the fact that $\mathsf{dpth}(\theta(\varphi\land\psi)) >  \mathsf{dpth}(\varphi\land\psi)$, both formulas have the same complexity   $\mathsf{cplx}(\theta(\varphi\land\psi))=\mathsf{cplx}(\varphi\land\psi)=1+\mathsf{Max}(\mathsf{cplx}(\varphi),\mathsf{cplx}(\theta(\varphi)),\mathsf{cplx}(\psi),\mathsf{cplx}(\theta(\psi)))$.
\end{example}

\begin{remark}\label{primseps}
Note that if~$\varphi$ is a proper ${\odot}$-formula then there can be no separator~$\theta_r$ with $r\neq 0$ that is a fit for~$\varphi$. 
Any other syntactic coincidences besides those that happen between $\theta_r$ with $r\neq 0$ and the head connective of~$\theta_r$ may in fact be considered a nuisance to the purity of our method. This includes the case when~$\theta_{r_1}$ and~$\theta_{r_2}$ have common instances with $r_1\neq r_2$ and $r_1,r_2\neq 0$, and in particular intersection formulas. However, we must mention that there is a simple way of avoiding such intricacies. The trick (cf.~\cite{mar:V2V}) is to require all separators to be primitive unary connectives of the logic.  When that is not originally the case, one may simply work with a suitable conservative extension of the given logic.  In the case of such an extension, there will obviously be a single fit for each formula~$\varphi$, what simplifies somewhat the calculation of $\mathsf{cplx}$.
\end{remark}

As we shall see later on, yet another issue that must be taken into account for the development of successful tableaux is the possibility of matching the same signed formula to the premises of distinct (therefore applicable) rules. The most usual form in which this problem presents itself will be dealt with in Def.~\ref{analytic}, by way of an analytic proof strategy intended to guarantee termination of the proof procedure.  For now, to deal with the exceptional cases mentioned in Lemma~\ref{lemma:intersect}, we need to devote some attention to the case of intersection formulas.

Given a ground formula~$\varphi$, it is clearly the case that $\widehat{\varphi}\in\mathcal{V}_n$ and thus either $\widehat{\varphi}\in\Vd_{m,n}$ or $\widehat{\varphi}\in\Vu_{m,n}$.  
We may capture the fact that $t(\widehat{\varphi})=X$, for $X\in\{F,T\}$, by the following bivalent `$G$-statement':
\begin{equation}
\XX{}{c}\varphi {\quad\Longrightarrow\quad} \divideontimes. \tag{$\textsl{G}(\varphi)$}\label{eq:ground}
\end{equation}

In what follows, recall that $\mathsf{Sem}$ collects all 
$n$-valuations characterizing~$\mathcal{L}$.

\begin{lemma}\label{okground}
Let $\val\in\mathsf{Sem}$ and let $\varphi\in\mathcal{S}$ be a ground formula. Then, $\val$ sat\-is\-fies $\textsl{G}(\varphi)$.
\end{lemma}
\begin{proof}
Note that $\val$ does not satisfy the left-hand side of $\textsl{G}(\varphi)$, as $t(\val(\varphi))=t(\widehat{\varphi})=X$, and $X\neq X^c$.\qed
\end{proof}

\begin{example}[Statements for intersection formulas]
\label{synt-coincidence-3}
Returning to G$^+_4$, from Ex.~\ref{synt-coincidence-1}, recall that $\iota(\theta_1 a_2,\theta_2 a_1)=(a_1\equiv a_2)$ is an intersection formula that takes an undesignated value, that is, such that $t\left(\frac{1}{3}\widehat{\equiv} \frac{2}{3}\right)=F$. The corresponding \textsl{G}-statement is 
\begin{equation*}
{\T {a_1\equiv a_2}} \quad\Longrightarrow\quad \divideontimes.
\end{equation*}
\end{example}

While \textit{decidability} of a given finite-valued truth-functional semantics follows straightforwardly using the truth-tabular method, it is not at all obvious that a similar result applies to logics defined by a non-truth-functional bivalent semantics.  However, for semantics defined as collections of bivalent statements as in Def.~\ref{bivstatements}, one may easily devise a decision procedure after showing that the value of any composite formula is uniquely determined by the value of its generalized subformulas.

Let us detail what we have just said.  We will call a semantics $\mathsf{Sem}$ \textsl{effective} if for determining satisfiability of a given set of formulas~$\Gamma$ it suffices to collect all partial valuations over the proper generalized subformulas of~$\Gamma$.  In particular, the value of any formula~$\varphi$ with $\mathsf{cplx}(\varphi)>0$ will be a function of the formulas in $\mathsf{pgsbf}(\varphi)$. 
Alternatively, for effectiveness one might require that any partial valuation~$\left.\val\right|_\mathcal{R}$ defined over a collection of formulas~$\mathcal{R}$ closed under generalized subformulas should be extendable into a full valuation~$\val\in\mathsf{Sem}$ (cf.~\cite{Avron:aNDVoNCN}).

The above notion of effectiveness relies on the intuition that there should be a computable set $\Gamma^\star$ that collects the formulas that `really matter' in determining the value of~$\Gamma$.  Typically, $\Gamma^\star$ denotes the set of atomic variables occurring in~$\Gamma$, or the set of proper subformulas of~$\Gamma$; here we consider instead the set of proper generalized subformulas of~$\Gamma$.
In fact, to endow an effective semantics with a decidability procedure, in general we only need $\mathsf{pgsbf}$ to be such that: 
$(a)$ $\mathsf{pgsbf}(\varphi)$ is finite, for every $\varphi$; $(b)$~$\mathsf{pgsbf}(\varphi)=\varnothing$ if $\varphi$ is a basic formula; $(c)$~$\mathsf{pgsbf}(\varphi)\subsetneq\mathsf{pgsbf}(\psi)$ if $\varphi\in\mathsf{pgsbf}(\psi)$.
In that case, everything boils down to noticing that the value of a (generalized) composite formula~$\varphi$ is a function of the values of its (generalized) atomic content (a specific subset of $\mathsf{pgsbf}(\varphi)$).


It is not hard to show that $\mathsf{Sem}_2$, as obtained from our algorithm, is effective.  In fact, the value of a composite formula~$\varphi$ is not only calculable from the value of its generalized subformulas, but it can be computed by at most $\mathsf{cplx}(\varphi)$-nested uses of the bivalent statements.  More precisely:

\begin{proposition}\label{effective}
Let $b\in\mathsf{Sem}_2$ and $\varphi(p_1,\dots,p_k)\in\mathcal{S}$. Then:
\begin{itemize}
\item the value $b(\varphi)$ is uniquely determined from the values $b(\theta_r(p_i))$ for all $0\leq r\leq s$ and $1\leq i\leq k$;
\item the value $b(\varphi)$ can be computed using the $\mathcal{B}(\mathcal{L},\overline{\theta})$ statements.
\end{itemize}
\end{proposition}
\begin{proof}
The first result is a corollary of the proof of Prop.~\ref{oksem}. Indeed, $b=t\circ\val$ for some valuation $\val\in\mathsf{Sem}$ such that, for each $1\leq i\leq k$,  $\val(p_i)=x_i$ is the unique value such that $\overline{\theta}(x_i)=\langle b(p_i),b(\theta_1(p_i)),\dots,b(\theta_s(p_i))\rangle$. Therefore, $b(\varphi)=t(\val(\varphi))=t(\widehat{\varphi}(x_1,\dots,x_n))$. We show next, by induction on the complexity~\textsf{cplx} of~$\varphi$, how $b(\varphi)$ can be computed using the bivalent statements.

If $\varphi=\theta_r(p_i)$ for some $0\leq r\leq s$ and $1\leq i\leq k$, then $b(\theta_r(p_i))$ is given. Alternatively, if $\varphi=\theta_r(a)$ for some $a\in\Sigma_0$ and $0\leq r\leq s$, then, as explained in Remark~\ref{complementaryBHV}, one of the rules $B^{\theta_r a}_X$ has $\top$ on the right-hand side, while the complementing rule $B^{\theta_r a}_{X^c}$ has $\divideontimes$ on the right-hand side. Clearly, $b(\theta_r(a))=X$. Otherwise, if $\varphi=\iota(\theta_{r_1}{\odot}_1,\theta_{r_2}{\odot}_2)$ is an intersection formula then~$b$ satisfies the $\textsl{G}(\varphi)$ statement with $\XX{}{c}\varphi$ on the left-hand side, and $b(\iota(\theta_{r_1}{\odot}_1,\theta_{r_2}{\odot}_2))=X$.
Suppose now that $\varphi=\theta_r({\odot}(\varphi_1(p_1,\dots,p_k),\dots,\varphi_m(p_1,\dots,p_k)))$ is an analyzable proper $\theta_r{\odot}$-formula.
By induction hypothesis, we can compute $\overline{X}_{\varphi_j}=\langle b(\theta_0(\varphi_j)),\dots,b(\theta_t(\varphi_j))\rangle$ for each $1\leq j\leq m$, 
thus determining 
a unique vector $\overline{x}=\langle x_1,\dots,x_m\rangle\in(\mathcal{V}_n)^m$ such that $\overline{\theta}(x_j)=\overline{X}_{\varphi_j}$. Now, if $\overline{x}\in R^{\theta_r{\odot}}_X$ then~$b$ satisfies a disjunct from the right-hand side of the rule, and so $b(\varphi)=X$.\qed
\end{proof}

Decidability is an obvious corollary of the above result.  Generalizations of the usual truth-tabular method in terms of the so-called \textit{quasi matrix procedure} (cf.~\cite{dac:alv:1977}) may be developed in order to take generalized subformulas appropriately into account (cf.~the `dyadic semantics' in~\cite{ccal:car:con:mar:03a}).  Instead of doing just that, in the next sections we will show instead how our bivalent semantics may be utilized in associating adequate analytic proof procedures to the same logics, that will at once guarantee decidability and also exhibit other pleasant computational features.

\section{Uniform Analytic Classic-Like Branching Tableaux} \label{TableauExtraction}

\noindent
The results from the preceding section will now be used in showing how the classic-like semantics associated to a given finite-valued logic by means of bivalent statements may be exploited so as to devise an adequate classic-like tableau-based proof formalism for the logic.
Before accomplishing that task, however, we first have to discuss a couple of technical issues related to the characterization of the unobtainable semantic scenarios (which will opportunely give rise to nonstandard tableau closure rules). 
Ultimately, to guarantee also termination in the practice of tableau development (thus determining a decision procedure), we resort in fact to a suitable generalization of analyticity.  To address efficiency aspects, at a later stage we will also consider a reformulation of the standard branching tableaux in terms of linear tableaux (with analytic cuts).

\subsection{Dealing with partial information and with intersection formulas}\label{subsec:partialinfo}

\noindent 
Recall from Def.~\ref{bivstatements} the description of the collection $\mathcal{B}(\mathcal{L},\overline{\theta})$ of bivalent statements ($\textsl{U}$-statements and $\textsl{B}$-statements) associated to a given separable finite-valued logic~$\mathcal{L}$ with separating sequence $\overline{\theta}=\langle\theta_r\rangle_{r=0}^s$.  To formulate our method for associating also a classic-like tableau system to~$\mathcal{L}$ 
we must first take a closer look at the unobtainable binary prints.


Intuitively, the binary prints in $\{F,T\}^{s+1}\setminus\overline{\theta}[\mathcal{V}_n]$
bring about information that does not represent any of the original truth-values of~$\mathcal{L}$.  
As we have shown, such unobtainable binary prints are conveniently expressed by the bivalent $\textsl{U\/}$-statements in $\mathcal{B}(\mathcal{L},\overline{\theta})$. 
However, for each unobtainable $\overline{X}\in(\{F,T\}^{s+1}\setminus\overline{\theta}[\mathcal{V}_n])$, the statement $\textsl{U\/}\overline{X}$ is, in general, too coarse for our purposes.
What we seek is to identify information ---even if partial--- about a binary print that leads forcibly and unambiguously to an unobtainable state-of-affairs. 

\begin{definition}\label{partial}
Let a \textsl{partial binary print} be any sequence $\overline{Y}\in\{F,T,{\uparrow}\}^{s+1}$ where the symbol $\uparrow$ stands for undefinedness (this definition
includes, of course, the \textsl{total} binary prints in $\{F,T\}^{s+1}$). By  $\dom(\overline{Y})$ we denote the set $\{0 \leq r\leq s:Y_r\neq {\uparrow}\}$.
Given two partial binary prints~$\overline{Y}$ and~$\overline{Z}$, we say that~$\overline{Y}$ \textsl{extends}~$\overline{Z}$ if
$\dom(\overline{Z})\subseteq \dom(\overline{Y})$ and $Z_r=Y_r$ for every $r\in \dom(\overline{Z})$.
A partial binary print $\overline{Y}\in\{F,T,\uparrow\}^{s+1}$ is \textsl{unobtainable} if all of its $2^{s+1-|\dom(\overline{Y})|}$ possible total extensions are unobtainable. 
Further, an unobtainable partial binary print~$\overline{Y}$ is said to be \textsl{minimal} if it is not an extension of another unobtainable partial binary print. 
\end{definition}

Extending notation from Section~\ref{sec:classiclike} to cover also partial binary prints, given $\overline{Y}\in\{F,T,{\uparrow}\}^{s+1}$, we will write
\begin{equation}
({\&}_{r\in\dom(\overline{Y})} \Y{r}\theta_r(\varphi)). \tag{${\textsl{V}(\varphi\,;\overline{Y})}$}\label{eq:value2}
\end{equation}
to say that such partial binary print describes (more economically) the original truth-value of~$\varphi$.
As before, we may capture the fact that $\overline{Y}$ is unobtainable by the following statement involving an arbitrary $\varphi\in\mathcal{S}$:
\begin{equation}
{\textsl{V}(\varphi\,;\overline{Y})} {\quad\Longrightarrow\quad} \divideontimes. \tag{$\textsl{U\/}\overline{Y}$}\label{eq:partial}
\end{equation}

\begin{lemma}\label{okpartial}
Let $b:\mathcal{S}\longrightarrow\{F,T\}$ be a bivaluation. 
The (total) binary print $\overline{X}=\langle{b(\varphi),b(\theta_1(\varphi)),\dots,b(\theta_s(\varphi))}\rangle$ is obtainable 
if and only if $b$ satisfies $\textsl{U\/}\overline{Y}$ for every minimal unobtainable partial binary print~$\overline{Y}$.
\end{lemma}
\begin{proof}
%
Suppose that $b$ does not satisfy $\textsl{U\/}\overline{Y}$ for some minimal unobtainable partial binary print $\overline{Y}$. Thus, $b$ satisfies the left-hand side of $\textsl{U\/}\overline{Y}$, and therefore $\overline{X}$ must extend $\overline{Y}$. Hence, as $\overline{Y}$ is unobtainable and $\overline{X}$ is total, it must be the case that $\overline{X}$ is also unobtainable.

Assume now $\overline{X}$ is unobtainable and let $\overline{Y}$ be one (of the possibly many) min\-imal partial binary print extended by~$\overline{X}$. Then, it is clear that $b$ satisfies the left-hand side of $\textsl{U\/}\overline{Y}$ and so~$\overline{Y}$ is unobtainable.\qed
\end{proof}

The latter result means that, in general, one may replace the $\textsl{U\/}$-statements concerning total unobtainable binary prints by the $\textsl{U\/}$-statements for minimal unobtainable binary prints. 

\begin{example}[Unobtainable binary prints]\label{ex:partialunobtainable}
Consider first the case of \L$_3$, from Ex.~\ref{lukas2} and~\ref{lukas3}. Recall that using the separating sequence $\overline{\theta}=\langle \mathsf{id},\theta\rangle$, where $\theta=\lambda p.(\neg p\supset p)$ defines the unary operator $j^1_\geq$ of Ex.~\ref{ex:separatingTF}, we get the binary prints $\overline{\theta}(0)=\langle F,F\rangle$, $\overline{\theta}(\frac{1}{2})=\langle F,T\rangle$ and $\overline{\theta}(1)=\langle T,T\rangle$. The only unobtainable binary print $\langle T,F\rangle$ is therefore also minimal.

Consider now the case of G$^+_4$, from Ex.~\ref{godel}, where we add to~G$_4$ the sentential constants~$a_1$ and $a_2$ such that $\widehat{a_1}=\frac{1}{3}$ and $\widehat{a_2}=\frac{2}{3}$, separated by $\overline{\theta}=\langle \mathsf{id},\theta_1,\theta_2\rangle$, where $\theta_1=\lambda p.(a_1\equiv p)$ and $\theta_2=\lambda p.(p\equiv a_2)$. Recall that $\theta_1$ and $\theta_2$ define the unary operators $k_1$ and $k_2$, respectively. We get $\overline{\theta}(0)=\langle F,F,F\rangle$, $\overline{\theta}(\frac{1}{3})=\langle F,T,F\rangle$, $\overline{\theta}(\frac{2}{3})=\langle F,F,T\rangle$ and $\overline{\theta}(1)=\langle T,F,F\rangle$. The remaining four binary prints $\langle T,T,F\rangle,\langle F,T,T\rangle,\langle T,F,T\rangle,\langle T,T,T\rangle$ are unobtainable, but are not minimal. It is clear that any binary print with more than one~$T$ is unobtainable. Thus, the minimal unobtainable (strictly partial, in this case) binary prints are $\langle T,T,\uparrow\rangle,\langle T,\uparrow,T\rangle,\langle \uparrow,T,T\rangle$. For the sake of the illustration, the $\textsl{U\/}$-statements that originate from these partial binary prints are listed in Table~\ref{tab:ubps}.
\end{example}
\begin{table}[htbp]\small
\begin{center}
\scalebox{.7}{
\begin{tabular}{|r|lcl|}
\hline
 ($\textsl{U\/}\langle T,T,\uparrow\rangle$) & $({\T\varphi} {\;\;\&\;\;} {\T\theta_1(\varphi)})$ & $\Longrightarrow$ & $\divideontimes$\\
\hline
 ($\textsl{U\/}\langle T,\uparrow,T\rangle$) & $({\T\varphi} {\;\;\&\;\;} {\T\theta_2(\varphi)})$ & $\Longrightarrow$ & $\divideontimes$\\
\hline
 ($\textsl{U\/}\langle\uparrow,T,T\rangle$) & $({\T\theta_1(\varphi)} {\;\;\&\;\;} {\T\theta_2(\varphi)})$ & $\Longrightarrow$ & $\divideontimes$\\
\hline
\end{tabular}
}
\end{center}
\caption{$\textsl{U\/}$-statements for G$^+_4$ separated by $\langle \mathsf{id},\theta_1,\theta_2\rangle$.}
\label{tab:ubps}
\end{table}
Taking minimal unobtainable partial binary prints into account will prove essential for guaranteeing completeness of our tableau systems, later on.  


\subsection{Tableaux from bivaluations}\label{subsec:branchtableaux}

\noindent
We are almost ready for defining an appropriate classic-like deductive characterization of the bivalent semantics of the logic at hand.

\begin{definition}\label{tabstatements}
The set $\mathcal{B}_\mathcal{T}(\mathcal{L},\overline{\theta})$ of \textsl{classic-like tableau statements} associated to $\mathcal{L}$, fixed a separating sequence $\overline{\theta}=\langle\theta_r\rangle_{r=0}^s$, is formed by all instances of the following {$\textsl{U\/}$-statements}, {$\textsl{B}$-statements} and {$\textsl{G}$-statements}:
\begin{itemize}
\item $\textsl{U\/}\overline{Y}$, for each minimal unobtainable partial binary print $\overline{Y}$;  
\item $\textsl{B\/}_X^{\theta_r{\odot}}$, for each $X\in\{F,T\}$, $0\leq r\leq s$ and ${\odot}\in\Sigma$; and
\item $\textsl{G}(\iota(\theta_{r_1}{\odot}_1,\theta_{r_2}{\odot}_2))$, for each intersection formula $\iota(\theta_{r_1}{\odot}_1,\theta_{r_2}{\odot}_2)$ with $0\leq r_1,r_2\leq s$ and ${\odot}_1,{\odot}_2\in\Sigma$.
\end{itemize}
\end{definition}

In the following, recall that $\mathsf{Sem}_2$ is the bivalent semantics of the $n$-valued logic $\mathcal{L}$ produced by Def.~\ref{dyadic} and used in Prop.~\ref{Sred}.

\begin{proposition}\label{ok}
$\mathsf{Sem}_2$ is the set of all bivaluations that satisfy $\mathcal{B}_\mathcal{T}(\mathcal{L},\overline{\theta})$.
\end{proposition}
\begin{proof}
Immediate from Prop.~\ref{oksem}, and Lemmas~\ref{okpartial} and \ref{okground}.\qed
\end{proof}

To formulate our classic-like tableau systems for separable finite-valued logics we need one final ingredient, captured by the statement below:
\begin{equation}
(\F\varphi {\;\;\&\;\;} \T\varphi) {\quad\Longrightarrow\quad} \divideontimes. \tag{$\textsl{ABS\/}$}\label{eq:absurd}
\end{equation}
\noindent 
Notice indeed that the left-hand side of such $\textsl{ABS\/}$-statement is satisfied by no bivaluation, given the functional character of valuations in general.

We are now ready to define our classically-labeled
tableau system for $\mathcal{L}$. As customary, we will represent (branching) \textsl{tableau rules} by
$$\frac{H_1,\dots,H_n}{C_{1,1},\dots,C_{1,n_1}\mid{\;\;\dots\;\;}\mid C_{k,1},\dots,C_{k,n_k}}$$  
\noindent 
where $H_1,\dots,H_n$ are the \textsl{premises} and $C_{1,1},\dots,C_{1,n_1}\mid\dots\mid C_{k,1},\dots,C_{k,n_k}$ is the \textsl{conclusion} of the rule, where each list $C_{i,1},\dots,C_{i,n_i}$, for $1\leq i\leq k$, represents a \textsl{branch}. In our setting, as all the $H$s and $C$s are classically-labeled formulas, we will denote such a rule by $$\mathcal{R}((H_1\&\dots\& H_n)\Longrightarrow
(C_{1,1}\&\dots\& C_{1,n_1})\mid\mid{\;\;\dots\;\;}\mid\mid (C_{k,1}\&\dots\& C_{k,n_k})
).$$

\noindent 
Note that this notation univocally associates a rule to each bivalent statement whose left-hand side is a conjunction and whose right-hand side is in disjunctive normal form (recall Remarks~\ref{DNF} and \ref{simplification}).

\begin{definition}\label{tableau}
The \textsl{classic-like tableau system ${\mathcal{T}}(\mathcal{L},\overline{\theta})$ associated to $\mathcal{L}$ (and~$\overline{\theta}$)} is composed of the rules $\mathcal{R}(S)$ for $S\in\mathcal{B}_\mathcal{T}(\mathcal{L},\overline{\theta})$, 
plus $\mathcal{R}(\textsl{ABS\/})$.
\end{definition}

We will call \textsl{closure rules} all those rules whose conclusion contain the single branch~$\divideontimes$, and nothing else. Clearly, this includes $\mathcal{R}(\textsl{ABS\/})$, as well as all the rules $\mathcal{R}(\textsl{U\/}\overline{Y})$, for minimally unobtainable~$\overline{Y}$, and also the rules for intersection formulas $\textsl{G}(\iota(\theta_{r_1}{\odot}_1,\theta_{r_2}{\odot}_2))$. Note that a closure rule may also result from $\mathcal{R}(\textsl{B\/}_X^{\theta_r{\odot}})$, with $X\in\{F,T\}$, $0\leq r\leq s$ and ${\odot}\in\Sigma_k$, in case the signed formula $\X{}\theta_r({\odot}(\varphi_1,\dots,\varphi_k))$ is unsatisfiable, that is, when $R_X^{\theta_r{\odot}}=\varnothing$ (see Remark~\ref{complementaryBHV}). 

\begin{example}[Classic-like tableau systems for \L$_3$ and for G$^+_4$]
\label{systems-and-tableaux}
Let's go back to \L$_3$, in the streamlined form given in Ex.~\ref{lukas3} and Table~\ref{stream:L3}, and using the separating sequence $\overline{\theta}=\langle \mathsf{id},\theta\rangle$ with $\theta=\lambda p.(\neg p\supset p)$.  In that case Def.~\ref{tableau} outputs the tableau system ${\mathcal{T}}(\textrm{\L$_3$},\overline{\theta})$ consisting of the rules below.

\begin{center}
\scalebox{1.1}{
\begin{tabular}[htbp]{c}
$\textrm{\tiny $\mathcal{R}(\textsl{ABS\/})$}
\frac{{\F\varphi},{\T\varphi}}{\divideontimes}\quad\quad\quad\quad
\textrm{\tiny $\mathcal{R}(\textsl{U\/}(T,F))$}\frac{{\T\varphi},{\F\theta(\varphi)}}{\divideontimes}$\\[2.5mm]
\end{tabular}
}

\scalebox{1.1}{
\begin{tabular}[htbp]{c}
$\textrm{\tiny $\mathcal{R}(\textsl{B}_F^{\,\neg})$}
\frac{{\F\neg\varphi}}{{\T\theta(\varphi)}}\quad\quad
\textrm{\tiny $\mathcal{R}(\textsl{B}_T^{\,\neg})$}
\frac{{\T\neg\varphi}}{{\F\theta(\varphi)}}\quad\quad
\textrm{\tiny $\mathcal{R}(\textsl{B}_F^{\,\theta\neg})$}
\frac{{\F\theta(\neg\varphi)}}{{\T\varphi}}\quad\quad
\textrm{\tiny $\mathcal{R}(\textsl{B}_T^{\,\theta\neg})$}
\frac{{\T\theta(\neg\varphi)}}{{\F\varphi}}$\\[2.5mm]
\end{tabular}
}

\scalebox{1.1}{
\begin{tabular}[htbp]{c}
$\textrm{\tiny $\mathcal{R}(\textsl{B}_F^{\,\supset})$}
\frac{{\F\varphi\supset\psi}}{{\T\varphi},{\F\psi}\mid {\T\theta(\varphi)},{\F\theta(\psi)}}\quad\quad\quad
\textrm{\tiny $\mathcal{R}(\textsl{B}_T^{\,\supset})$}
\frac{{\T\varphi\supset\psi}}{{\F\varphi},{\T\theta(\psi)}\mid{\F\theta(\varphi)}\mid {\T\psi}}$\\[2.5mm]
\end{tabular}
}

\scalebox{1.1}{
\begin{tabular}[htbp]{c}
$
\textrm{\tiny $\mathcal{R}(\textsl{B}_F^{\,\lor})$}
\frac{{\F\varphi\lor\psi}}{{\F\varphi},{\F\psi}}\quad\quad
\textrm{\tiny $\mathcal{R}(\textsl{B}_T^{\,\lor})$}
\frac{{\T\varphi\lor\psi}}{{\T\varphi}\mid {\T\psi}}
\quad\quad
\textrm{\tiny $\mathcal{R}(\textsl{B}_F^{\,\land})$}
\frac{{\F\varphi\land\psi}}{{\F\varphi}\mid{\F\psi}}\quad\quad
\textrm{\tiny $\mathcal{R}(\textsl{B}_T^{\,\land})$}
\frac{{\T\varphi\land\psi}}{{\T\varphi},{\T\psi}}
$\\[2.5mm]
\end{tabular}
}

\scalebox{1.1}{
\begin{tabular}[htbp]{c}
$
\textrm{\tiny $\mathcal{R}(\textsl{B}_F^{\,\theta\supset})$}
\frac{{\F\theta(\varphi\supset\psi)}}{{\T\varphi},{\F\theta(\psi)}}
\quad\quad
\textrm{\tiny $\mathcal{R}(\textsl{B}_F^{\,\theta\lor})$}
\frac{{\F\theta(\varphi\lor\psi)}}{{\F\theta(\varphi)},{\F\theta(\psi)}}
\quad\quad
\textrm{\tiny $\mathcal{R}(\textsl{B}_F^{\,\theta\land})$}
\frac{{\F\theta(\varphi\land\psi)}}{{\F\theta(\varphi)}\mid{\F\theta(\psi)}}
$
\\[2.5mm]
\end{tabular}
}

\scalebox{1.1}{
\begin{tabular}[htbp]{c}
$
\textrm{\tiny $\mathcal{R}(\textsl{B}_T^{\,\theta\supset})$}
\frac{{\T\theta(\varphi\supset\psi)}}{{\F\varphi}\mid {\T\theta(\psi)}}
\quad\quad
\textrm{\tiny $\mathcal{R}(\textsl{B}_T^{\,\theta\lor})$}
\frac{{\T\theta(\varphi\lor\psi)}}{{\T\theta(\varphi)}\mid {\T\theta(\psi)}}
\quad\quad
\textrm{\tiny $\mathcal{R}(\textsl{B}_T^{\,\theta\land})$} 
\frac{{\T\theta(\varphi\land\psi)}}{{\T\theta(\varphi)},{\T\theta(\psi)}}
$  
\end{tabular}}
\end{center}
Let us now return to G$^+_4$, from Ex.~\ref{godel} and~\ref{ex:partialunobtainable}, where we added to G$_4$ the sentential constants~$a_1$ and~$a_2$, and employed the separating sequence $\overline{\theta}=\langle \mathsf{id},\theta_1,\theta_2\rangle$, with $\theta_1=\lambda p.(a_1\equiv p)$ and $\theta_2=\lambda p.(p\equiv a_2)$. According to Def.~\ref{tableau}, the corresponding classic-like (streamlined) tableau system ${\mathcal{T}}(\textrm{G$^+_4$},\overline{\theta})$ is composed of the rules below.
\begin{center}
\scalebox{1.1}{
\begin{tabular}[htbp]{c}
$\textrm{\tiny $\mathcal{R}(\textsl{ABS})$}
\frac{{\F\varphi},{\T\varphi}}{\divideontimes}\quad\quad
\textrm{\tiny $\mathcal{R}(\textsl{U\/}\langle T,T,\uparrow\rangle)$}
\frac{{\T\varphi},{\T\theta_1(\varphi)}}{\divideontimes}\quad\quad
\textrm{\tiny $\mathcal{R}(\textsl{U\/}\langle T,\uparrow,T\rangle)$}
\frac{{\T\varphi},{\T\theta_2(\varphi)}}{\divideontimes}\quad\quad$\\[2.5mm]
\end{tabular}
}

\scalebox{1.1}{
\begin{tabular}[htbp]{c}
$\textrm{\tiny $\mathcal{R}(\textsl{U\/}\langle\uparrow,T,T\rangle)$}
\frac{{\T\theta_1(\varphi)},{\T\theta_2(\varphi)}}{\divideontimes}\quad\quad
\textrm{\tiny $\mathcal{R}(\textsl{B}^{\theta_1 a_1}_F)$}
\frac{{\F\theta_1(a_1)}}{\divideontimes}\quad\quad
\textrm{\tiny $\mathcal{R}(\textsl{B}^{\theta_2 a_2}_F)$}
\frac{{\F\theta_2(a_2)}}{\divideontimes}\quad\quad$\\[2.5mm]
\end{tabular}
}

\scalebox{1.1}{
\begin{tabular}[htbp]{c}
$\textrm{\tiny $\mathcal{R}(\textsl{B}^{a_1}_T)$}
\frac{{\T a_1}}{\divideontimes}\quad\quad
\textrm{\tiny $\mathcal{R}(\textsl{B}^{a_2}_T)$}
\frac{{\T a_2}}{\divideontimes}\quad\quad
\textrm{\tiny $\mathcal{R}(\textsl{B}^{\theta_1\neg}_T)$}
\frac{{\T\theta_1(\neg\varphi)}}{\divideontimes}\quad\quad
\textrm{\tiny $\mathcal{R}(\textsl{B}^{\theta_2\neg}_T)$}
\frac{{\T\theta_2(\neg\varphi)}}{\divideontimes}$\\[2.5mm]
\end{tabular}
}

\scalebox{1.1}{
\begin{tabular}[htbp]{c}
$\textrm{\tiny $\mathcal{R}(\textsl{G}(a_1\equiv a_2))$}
\frac{{\T a_1\equiv a_2}}{\divideontimes}\quad\quad
\textrm{\tiny $\mathcal{R}(\textsl{B}^{\land}_F)$}
\frac{{\F\varphi\land\psi}}{{\F\varphi}\mid {\F\psi}}\quad\quad
\textrm{\tiny $\mathcal{R}(\textsl{B}^{\lor}_F)$}
\frac{{\F\varphi\lor\psi}}{{\F\varphi},{\F\psi}}\quad\quad
\textrm{\tiny $\mathcal{R}(\textsl{B}^{\land}_T)$}
\frac{{\T\varphi\land\psi}}{{\T\varphi},{\T\psi}}$\\[2.5mm]
\end{tabular}
}

\scalebox{1.1}{
\begin{tabular}[htbp]{c}
$\textrm{\tiny $\mathcal{R}(\textsl{B}^{\lor}_T)$}
\frac{{\T\varphi\lor\psi}}{{\T\varphi}\mid{\T\psi}}\quad\quad
\textrm{\tiny $\mathcal{R}(\textsl{B}^{\neg}_F)$}
\frac{{\F\neg\varphi}}{{\T\varphi}\mid{\T\theta_1(\varphi)}\mid {\T\theta_2(\varphi)}}\quad\quad
\textrm{\tiny $\mathcal{R}(\textsl{B}^{\neg}_T)$}
\frac{{\T\neg\varphi}}{{\F\varphi},{\F\theta_1(\varphi)},{\F\theta_2(\varphi)}}$\\[2.5mm]
\end{tabular}
}

\scalebox{1.1}{
\begin{tabular}[htbp]{c}
$\textrm{\tiny $\mathcal{R}(\textsl{B}^{\supset}_F)$}
\frac{{\F\varphi\supset\psi}}{{\T\varphi},{\F\psi}\mid{\T\theta_2(\varphi)},{\F\psi},{\F\theta_2(\psi)}\mid {\T\theta_1(\varphi)},{\F\psi},{\F\theta_1(\psi)},{\F\theta_2(\psi)}}$\\[2.5mm]
\end{tabular}
}

\scalebox{1.1}{
\begin{tabular}[htbp]{c}
$\textrm{\tiny $\mathcal{R}(\textsl{B}^{\supset}_T)$}
\frac{{\T\varphi\supset\psi}}{{\T\psi}\mid{\F\varphi},{\T\theta_2(\psi)}\mid {\F\varphi},{\F\theta_1(\varphi)},{\F\theta_2(\varphi)}\mid{\T\theta_1(\varphi)},{\T\theta_1(\psi)}}$\\[2.5mm]
\end{tabular}
}

\scalebox{1.1}{
\begin{tabular}[htbp]{c}
$\textrm{\tiny $\mathcal{R}(\textsl{B}^{\theta_1\supset}_F)$}
\frac{{\F\theta_1(\varphi\supset\psi)}}{{\F\theta_1(\psi)}\mid {\F\varphi},{\F\theta_2(\varphi)}}\quad\quad
\textrm{\tiny $\mathcal{R}(\textsl{B}^{\theta_2\supset}_F)$}
\frac{{\F\theta_2(\varphi\supset\psi)}}{{\F\varphi}\mid {\F\theta_2(\psi)}}\quad\quad
\textrm{\tiny $\mathcal{R}(\textsl{B}^{\theta_2\supset}_T)$}
\frac{{\T\theta_2(\varphi\supset\psi)}}{{\T\varphi},{\T\theta_2(\psi)}}$\\[2.5mm]
\end{tabular}
}

\scalebox{1.1}{
\begin{tabular}[htbp]{c}
$\textrm{\tiny $\mathcal{R}(\textsl{B}^{\theta_1\supset}_T)$}
\frac{{\T\theta_1(\varphi\supset\psi)}}{{\T\varphi},{\T\theta_1(\psi)}\mid {\T\theta_2(\varphi)},{\T\theta_1(\psi)}}\quad\quad
\textrm{\tiny $\mathcal{R}(\textsl{B}^{\theta_2\lor}_F)$}
\frac{{\F\theta_2(\varphi\lor\psi)}}{{\F\theta_2(\varphi)},{\F\theta_2(\psi)}\mid{\T\varphi},{\T\psi}}$\\[2.5mm]
\end{tabular}
}

\scalebox{1.1}{
\begin{tabular}[htbp]{c}
$\textrm{\tiny $\mathcal{R}(\textsl{B}^{\theta_1\lor}_F)$}
\frac{{\F\theta_1(\varphi\lor\psi)}}{{\F\theta_1(\varphi)},{\F\theta_1(\psi)}\mid{\T\varphi},{\T\theta_2(\varphi)}\mid {\T\psi},{\T\theta_2(\psi)}}$\\[2.5mm]
\end{tabular}
}

\scalebox{1.1}{
\begin{tabular}[htbp]{c}
$\textrm{\tiny $\mathcal{R}(\textsl{B}^{\theta_1\land}_F)$}
\frac{{\F\theta_1(\varphi\land\psi)}}{{\F\theta_1(\varphi)},{\F\theta_1(\psi)}\mid{\F\varphi},{\F\theta_1(\varphi)},{\F\theta_2(\varphi)}\mid {\F\psi},{\F\theta_1(\psi)},{\F\theta_2(\psi)}}$\\[2.5mm]
\end{tabular}
}

\scalebox{1.1}{
\begin{tabular}[htbp]{c}
$\textrm{\tiny $\mathcal{R}(\textsl{B}^{\theta_2\land}_F)$}
\frac{{\F\theta_2(\varphi\land\psi)}}{{\F\theta_2(\varphi)},{\F\theta_2(\psi)}\mid{\F\varphi},{\F\theta_2(\varphi)}\mid {\F\psi},{\F\theta_2(\psi)}}$\\[2.5mm]
\end{tabular}
}

\scalebox{1.1}{
\begin{tabular}[htbp]{c}
$\textrm{\tiny $\mathcal{R}(\textsl{B}^{\theta_1\land}_T)$}
\frac{{\T\theta_1(\varphi\land\psi)}}{{\T\varphi},{\T\theta_1(\psi)}\mid{\T\theta_1(\varphi)},{\T\psi}\mid {\T\theta_1(\varphi)},{\T\theta_1(\psi)}\mid {\T\theta_1(\varphi)},{\T\theta_2(\psi)}\mid {\T\theta_2(\varphi)},{\T\theta_1(\psi)}}$\\[2.5mm]
\end{tabular}
}

\scalebox{1.1}{
\begin{tabular}[htbp]{c}
$\textrm{\tiny $\mathcal{R}(\textsl{B}^{\theta_1\lor}_T)$}
\frac{{\T\theta_1(\varphi\lor\psi)}}{{\F\varphi},{\F\theta_2(\varphi)},{\T\theta_1(\psi)}\mid{\T\theta_1(\varphi)}, {\F\psi},{\F\theta_2(\psi)}}$\\[2.5mm]
\end{tabular}
}

\scalebox{1.1}{
\begin{tabular}[htbp]{c}
$\!\!\textrm{\tiny $\mathcal{R}(\textsl{B}^{\theta_2\land}_T)$}
\frac{{\T\theta_2(\varphi\land\psi)}}{{\T\varphi},{\T\theta_2(\psi)}\mid{\T\theta_2(\varphi)},{\T\psi}\mid {\T\theta_2(\varphi)},{\T\theta_2(\psi)}}\;\;
\textrm{\tiny $\mathcal{R}(\textsl{B}^{\theta_2\lor}_T)$}
\frac{{\T\theta_2(\varphi\lor\psi)}}{{\F\varphi},{\T\theta_2(\psi)}\mid{\T\theta_2(\varphi)},{\F\psi}}$\\[2.5mm]
\end{tabular}}
\end{center}
In accordance with our streamlining procedure, described in Remark~\ref{simplification}, we have chosen to omit above all tautological rules, namely those corresponding to $\textsl{B}^{a_1}_F$, $\textsl{B}^{a_2}_F$, $\textsl{B}^{\theta_1 a_1}_T$, $\textsl{B}^{\theta_2 a_2}_T$, $\textsl{B}^{\theta_1 a_2}_F$, $\textsl{B}^{\theta_2 a_1}_F$, $\textsl{B}^{\theta_1 \neg}_F$ and $\textsl{B}^{\theta_2 \neg}_F$. Also missing are the rules corresponding to $\textsl{B}^{\theta_1 a_2}_T$ and $\textsl{B}^{\theta_2 a_1}_T$ as they in fact coincide with $\mathcal{R}(\textsl{G}(a_1\equiv a_2))$.
\end{example}

It is worth noting that for classical logic (no separator formulas needed besides identity) our procedure will output essentially Smullyan's analytic tableaux (cf.~\cite{smu:FOL}).

Tableaux develop as usual, by applying rules and building trees that start from some root consisting of a given set of classically-labeled formulas. In practical terms, if the premises in the head of a rule~$\mathcal{R}$ are jointly matched by formulas in a certain branch of the tableau, then, through the \textsl{application} of~$\mathcal{R}$ the tableau is extended by ramifying that very branch into as many branches as those in the conclusion of~$\mathcal{R}$, each such branch comprising the labeled formulas in the original branch plus the suitably instantiated formulas from the corresponding branch in the conclusion of~$\mathcal{R}$.
A branch is said to be \textsl{closed} if it contains~$\divideontimes$, and a \textsl{closed tableau} is one whose branches are all closed.  If a branch is not closed it is called \textsl{open}; analogously, an \textsl{open tableau} is a tableau that has some open branch.
As usual, a branch of a tableau is said to be \textsl{exhausted} if 
all applicable rules have already been applied to it. An \textsl{exhausted tableau} is one whose branches are all exhausted.

From the general definition of our tableau systems, it is easy to check the following result with respect to the initially given truth-functional semantics.

\begin{proposition}[Soundness]\label{sound}
If an $n$-valued valuation in $\mathsf{Sem}$ satisfies some initial root set of classically-labeled
formulas, then it satisfies all the formulas in some open branch
of any tableau that develops from that root set.
\end{proposition}
\begin{proof}
We already know from Prop.~\ref{oksem}, \ref{Sred} and~\ref{ok} that a valuation
$\val:\mathbb{S}\longrightarrow\mathbb{V}\in\mathsf{Sem}$ satisfies all the characterizing tableau-like bivalent statements associated to $\mathcal{L}$. As it is obvious that $\val$ also satisfies all instances of $\textsl{ABS\/}$, this means that if $\val$ satisfies the premises of a tableau rule then $\val$ must also satisfy one of its branches. In particular, $\val$ cannot satisfy the premises of any of the closure rules.

If $\val$ satisfies an initial root set, then, by definition of tableau and repeated application of $\mathcal{R}(\textsl{B\/}_X^{\theta_r{\odot}})$ rules (or equivalent simplified versions thereof), the result is immediate, as long as we show that the branch satisfied by $\val$ can never be closed. 
Indeed, as $\val$ cannot satisfy the premises of any closure rule, this means that no closure rule can be applied and the branch is always open.
\qed
\end{proof}

According to the latter result, if one is able to produce a closed tableau from a given root set of labeled formulas, then the root set is unsatisfiable.

\begin{example}[A closed tableau]
\label{closed-tableau}
For illustration, let us consider a well-known theorem of \L$_3$: $((p\supset\neg p)\supset p)\supset p$. A closed tableau for an attempt at falsifying this formula is depicted in Fig.~\ref{extabL3}.
In this tableau \ding{172} denotes $\mathcal{R}(\textsl{B}^{\supset}_F)$, \ding{173} denotes $\mathcal{R}(\textsl{B}^{\supset}_T)$, \ding{174} denotes $\mathcal{R}(\textsl{ABS\/})$, \ding{175} denotes $\mathcal{R}(\textsl{B}^{\theta\neg}_F)$, \ding{176} denotes $\mathcal{R}(\textsl{B}^{\theta\supset}_F)$, \ding{177} denotes $\mathcal{R}(\textsl{B}^{\theta\supset}_T)$, and \ding{178} denotes $\mathcal{R}(\textsl{U\/}\langle T,F\rangle)$.  
Note that the tableau is not exhausted: there are indeed (closed) branches containing nonbasic formulas to which \ding{175} could still have been applied.
\end{example}
\begin{figure}[htbp]\scriptsize
\begin{center}
\scalebox{.8}{
\begin{tikzpicture}[sibling distance=1.14em]
	\hspace{-5mm}
	\tikzset{every tree node/.style={align=center,anchor=north}}
	\tikzset{level 1/.style={level distance=40pt}}
	\tikzset{level 2/.style={level distance=40pt}}
	\tikzset{level 3/.style={level distance=40pt}}
	\tikzset{level 4/.style={level distance=40pt}}
	\tikzset{level 5/.style={level distance=40pt}}
	 \Tree[.{$\Fs ((p\supset\neg p)\supset p)\supset p$}
	        \edge node[auto=left]{\ding{172}};
	        [.{$\Ts (p\supset\neg p)\supset p$ \\[1mm] $\Fs p$}
	          \edge node[auto=left]{\ding{173}};
	          [.{$\Fs p\supset\neg p$ \\[1mm] $\Ts \theta(p)$}
	            \edge node[auto=left]{\ding{172}};
	            [.{$\Ts p$ \\[1mm] $\Fs \neg p$}
	              \edge node[auto=left]{\ding{174}};
	              [.{$\divideontimes$}
	              ]
	            ]
	            [.{$\Fs \theta(\neg p)$}
	              \edge node[auto=left]{\ding{175}};
	              [.{$\Ts p$}
	                \edge node[auto=left]{\ding{174}};
	                [.{$\divideontimes$}
	                ]
	              ]
	            ]
	          ]
	          [.{$\Fs \theta(p\supset\neg p)$}
	            \edge node[auto=left]{\ding{176}};
	            [.{$\Ts p$ \\[1mm] $\Fs \theta(\neg p)$}
	              \edge node[auto=left]{\ding{174}};
	              [.{$\divideontimes$}
	              ]
	            ]
	          ]
	          [.{$\Ts p$}
	            \edge node[auto=left]{\ding{174}};
	            [.{$\divideontimes$}
	            ]
	          ]
	        ]
	        [.{$\Ts \theta((p\supset\neg p)\supset p)$ \\[1mm] $\Fs \theta(p)$}
	          \edge node[auto=left]{\ding{177}};
	          [.{$\Fs p\supset\neg p$}
	            \edge node[auto=left]{\ding{172}};
	            [.{$\Ts p$ \\[1mm] $\Fs \neg p$}
	              \edge node[auto=left]{\ding{178}};
	              [.{$\divideontimes$}
	              ]
	            ]
	            [.{$\Ts \theta(p)$ \\[1mm] $\Fs \theta(\neg p)$}
	              \edge node[auto=left]{\ding{174}};
	              [.{$\divideontimes$}
	              ]
	            ]
	          ]
	          [.{$\Ts \theta(p)$}
	            \edge node[auto=left]{\ding{174}};
	            [.{$\divideontimes$}
	            ]
	          ]
	        ]
	      ] %
	\end{tikzpicture}}
\caption{A closed tableau for $\F((p\supset\neg p)\supset p)\supset p$ in ${\mathcal{T}}(\textrm{\L$_3$},\overline{\theta})$.}\label{extabL3}	
\end{center}
\end{figure}

Given some particular tableau branch, say that a binary print $\langle X_r\rangle_{r=0}^s$ \textsl{agrees with the information available in the branch} if $\XX{r}{c}\theta_r(p)$ does not occur in the branch, for each $0\leq r\leq s$.
Note that this means that either $\X{r}\theta_r(p)$ is in the branch, or else neither $\T\theta_r(p)$ nor $\F\theta_r(p)$ occur in the branch.
The following result may then be proven.

\begin{proposition}[Completeness]\label{complete}
From every open branch of an exhausted tableau one may extract a valuation in $\mathsf{Sem}$ satisfying its root set.
\end{proposition}
\begin{proof}
Let us consider an open branch of an exhausted tableau. Clearly, by definition, none of the tableau closure rules is applicable, i.e., the branch does not contain the premises of any of the closure rules.

Take any
assignment $\eval:\mathcal{A}\rightarrow\mathcal{V}_n$ such that,
for every $p\in\mathcal{A}$, the binary print
$\overline{\theta}({\eval(p)})=\langle X_r\rangle_{r=0}^s$ agrees with the information
available in that branch.
Clearly, such an assignment always exists. Just consider the (possibly partial) binary print $\overline{X}_p=\langle X_{p,r} \rangle_{r=0}^s$ where $X_{p,r}=X_r$ if $\X{r}\theta_r(p)$ is in the branch, and $X_{p,r}={\uparrow}$ otherwise. This sequence is clearly well-defined, given that the rule $\mathcal{R}(\textsl{ABS\/})$ is not applicable. Moreover, $\overline{X}_p$ is obtainable, for none of the rules $\mathcal{R}(\textsl{U\/}\overline{Y})$, with minimally unobtainable~$\overline{Y}$, is applicable. Therefore, as $\overline{\theta}({\eval(p)})$ extends $\overline{X}_p$, we are done.

We will now show that the homomorphic extension $\val^\eval:\mathbb{S}\longrightarrow\mathbb{V}\in\mathsf{Sem}$ satisfies all the signed formulas in the branch, and consequently also the root set. The proof is somewhat similar to the proof of Prop.~\ref{oksem}, but using induction on the formula complexity instead of on the formula depth.

The base case is actually the most interesting. There are three subcases.
\begin{enumerate}
\item If $\X{}\theta_r(p)$ is in the branch for some $p\in\mathcal{A}$ and $0\leq r\leq s$ then $t(\val^\eval(\theta_r(p)))=X_{p,r}=X$, by the definition of~$e$. 
\item If $\X{}\theta_r(a)$ is in the branch for some $a\in\Sigma_0$ and $0\leq r\leq s$ then we just need to note that $\mathcal{R}(\textsl{B\/}_{Y^c}^{\theta_r a})$, where $Y=t(\val^\eval(a))$, is a closure rule. As the branch is exhausted yet open, we must have $Y=X$.
\item If $\X{}\iota(\theta_{r_1}{\odot}_1,\theta_{r_2}{\odot}_2)$ for some intersection formula then we just need to note that $\mathcal{R}(\textsl{G}(\iota(\theta_{r_1}{\odot}_1,\theta_{r_2}{\odot}_2)))$ is a closure rule. Again, as the branch is exhausted yet open, we have that $t(\val(\iota(\theta_{r_1}{\odot}_1,\theta_{r_2}{\odot}_2)))=X$.
\end{enumerate}
For the induction step, let $\X{}\theta_r{\odot}(\varphi_1,\dots,\varphi_k)$ be a proper $\theta_r{\odot}$-formula appearing in the branch, where $0\leq r\leq s$, and ${\odot}\in\Sigma_k$ for $k\neq 0$. As the branch is exhausted, all the formulas in one of the conclusions of $\mathcal{R}(\textsl{B\/}_{X}^{\theta_r {\odot}})$ are also in the branch. By definition of $\textsl{B\/}_{X}^{\theta_r {\odot}}$, all these formulas are of the form $\X{ti}\theta_t(\varphi_i)$ with $0\leq t\leq s$ and $1\leq i\leq k$. Clearly, $\mathsf{cplx}(\theta_t(\varphi_i))<\mathsf{cplx}(\theta_r{\odot}(\varphi_1,\dots,\varphi_k))$. 
Therefore, by induction hypothesis, $\val^\eval$ satisfies all the formulas in the given branch of the conclusion of the rule $\mathcal{R}(\textsl{B\/}_{X}^{\theta_r {\odot}})$.  
Recall that the right-hand sides of~$\textsl{B\/}_{X}^{\theta_r {\odot}}$ and~$\textsl{B\/}_{X^c}^{\theta_r {\odot}}$ are disjoint.
So, $\val^\eval$ cannot satisfy any of the disjuncts in the right-hand side of $\textsl{B\/}_{X^c}^{\theta_r {\odot}}$.
It follows that $\val^\eval$ falsifies $\textsl{B\/}_{X^c}^{\theta_r {\odot}}$, meaning that $\val^\eval(\theta_r{\odot}(\varphi_1,\dots,\varphi_k))\neq X^c$, thus~$\val^\eval$ satisfies $\X{}\theta_r{\odot}(\varphi_1,\dots,\varphi_k)$.
\qed
\end{proof}

\begin{remark}
In view of Remark~\ref{complementaryBHV}, the reader with a proof-theoretic eye will have noticed that our tableau rules obtained as counterparts of $B$-statements are \textit{invertible}.
This feature, interesting in itself and related to the desirable reduction of nondeterminism in proof-search, has played the expected role in our completeness proof, above.
\end{remark}

\begin{example}[An infinite tableau.]
\label{infinite-tableau}
By themselves, our tableau systems do not ensure termination. Indeed, we need to go beyond the usual subformula property in order to define a terminating proof procedure. Fig.~\ref{tab-infinite} depicts a simple example of an infinite proof in the system ${\mathcal{T}}(\textrm{\L$_3$},\overline{\theta})$ of Ex.~\ref{systems-and-tableaux}.  
In this tableau, we refer to rule $\mathcal{R}(\textsl{B}^{\supset}_T)$ as~\ding{172} and refer to rule $\mathcal{R}(\textsl{B}^{\supset}_F)$ as~\ding{173}, and we prune the derivation tree in order to concentrate only on the second branch of each rule application.
It should be clear that the illustrated unwise choice of rules, alternating~\ding{172} and~\ding{173}, will indeed lead to a nonterminating tableau.
\end{example}
\begin{figure}[htbp]\scriptsize
\begin{center}
\scalebox{.8}{
\begin{tikzpicture}[sibling distance=1.14em]
\hspace{-5mm}
\tikzset{every tree node/.style={align=center,anchor=north}}
\tikzset{level 1/.style={level distance=40pt}}
\tikzset{level 2/.style={level distance=40pt}}
\tikzset{level 3/.style={level distance=40pt}}
\tikzset{level 4/.style={level distance=40pt}}
\tikzset{level 5/.style={level distance=40pt}}
 \Tree[.{$\Ts \theta(p)$ \\[1mm] $\textsl{i.e. }\Ts \neg p\supset p$}
        [.{$\Fs \neg p$ \\[1mm] $\Ts \theta(p)$ \\[1mm] $\vdots$}
        ]
	  \edge node[auto=left]{\ding{172}};
        [.{$\Fs \theta(\neg p)$ \\[1mm] $\textsl{i.e. }\Fs \neg\neg p\supset \neg p$}
          [.{$\Ts \neg\neg p$ \\[1mm] $\Fs \neg p$ \\[1mm] $\vdots$}
          ]
	    \edge node[auto=left]{\ding{173}};
          [.{$\Ts \theta(\neg \neg p)$\\[1mm] $\textsl{i.e. }\Ts \neg\neg\neg p\supset \neg \neg p$}
            [.{$\vdots$}
            ]
	      \edge node[auto=left]{\ding{172}};
            [.{$\Fs \theta(\neg\neg\neg p)$ \\[1mm] $\textsl{i.e. }\Fs \neg\neg\neg\neg p\supset \neg\neg\neg p$}
              [.{$\vdots$}
              ]
	        \edge node[auto=left]{\ding{173}};
              [.{$\Ts \theta(\neg \neg\neg \neg p)$\\[1mm] $\textsl{i.e. }\Ts\neg\neg\neg\neg\neg p\supset \neg \neg\neg \neg p$}
                [.{$\vdots$}
                ]
                [.{$\vdots$}
                ]
	          \edge node[auto=left]{\ding{172}};
                [.{$\vdots$}
                ]
              ]
            ]
            [.{$\vdots$}
            ]
          ]
        ]
        [.{$\Ts p$\\[1mm] $\vdots$}
        ]
      ] %
\end{tikzpicture}}
\caption{Infinite tableau for $\T\theta(p)$ in ${\mathcal{T}}(\textrm{\L$_3$},\overline{\theta})$.}\label{tab-infinite}
\end{center}
\end{figure}
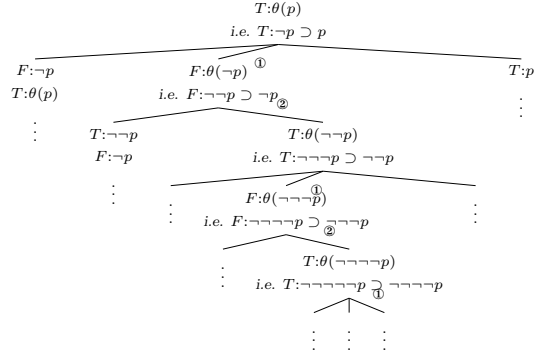

\subsection{Generalized analyticity}\label{subsec:branchcomplete}

\noindent
Though a completely unrestrained choice of rule applications in developing tableaux may be inconclusive, as illustrated in Ex.~\ref{infinite-tableau}, the very proof of the completeness result, in Prop.~\ref{complete}, suggests that we can do much better by wisely choosing the rule to be applied in each case.  To help formulating a suitable strategy, given a labeled formula~$\X{}{\varphi}$ to which a number of different rules, $\mathcal{R}_1$, $\mathcal{R}_2$, \ldots, and $\mathcal{R}_k$, might equally be applied, we will call $\mathcal{R}_j$ the \textsl{most concrete applicable rule} in case it contains the most concrete head matching~$\varphi$.  Then:

\begin{definition}[Analytic proof-strategy]\label{analytic} 
When developing a tableau in ${\mathcal{T}}(\mathcal{L},\overline{\theta})$, first:
\begin{itemize}
\item apply a closure rule, if possible; or else 
\item use $\mathsf{cplx}$ to choose the most complex~$\varphi$ such that $\X{}{\varphi}$ appears in an open branch of the tableau, where $\varphi$ is an analyzable formula whose most concrete applicable rule 
has not yet been applied, and then apply this rule.
%
%
\end{itemize}
\end{definition}

A branch of a tableau is said to be \textsl{analytic} if it is either closed or all applicable rules according to the analytic proof-strategy have already been applied to it. An \textsl{analytic tableau} is a tableau containing only analytic branches.  Globally speaking, a tableau system is called analytic if all the formulas that appear in its branches are proper generalized subformulas of the formulas that appear in the root set, and if there is an analytic proof strategy that guarantees the construction of derivation trees to be a terminating procedure.  We can here prove that:

\begin{proposition}\label{prop:analytic}
Finite analytic tableaux exist for any given finite root set.
\end{proposition}
\begin{proof}
Recall that the lexicographic order on the naturals is well-founded. Let us associate the triple $(i,j,k)\in\mathbb{N}\times\mathbb{N}\times\mathbb{N}$ to each tableau, where~$i$ is the maximum complexity of a formula occurring in an open branch of the tableau whose corresponding most concrete applicable rule has not yet been applied, $j>0$ is the number of formulas in open branches of the tableau that have complexity~$i$ and whose corresponding most concrete applicable rule has not yet been applied, and $k$ is the number of open branches of the tableau. 
We just need to note that each rule application, according to the analytic proof strategy, leads to a tableau whose associated triple $(i',j',k')$ is such that either $i'<i$, or $i'=i$ but $j'<j$, or $i'=i$ and $j'=j$ but $k'<k$. 

Clearly, by applying a closure rule, we get $k'<k$. If all the undeveloped formulas of complexity $i$ are in the branch being closed then $i'<i$. Otherwise, $i'=i$ and thus $j'\leq j$.
By applying a rule $\mathcal{R}(\textsl{B\/}_X^{\theta_r{\odot}})$ to a $\theta_r{\odot}$-proper formula, we either get $i'<i$, or else $i'=i$ with $j'<j$.
\qed
\end{proof}

The following result follows from the proof of Prop.~\ref{complete}.

\begin{proposition}[Completeness by analyticity]\label{analytic-complete}
From every open branch of an analytic tableau one may extract a valuation in $\mathsf{Sem}$ satisfying its root set.
\end{proposition}

\begin{corollary}
For a given finite-valued logic $\mathcal{L}$ with semantics
$\mathsf{Sem}$, we have that
$\gamma_1,\gamma_2,\ldots,\gamma_k\models_\mathsf{Sem}\varphi$ if and only if there is a closed analytic tableau for the root set 
$\{\T\gamma_1,\T\gamma_2,\,\ldots,\,\T\gamma_k,\,\F\varphi\}$. 
Hence, the development of an analytic tableau constitutes a \textit{decision procedure} for $\mathcal{L}$.
\end{corollary}

\begin{example}[Fat tableaux]
\label{fat-tableau}
Given $k\in\mathbb{N}$ and $A\subseteq\{1,\dots,k\}$, let $\varphi_A=(\bigvee_{i=1}^k\varphi_{A,i})$ with $\varphi_{A,i}=p_i$ if $i\in A$, and $\varphi_{A,i}=\neg p_i$ if $i\notin A$. For each $k\in\mathbb{N}$, let  $\Phi_k=(\bigwedge_{A\subseteq\{1,\dots,k\}}\varphi_A)$ be the \textsl{$k$-th fat formula} of~\cite{dag:handbook:99}. It is easy to check that all fat formulas are unsatisfiable not only in Classical Logic but also in the logics \L$_n$ and G$_n$, for any $n\in\mathbb{N}$. Fig.~\ref{fat2} shows a closed analytic tableau for $T{:}\Phi_2$ that could have been built not just in a tableau system for Classical Logic, but also in any of the tableau systems ${\mathcal{T}}(\textrm{\L$_n$},\overline{\theta})$ or ${\mathcal{T}}(\textrm{G$^+_n$},\overline{\theta})$. Note that the label \ding{172} indicates two subsequent applications of the rule $\mathcal{R}(\textsl{B}^{\land}_T)$ in the tableau systems ${\mathcal{T}}(\textrm{\L$_3$},\overline{\theta})$ or ${\mathcal{T}}(\textrm{G$^+_4$},\overline{\theta})$ of Ex.~\ref{systems-and-tableaux}.
Similarly, the labels \ding{173} indicate an application of the rule $\mathcal{R}(\textsl{B}^{\lor}_T)$ in the same systems.
The labels \ding{174} indicate the closure of a branch including $\Ts \varphi$ and $\Ts \neg\varphi$ for any formula $\varphi$, namely by using subsequently $\mathcal{R}(\textsl{B}^{\neg}_T)$ and $\mathcal{R}(\textsl{ABS})$.
It is a simple corollary of our soundness and completeness results that similar rules exist in the systems ${\mathcal{T}}(\textrm{\L$_n$},\overline{\theta})$ and ${\mathcal{T}}(\textrm{G$^+_n$},\overline{\theta})$ for arbitrary $n$.

Closed tableaux for other fat formulas can be similarly obtained. However, a tableau for $T{:}\Phi_3$, for instance, has already hundreds of branches. Asymptotically,~\cite{dag:handbook:99} shows that a closed tableau for $T{:}\Phi_k$ has more than $k!$ branches.
\end{example}

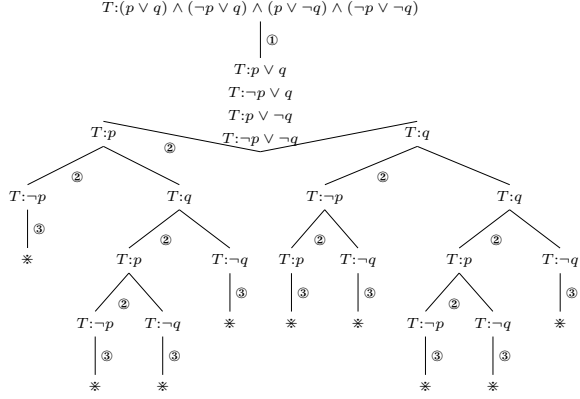
\begin{figure}[htbp]\scriptsize
\begin{center}
\scalebox{.8}{
\begin{tikzpicture}[sibling distance=1.14em]
\hspace{-5mm}
\tikzset{every tree node/.style={align=center,anchor=north}}
\tikzset{level 1/.style={level distance=40pt}}
\tikzset{level 2/.style={level distance=60pt}}
\tikzset{level 4/.style={level distance=40pt}}
\tikzset{level 5/.style={level distance=40pt}}
\Tree[.{$\Ts (p\vee q)\wedge(\neg p\vee q)\wedge(p\vee \neg q)\wedge(\neg p\vee \neg q)$}
       \edge node[auto=left]{\ding{172}};
       [.{$\Ts p\vee q$ \\[1mm] $\Ts \neg p\vee q$ \\[1mm] $\Ts p\vee \neg q$ \\[1mm] $\Ts \neg p\vee \neg q$}
         \edge node[auto=left]{\ding{173}};
         [.{$\Ts p$}
            \edge node[auto=left]{\ding{173}};
           [.{$\Ts \neg p$}
             \edge node[auto=left]{\ding{174}};
             [.$\divideontimes$
             ]
         ]
         [.{$\Ts q$}
        \edge node[auto=left]{\ding{173}};
		[.{$\Ts p$}
		  \edge node[auto=left]{\ding{173}};
		  [.{$\Ts \neg p$}
		    \edge node[auto=left]{\ding{174}};
		    [.$\divideontimes$
		    ]
         ]
         [.{$\Ts \neg q$}
           \edge node[auto=left]{\ding{174}};
           [.$\divideontimes$
           ]
         ]
         ]
         [.{$\Ts \neg q$}
           \edge node[auto=left]{\ding{174}};
           [.$\divideontimes$
           ]
         ]
         ]
         ]
         [.{$\Ts q$}
           \edge node[auto=left]{\ding{173}};
           [.{$\Ts \neg p$}
             \edge node[auto=left]{\ding{173}};
             [.{$\Ts p$}
               \edge node[auto=left]{\ding{174}};
               [.$\divideontimes$
               ]
	         ]
	         [.{$\Ts \neg q$}
	           \edge node[auto=left]{\ding{174}};
	           [.$\divideontimes$
	           ]
	         ]
         ]
         [.{$\Ts q$}
            \edge node[auto=left]{\ding{173}};
            [.{$\Ts p$}
              \edge node[auto=left]{\ding{173}};
              [.{$\Ts \neg p$}
                \edge node[auto=left]{\ding{174}};
                [.$\divideontimes$
                ]
	         ]
	         [.{$\Ts \neg q$}
	           \edge node[auto=left]{\ding{174}};
	           [.$\divideontimes$
	           ]
	         ]
	         ]
	         [.{$\Ts \neg q$}
	           \edge node[auto=left]{\ding{174}};
	           [.$\divideontimes$
	           ]
	         ]
         ]
         ]
       ] 
	 ] %
\end{tikzpicture}}
\end{center}
\caption{A branching closed tableau for $T{:}\Phi_2$.}\label{fat2}	
\end{figure}

\begin{example}[An open tableau]
\label{open-tableau}
Recall that the root of the non-analytic infinite tableau of Ex.~\ref{infinite-tableau}, Fig.~\ref{tab-infinite}, is a single formula, thus leading trivially to an exhausted tableau.
Part of a more interesting exhausted tableau with an open branch in the system ${\mathcal{T}}(\textrm{G$^+_4$},\overline{\theta})$ of Ex.~\ref{systems-and-tableaux} is depicted in Fig.~\ref{tab-open}. The open exhausted branch (the fifth branch from the left) yields a falsifying valuation for $((p\supset\neg p)\supset p)\supset p$ with $\val(p)=\frac{2}{3}$. The three rightmost unfinished branches can all be easily closed. The remaining falsifying valuation for the formula with $\val(p)=\frac{1}{3}$ will be yielded by developing the leftmost unfinished branch of the tableau.  In this tableau we use~\ding{172} to refer to rule $\mathcal{R}(\textsl{B\/}^{\supset}_F)$, \ding{173} to refer to rule $\mathcal{R}(\textsl{B\/}^{\supset}_T)$, \ding{174} to refer to rule $\mathcal{R}(\textsl{ABS\/})$, \ding{175} to refer to rule $\mathcal{R}(\textsl{B\/}^{\neg}_F)$, and \ding{176} to refer to rule $\mathcal{R}(\textsl{U\/}\langle \uparrow,T,T\rangle)$.
\end{example}

\begin{figure}[htbp]\scriptsize
\begin{center}
\scalebox{.8}{
\begin{tikzpicture}[sibling distance=1.14em]
\hspace{-5mm}
\tikzset{every tree node/.style={align=center,anchor=north}}
\tikzset{level 1/.style={level distance=40pt}}
\tikzset{level 2/.style={level distance=60pt}}
\tikzset{level 3/.style={level distance=40pt}}
\tikzset{level 4/.style={level distance=60pt}}
\tikzset{level 5/.style={level distance=40pt}}
 \Tree[.{$\Fs ((p\supset\neg p)\supset p)\supset p$}
	  \edge node[auto=left]{\ding{172}};
        [.{$\Ts (p\supset\neg p)\supset p$ \\[1mm] $\Fs p$}
	  \edge node[auto=left]{\ding{173}};
        [.{$\Ts p$}
	    \edge node[auto=left]{\ding{174}};
          [.{$\divideontimes$}
          ]
        ]
        [.{$\Fs p\supset\neg p$ \\[1mm] $\Ts \theta_2(p)$}
	    \edge node[auto=left]{\ding{172}};
          [.{$\Ts p$ \\[1mm] $\Fs \neg p$}
	      \edge node[auto=left]{\ding{174}};
            [.{$\divideontimes$}
            ]
          ]
          [.{$\Fs \neg p$ \\[1mm] $\Fs \theta_2(\neg p)$}
	      \edge node[auto=left]{\ding{175}};
            [.{$\Ts p$}
	            \edge node[auto=left]{\ding{172}};
	            [.{$\divideontimes$}
	            ]
            ]
            [.{$\Ts  \theta_1(p)$}
	            \edge node[auto=left]{\ding{176}};
	            [.{$\divideontimes$}
	            ]
            ]
            [.{$\Ts  \theta_2(p)$}
            ]
          ]
          [.{$\Ts \theta_1(p)$ \\[1mm] $\Fs \neg p$\\[1mm] $\Fs \theta_1(\neg p)$\\[1mm] $\Fs \theta_2(\neg p)$}
	      \edge node[auto=left]{\ding{176}};
            [.{$\divideontimes$}
            ]
          ]
        ]
        [.{
          \quad$\vdots$\quad}
        ]
        [.{
          \quad$\vdots$\quad}
        ]
        ]
        [.{
           \quad$\vdots$\quad}
        ]
        [.{
          \quad$\vdots$\quad}
        ]
      ] %
\end{tikzpicture}}
\caption{An open exhausted tableau in ${\mathcal{T}}(\textrm{G$^+_4$},\overline{\theta})$.}\label{tab-open}
\end{center}
\end{figure}
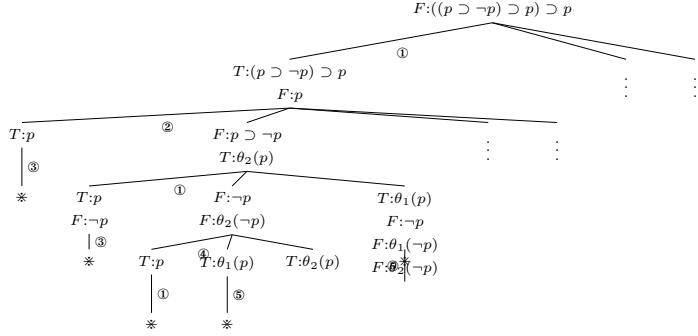

\section{Uniform Analytic Classic-Like Cut-Based Tableaux}\label{CutBased}
\noindent
The tableau systems produced using the recipe in Section~\ref{TableauExtraction} may originate very redundant and highly branching derivation trees, such as some of the tableau proofs pictured above. 
This unpleasant fact is actually a common feature of branching tableau systems, and an extreme case of the undesirable explosion that might originate from accumulating such redundancies is provided by the fat formulas of Ex.~\ref{fat-tableau}, which also show that branching tableaux cannot polynomially simulate the truth-table method. 
The key point here is that while the complexity of truth-tables for finite-valued logics depends only on the number of distinct atomic variables occurring in the formula to be decided, the size of branching tableaux, in the worst cases, might depend essentially on the length of such a formula.

To circumvent such problems, in the case of Classical Logic, D'Agostino and Mondadori~\cite{dag:mon:taming,dag:handbook:99} have introduced the cut-based `KE system' of tableaux, and shown that such system allows for much more efficient tableau proofs; namely, they have proven that KE tableaux \textit{can} polynomially simulate truth tables (while the simulation does not work the other way round) and are therefore in general more efficient than Smullyan's analytic tableaux. The main feature of the system is that it consists of linear rules for the connectives, and a unique branching rule --- the \textit{cut rule}.
We should note that an application of the KE approach to the case of many-valued logics has been proposed in~\cite[Chapter 6.1]{hah:book:94}, where a cut-based version of the sets-as-signs tableau systems is briefly described.

In this section we show how one can adapt the approach of the previous section in order to obtain, in the general case, more efficient cut-based classic-like tableau systems for finite-valued logics, in which the only branching rule is an analytic version of the cut rule. Generalizing D'Agostino and Mondadori's results, we will show that such cut-based systems allow in general to develop tableaux that polynomially simulate $n$-valued truth-tables, thus providing further evidence of their advantages over the branching systems.

\subsection{Linear bivalent statements} \label{subsec:linear-stat}

\noindent 
We start by looking for a way of replacing the $\textsl{B}$-statements of Subsection~\ref{sec:classiclike} by a collection of equivalent \textit{linear} statements --- that is, statements whose right-hand sides are just conjunctions of labeled formulas. Let us here fix an $n$-valued logic~$\mathcal{L}$ with a set of constructors~$\Sigma$ and an appropriate separating sequence $\overline{\theta}=\langle\theta_r\rangle_{r=0}^s$. We are interested in exploiting to our advantage the information carried by labeled formulas such as $\X{}\theta_r({\odot}(\varphi_1,\dots,\varphi_k))$, where $0\leq r\leq s$ and ${\odot}\in\Sigma_k$, $X\in\{F,T\}$ and $\varphi_1,\dots,\varphi_k\in\mathcal{S}$, and on that quest we will be guided by the following questions: given a certain amount of partial information about a given valuation, can we conclude that the labeled formula is satisfied?\ and if so, how much more information about that valuation can we gather?

\begin{definition}\label{partialvec}
A \textsl{vector of partial binary prints} is a finite sequence
$\overline{\overline{Y}}=\langle \overline{Y_1},\dots,\overline{Y_k}\rangle$ where $\overline{Y_i}$ is a partial binary print, for each $1\leq i\leq k$.  
We add an extra dimension to the definition of $\dom$ (Def.~\ref{partial}) and use here $\dom(\overline{\overline{Y}})$ to denote the set $\{\langle i,r\rangle:1\leq i\leq k\mbox{ and }r\in\dom(\overline{Y_i})\}$ (the context of usage will always take care that such overload of $\dom$ does not get us into trouble).
Given two $k$-long vectors~$\overline{\overline{Y}}$ and~$\overline{\overline{Z}}$, we say that~$\overline{\overline{Y}}$ \textsl{extends}~$\overline{\overline{Z}}$ if $\overline{Y_i}$ extends $\overline{Z_i}$ for all $1\leq i\leq k$.
\end{definition}

Let $\overline{\overline{Y}}=\langle \overline{Y_1},\dots,\overline{Y_k}\rangle$ be a vector of partial binary prints, and recall from Section~\ref{TFvsBiv} the definition of $R_X^{\theta_r{\odot}}$ as $\{\overline{x}\in(\mathcal{V}_n)^k: t(\widehat{\theta_r}(\widehat{{\odot}}(\overline{x})))=X\}$.  
Given $\overline{x}=\langle x_1,\ldots,x_k \rangle \in(\mathcal{V}_n)^k$, we use $\overline{\theta}(\overline{x})$ to denote the vector of binary prints $\langle \overline{\theta}(x_1),\ldots,\overline{\theta}(x_k) \rangle$ and, by slightly abusing notation, $R_X^{\theta_r{\odot}}\cap\overline{\overline{Y}}$ to denote the set $\{\overline{x}\in R_X^{\theta_r{\odot}}:\overline{\theta}(\overline{x})\textrm{ extends }\overline{\overline{Y}}\}$.
We say that $\X{}\theta_r({\odot}(\varphi_1,\dots,\varphi_k))$ is \textsl{satisfied by $\overline{\overline{Y}}$} if $R^{\theta_r{\odot}}_X\cap\overline{\overline{Y}}\neq\varnothing$.
We can characterize the situation in which $\X{}\theta_r({\odot}(\varphi_1,\dots,\varphi_k))$ is not satisfied by $\overline{\overline{Y}}$ using the following linear statement:
\begin{equation}
{\X{}\theta_r({\odot}(\varphi_1,\dots,\varphi_k)) {\;\;\&\;\;} V(\varphi_1,\dots,\varphi_k;\overline{\overline{Y}})} {\quad\Longrightarrow\quad} {\divideontimes}. \tag{$\textsl{L}_{X \overline{\overline{Y}}}^{\theta_r{\odot}}$}\label{eq:behavior-L1}
\end{equation}
When $\X{}\theta_r({\odot}(\varphi_1,\dots,\varphi_k))$ is indeed satisfied by $\overline{\overline{Y}}$ we must look for additional information. Let $\overline{\overline{Z}}$ be another $k$-long vector such that $\dom(\overline{\overline{Y}})\cap \dom(\overline{\overline{Z}})=\varnothing$. We say that $\overline{\overline{Y}}$ \textsl{entails} $\overline{\overline{Z}}$ \textsl{with respect to} $\X{}\theta_r({\odot}(\varphi_1,\dots,\varphi_k))$ just in case $R_X^{\theta_r{\odot}}\cap\overline{\overline{Y}}\subseteq R_X^{\theta_r{\odot}}\cap\overline{\overline{Z}}$.
Note that if $\overline{\overline{Y}}$ entails both $\overline{\overline{Z}}$ and $\overline{\overline{W}}$ with respect to $\X{}\theta_r({\odot}(\varphi_1,\dots,\varphi_k))$ then $\overline{\overline{Z}}$ and $\overline{\overline{W}}$ are necessarily \textsl{compatible}, in the sense that, for each $1\leq i\leq k$ and $0\leq t\leq s$, if $Z_{it}\neq{\uparrow}$ and $W_{it}\neq{\uparrow}$ then $Z_{it}=W_{it}$. This means that $\overline{\overline{Z}}$ and $\overline{\overline{W}}$ may be merged into a single vector that extends both and that is also entailed 
by $\overline{\overline{Y}}$ with respect to $\X{}\theta_r({\odot}(\varphi_1,\dots,\varphi_k))$. Hence, there is a \emph{largest} (i.e., most defined) vector entailed by $\overline{\overline{Y}}$ with respect to such given labeled formula, which we will denote by $\overline{\overline{M}}(R_X^{\theta_r{\odot}},\overline{\overline{Y}})$. 
Clearly, $\overline{\overline{M}}(R_X^{\theta_r{\odot}},\overline{\overline{Y}})$ contains all the new information (not in $\overline{\overline{Y}}$) that is invariant among the elements of $R_X^{\theta_r{\odot}}\cap\overline{\overline{Y}}$. Thus, we have that $\overline{\overline{M}}(R_X^{\theta_r{\odot}},\overline{\overline{Y}})_{it}={\uparrow}$
if $\langle i,t\rangle\in\dom(\overline{\overline{Y}})$, or if there exist $\overline{u},\overline{v}\in R_X^{\theta_r{\odot}}\cap\overline{\overline{Y}}$ such that $t(\widehat{\theta_t}(u_i))\neq t(\widehat{\theta_t}(v_i))$. Otherwise $\overline{\overline{M}}(R_X^{\theta_r{\odot}},\overline{\overline{Y}})_{it}=t(\widehat{\theta_t}(x_i))$ for any $\overline{x}\in R_X^{\theta_r{\odot}}\cap\overline{\overline{Y}}$.

We can finally describe the information extractible from $\X{}\theta_r({\odot}(\varphi_1,\dots,\varphi_k))$ satisfied by $\overline{\overline{Y}}$ by means of the following linear statement:
\begin{equation}
\hspace{-7mm}
{\X{}\theta_r({\odot}(\varphi_1,\dots,\varphi_k)) {\,\&\,} V(\varphi_1,\dots,\varphi_k;\overline{\overline{Y}})} {\,\Longrightarrow\,} V(\varphi_1,\dots,\varphi_k;\overline{\overline{M}}(R_X^{\theta_r{\odot}},\overline{\overline{Y}})).\!\!\!\!\!\!\!\!\!\! \tag{$\textsl{L}_{X \overline{\overline{Y}}}^{\theta_r{\odot}}$}\label{eq:behavior-L2}
\end{equation}

\begin{lemma}\label{okpartial-cutbased}
Let $b:\mathcal{S}\longrightarrow\{F,T\}$ be a bivaluation, $X\in\{F,T\}$, $0\leq r\leq s$, and ${\odot}\in\Sigma_k$. Then, $b$ satisfies $\textsl{B}_{X}^{\theta_r{\odot}}$ if and only if, for every $k$-long vector~$\overline{\overline{Y}}$ of partial binary prints, $b$ satisfies $\textsl{L}_{X \overline{\overline{Y}}}^{\theta_r{\odot}}$.
\end{lemma}
\begin{proof}
Given a bivaluation $b$ let us denote by $\overline{\overline{B}}$ the vector of (total) binary prints $\langle \langle b(\theta_t(\varphi_1))\rangle_{t=0}^s,\dots, \langle b(\theta_t(\varphi_k))\rangle_{t=0}^s\rangle$ induced by $b$.	
		
Let $b$ satisfy $\textsl{B}_{X}^{\theta_r{\odot}}$, and consider a vector $\overline{\overline{Y}}$. If $b$ satisfies the left-hand side of $\textsl{L}_{X \overline{\overline{Y}}}^{\theta_r{\odot}}$ then, in particular, $b$ satisfies the left-hand side of $\textsl{B}_{X}^{\theta_r{\odot}}$ and therefore also one of the disjuncts on the right, i.e., $b$ satisfies $\textsl{V}(\varphi_1,\dots,\varphi_k \, ;\overline{\theta}(\overline{x}))$ for some $\overline{x}\in R_X^{\theta_r{\odot}}$. Note that this means precisely that  $\overline{\overline{B}}=\overline{\theta}(\overline{x})$. As we also know that $\overline{\overline{B}}$ extends $\overline{\overline{Y}}$
we then get $\overline{x}\in R_X^{\theta_r{\odot}}\cap\overline{\overline{Y}}\subseteq R_X^{\theta_r{\odot}}\cap\overline{\overline{M}}(R_X^{\theta_r{\odot}},\overline{\overline{Y}})$.
In that case, $\overline{\overline{B}}$ also extends $\overline{\overline{M}}(R_X^{\theta_r{\odot}},\overline{\overline{Y}})$ and $b$ satisfies thus the right-hand side of $\textsl{L}_{X \overline{\overline{Y}}}^{\theta_r{\odot}}$.

For the other direction, assume that $b$ satisfies $\textsl{L}_{X \overline{\overline{Y}}}^{\theta_r{\odot}}$ for every $k$-long vector~$\overline{\overline{Y}}$. If $b$ also satisfies the left-hand side of $\textsl{B}_{X}^{\theta_r{\odot}}$ then $b$ necessarily satisfies the left-hand side of $\textsl{L}_{X \overline{\overline{B}}}^{\theta_r{\odot}}$. Thus, $b$ must also satisfy the right-hand side of $\textsl{L}_{X \overline{\overline{B}}}^{\theta_r{\odot}}$, and this implies that $R_X^{\theta_r{\odot}}\cap\overline{\overline{B}}\neq\varnothing$. Therefore, there exists $\overline{x}\in R_X^{\theta_r{\odot}}$ such that $\overline{\theta}(\overline{x})$ extends $\overline{\overline{B}}$, that is $\overline{\theta}(\overline{x})=\overline{\overline{B}}$, and so $b$ satisfies the disjunct $\textsl{V}(\varphi_1,\dots,\varphi_k \, ;\overline{\theta}(\overline{x}))$ on the right-hand side of $\textsl{B}_{X}^{\theta_r{\odot}}$.
\qed
\end{proof}

\begin{remark}\label{rem:lesscuts}
One should note that, according to the above proof, it would suffice to consider the statements $\textsl{L}_{X \overline{\overline{Y}}}^{\theta_r{\odot}}$ where $\overline{\overline{Y}}$ is a total vector. However, we shall see at the end of Subsection~\ref{Complexity} that the additional statements will allow, within our cut-based tableau systems, for the development of even more economical derivation trees.
\end{remark}

The following definition should be contrasted to the earlier Def.~\ref{tabstatements}.

\begin{definition}\label{bivstatements-cutbased}
The set $\mathcal{B}_{\mathcal{T}}^{\textsf{lin}}(\mathcal{L},\overline{\theta})$ of \textsl{linear bivalent statements} associated to~$\mathcal{L}$, fixed a separating sequence $\overline{\theta}=\langle\theta_r\rangle_{r=0}^s$, is formed by all instances of:
\begin{itemize}
\item $\textsl{U}\langle\overline{Y}\rangle$, for each minimal unobtainable partial binary print $\overline{Y}$;  
\item $\textsl{L}_{X\overline{\overline{Y}}}^{\theta_r{\odot}}$, for each $X\in\{F,T\}$, $0\leq r\leq s$, ${\odot}\in\Sigma_k$ and each $k$-long vector~$\overline{\overline{Y}}$ of partial binary prints;\hfill ($\textsl{L}$-statements)
\item $\textsl{G}(\iota(\theta_{r_1}{\odot}_1,\theta_{r_2}{\odot}_2))$, for each intersection formula $\iota(\theta_{r_1}{\odot}_1,\theta_{r_2}{\odot}_2)$ with $0\leq r_1,r_2\leq s$ and ${\odot}_1,{\odot}_2\in\Sigma$.
\end{itemize}
\end{definition}

In what follows, recall that $\mathsf{Sem}_2$ is the bivalent semantics of the $n$-valued logic $\mathcal{L}$ produced by Def.~\ref{dyadic}.

\begin{proposition}\label{ok-cutbased}
$\mathsf{Sem}_2$ is the set of all bivaluations that satisfy $\mathcal{B}_{\mathcal{T}}^{\textsf{lin}}(\mathcal{L},\overline{\theta})$.
\end{proposition}
\begin{proof}
Immediate from Prop.~\ref{ok} and Lemma~\ref{okpartial-cutbased}.\qed
\end{proof}

\begin{example}[Linear characterization of $\neg$ in \L$_3$]\label{lukascutsneg}
Back to the example of \L$_3$, separated by $\overline{\theta}=\langle \mathsf{id},\theta\rangle$, where $\theta=\lambda p.(\neg p\supset p)$, the bivalent statements $\textsl{L}_{F\overline{\overline{Y}}}^{\neg}$ are shown in Table~\ref{neg:L3}.
\begin{table}[htbp]\small
	\begin{center}
\scalebox{.63}{
\begin{tabular}{|c|lcl|}
\hline
($\textsl{L}_{F \langle\uparrow\uparrow\rangle}^{\neg}$) & $\F{\neg\varphi}$  & $\Longrightarrow$ & $\T{\theta(\varphi)}$\\
\hline
($\textsl{L}_{F \langle F\uparrow\rangle}^{\neg}$) & $\F{\neg\varphi} \quad \& \quad \F{\varphi}$ & $\Longrightarrow$ &  
           $\T{\theta(\varphi)}$\\
\hline
($\textsl{L}_{F \langle T\uparrow\rangle}^{\neg}$) & $\F{\neg\varphi} \quad \& \quad \T{\varphi}$ & $\Longrightarrow$ &  
           $\T{\theta(\varphi)}$\\
\hline
($\textsl{L}_{F \langle\uparrow F\rangle}^{\neg}$) & $\F{\neg\varphi} \quad \& \quad \F{\theta(\varphi)}$ & $\Longrightarrow$ &  
           $\divideontimes$\\
\hline
($\textsl{L}_{F \langle\uparrow T\rangle}^{\neg}$) & $\F{\neg\varphi} \quad \& \quad \T{\theta(\varphi)}$ & $\Longrightarrow$ &  
           $\top$\\
\hline
($\textsl{L}_{F \langle FF\rangle}^{\neg}$) & $\F{\neg\varphi} \quad \& \quad \F{\varphi} \quad \& \quad \F{\theta(\varphi)}$ & $\Longrightarrow$ &  
           $\divideontimes$\\
\hline
($\textsl{L}_{F \langle FT\rangle}^{\neg}$) & $\F{\neg\varphi} \quad \& \quad \F{\varphi} \quad \& \quad \T{\theta(\varphi)}$ & $\Longrightarrow$ &  
           $\top$\\
\hline
($\textsl{L}_{F \langle TF\rangle}^{\neg}$) & $\F{\neg\varphi} \quad \& \quad \T{\varphi} \quad \& \quad \F{\theta(\varphi)}$ & $\Longrightarrow$ &  
           $\divideontimes$\\
\hline
($\textsl{L}_{F \langle TT\rangle}^{\neg}$) & $\F{\neg\varphi} \quad \& \quad \T{\varphi} \quad \& \quad \T{\theta(\varphi)}$ & $\Longrightarrow$ &  
           $\top$\\
\hline
\end{tabular}
}
\end{center}
\caption{Some $\mathcal{B}_{\mathcal{T}}^{\textsf{lin}}(\textsl{\L$_3$},\langle \mathsf{id},\lambda p.(\neg p\supset p)\rangle)$ statements for $\neg$.}
\label{neg:L3}
\end{table}
\end{example}

\begin{remark}\label{cutstream}
The set of linear statements associated to a logic may often be substantially simplified, again with no danger of spoiling the result of Prop.~\ref{ok-cutbased}, nor any of the subsequent results. The rationale is not to simplify each rule \textit{per se}, but the collection of all statements $\textsl{L}_{X\overline{\overline{Y}}}^{\theta_r{\odot}}$ for fixed $X\in\{F,T\}$, $0\leq r\leq s$ and ${\odot}\in\Sigma$. Indeed, such rules may contain a lot of redundancies. Consider for instance the rule $\textsl{L}_{F \langle F\uparrow\rangle}^{\neg}$ above. Clearly, it is implied by the rule $\textsl{L}_{F \langle \uparrow\uparrow\rangle}^{\neg}$, as no new information is obtained by adding $\F{\varphi}$. Hence, $\textsl{L}_{F \langle F\uparrow\rangle}^{\neg}$ may be safely eliminated. Pick now the rule $\textsl{L}_{F \langle \uparrow F\rangle}^{\neg}$ above. Note that the right-hand side of $\textsl{L}_{F \langle \uparrow\uparrow\rangle}^{\neg}$ contradicts $\F{\theta(\varphi)}$. Again, rule $\textsl{L}_{F \langle \uparrow F\rangle}^{\neg}$ may be safely dispensed with.

In general, a rule $\textsl{L}_{X\overline{\overline{Z}}}^{\theta_r{\odot}}$ may be eliminated whenever there exists a distinct vector $\overline{\overline{Y}}$ such that $\overline{\overline{Z}}$ extends $\overline{\overline{Y}}$ and $\overline{\overline{M}}(R_X^{\theta_r{\odot}},\overline{\overline{Y}})$ extends $\overline{\overline{M}}(R_X^{\theta_r{\odot}},\overline{\overline{Z}})$.
Further, a rule $\textsl{L}_{X\overline{\overline{Z}}}^{\theta_r{\odot}}$ where $\X{}\theta_r({\odot}(\varphi_1,\dots,\varphi_k))$ is not satisfied by $\overline{\overline{Z}}$ may be eliminated when there exists a distinct vector $\overline{\overline{Y}}$ such that $\overline{\overline{Z}}$ extends $\overline{\overline{Y}}$ but $\overline{\overline{Z}}$ is incompatible with $\overline{\overline{M}}(R_X^{\theta_r{\odot}},\overline{\overline{Y}})$.
Such a simplification strategy may be systematically applied to reach a streamlined (shorter but equivalent) version of the set of linear statements. Again, we should note that none of the results in this paper depends on (or is affected by) performing such a simplification.

Back to the example in Table~\ref{neg:L3}, it is easy to check that the only strictly necessary rule is $\textsl{L}_{F \langle \uparrow\uparrow\rangle}^{\neg}$.
\end{remark}

\begin{example}[Linear characterization of \L$_3$]\label{lukascuts}
The (streamlined) set of linear bivalent statements characterizing \L$_3$, separated by $\overline{\theta}=\langle \mathsf{id},\theta\rangle$ where $\theta=\lambda p.(\neg p\supset p)$, includes $\textsl{U}\langle T,F\rangle$ plus the statements shown in Table~\ref{linbiv:L3}.
\end{example}

\begin{table}[bp]\small
	\begin{center}
\scalebox{.63}{
\begin{tabular}{|c|lcl|}
	\hline
	  ($\textsl{L}_{F \langle \uparrow\uparrow\rangle}^{\neg}$) & $\F{\neg\varphi}$  & $\Longrightarrow$ & $\T{\theta(\varphi)}$\\
\hline
	  ($\textsl{L}_{T \langle \uparrow\uparrow\rangle}^{\neg}$) & $\T{\neg\varphi}$ & $\Longrightarrow$ & $\F{\varphi} \quad \& \quad \F{\C(\varphi)}$\\
\hline
	  ($\textsl{L}_{F \langle\uparrow\uparrow\rangle}^{\theta\neg}$) & $\F{\C(\neg\varphi)}$ & $\Longrightarrow$ & $\T{\varphi} \quad \& \quad \T{\C(\varphi)}$\\
\hline	  
	  ($\textsl{L}_{T \langle \uparrow\uparrow\rangle}^{\theta\neg}$) & $\T{\C(\neg\varphi)}$ & $\Longrightarrow$ & $\F{\varphi}$\\
	\hline
	  ($\textsl{L}_{F \langle\uparrow\uparrow\uparrow\uparrow\rangle}^{\supset}$) & $\F{\varphi \supset \psi}$ & $\Longrightarrow$ & $\T{\C(\varphi)} \quad \& \quad \F{\psi}$\\
\hline
	  ($\textsl{L}_{F \langle\uparrow\uparrow\uparrow T\rangle}^{\supset}$) & $\F{\varphi \supset \psi} \quad \& \quad \T{\C(\psi)}$ & $\Longrightarrow$ & $\T{\varphi} \quad \& \quad \T{\C(\varphi)} \quad \& \quad \F{\psi}$\\
	  \hline\
	  ($\textsl{L}_{F \langle F\uparrow\uparrow\uparrow\rangle}^{\supset}$) & $\F{\varphi \supset \psi} \quad \& \quad \F{\varphi}$ & $\Longrightarrow$ & $\T{\C(\varphi)}\quad \& \quad \F{\psi}\quad \& \quad \F{\C(\psi)}$\\
	  \hline
	  ($\textsl{L}_{T \langle \uparrow\uparrow\uparrow F\rangle}^{\supset}$) & $\T{\varphi \supset \psi} \quad \& \quad \F{\C(\psi)}$ & $\Longrightarrow$ & $\F{\varphi}\quad \& \quad \F{\C(\varphi)}\quad \& \quad \F{\psi}$\\ 
	  \hline
	  ($\textsl{L}_{T \langle \uparrow\uparrow F \uparrow\rangle}^{\supset}$) & $\T{\varphi \supset \psi} \quad \& \quad \F{\psi}$ & $\Longrightarrow$ & $\F{\varphi}$\\ 
	  \hline
	  ($\textsl{L}_{T \langle \uparrow\uparrow T \uparrow\rangle}^{\supset}$) & $\T{\varphi \supset \psi} \quad \& \quad \T{\psi}$ & $\Longrightarrow$ & $\T{\C(\psi)}$\\
	  \hline
	  ($\textsl{L}_{T \langle \uparrow F \uparrow\uparrow\rangle}^{\supset}$) & $\T{\varphi \supset \psi} \quad \& \quad \F{\C(\varphi)}$ & $\Longrightarrow$ & $\F{\varphi}$\\
	  \hline
	  ($\textsl{L}_{T \langle \uparrow F T\uparrow\rangle}^{\supset}$) & $\T{\varphi \supset \psi} \quad \& \quad \F{\C(\varphi)} \quad \& \quad \T{\psi}$ & $\Longrightarrow$ & $\F{\varphi} \quad \& \quad \T{\C(\psi)}$\\
	  \hline
	  ($\textsl{L}_{T \langle \uparrow T \uparrow\uparrow\rangle}^{\supset}$) & $\T{\varphi \supset \psi} \quad \& \quad \T{\C(\varphi)}$ & $\Longrightarrow$ & $\T{\C(\psi)}$\\
\hline	  
	  ($\textsl{L}_{T \langle \uparrow T F \uparrow\rangle}^{\supset}$) & $\T{\varphi \supset \psi} \quad \& \quad \T{\C(\varphi)} \quad \& \quad \F{\psi}$ & $\Longrightarrow$ & $\F{\varphi} \quad \& \quad \T{\C(\psi)}$\\
\hline	  
	  ($\textsl{L}_{T \langle T\uparrow\uparrow\uparrow\rangle}^{\supset}$) & $\T{\varphi \supset \psi} \quad \& \quad \T{\varphi}$ & $\Longrightarrow$ & $\T{\C(\varphi)} \quad \& \quad \T{\psi} \quad \& \quad \T{\C(\psi)}$\\
	\hline
	  ($\textsl{L}_{F \langle \uparrow\uparrow\uparrow\uparrow\rangle}^{\theta\supset}$) & $\F{\C(\varphi \supset \psi)}$ & $\Longrightarrow$ & $\T{\varphi} \quad \& \quad \T{\C(\varphi)} \quad \& \quad \F{\psi} \quad \& \quad \F{\C(\psi)}$\\
\hline	  
	  ($\textsl{L}_{T \langle \uparrow\uparrow\uparrow F\rangle}^{\theta\supset}$) & $\T{\C(\varphi \supset \psi)} \quad \& \quad \F{\C(\psi)}$ & $\Longrightarrow$ & $\F{\varphi} \quad \& \quad \F{\psi}$\\
\hline	  
	  ($\textsl{L}_{T \langle \uparrow\uparrow T \uparrow\rangle}^{\theta\supset}$) & $\T{\C(\varphi \supset \psi)} \quad \& \quad \T{\psi}$ & $\Longrightarrow$ & $\T{\C(\psi)}$\\
\hline	  
	  ($\textsl{L}_{T \langle \uparrow F \uparrow\uparrow\rangle}^{\theta\supset}$) & $\T{\C(\varphi \supset \psi)} \quad \& \quad \F{\C(\varphi)}$ & $\Longrightarrow$ & $\F{\varphi}$\\
\hline	  
	  ($\textsl{L}_{T \langle \uparrow F T \uparrow\rangle}^{\theta\supset}$) & $\T{\C(\varphi \supset \psi)} \quad \& \quad \F{\C(\varphi)} \quad \& \quad \T{\psi}$ & $\Longrightarrow$ & $\F{\varphi} \quad \& \quad \T{\C(\psi)}$\\
\hline	  
	  ($\textsl{L}_{T \langle T\uparrow\uparrow\uparrow\rangle}^{\theta\supset}$) & $\T{\C(\varphi \supset \psi)} \quad \& \quad \T{\varphi}$ & $\Longrightarrow$ & $\T{\C(\varphi)} \quad \& \quad \T{\C(\psi)}$\\
\hline
	  ($\textsl{L}_{F \langle \uparrow\uparrow\uparrow\uparrow\rangle}^{\vee}$) & $\F{\varphi \vee \psi}$ & $\Longrightarrow$ & $\F{\varphi} \quad \& \quad \F{\psi}$\\
\hline	  
	  ($\textsl{L}_{T \langle \uparrow\uparrow\uparrow F\rangle}^{\vee}$) & $\T{\varphi \vee \psi} \quad \& \quad \F{\C(\psi)}$ & $\Longrightarrow$ & $\T{\varphi} \quad \& \quad \T{\C(\varphi)} \quad \& \quad \F{\psi}$\\
\hline	  
	  ($\textsl{L}_{T \langle \uparrow\uparrow F \uparrow\rangle}^{\vee}$) & $\T{\varphi \vee \psi} \quad \& \quad \F{\psi}$ & $\Longrightarrow$ & $\T{\varphi} \quad \& \quad \T{\C(\varphi)}$\\
\hline	  
	  ($\textsl{L}_{T \langle \uparrow\uparrow T \uparrow\rangle}^{\vee}$) & $\T{\varphi \vee \psi} \quad \& \quad \T{\psi}$ & $\Longrightarrow$ & $\T{\C(\psi)}$\\
\hline	  
	  ($\textsl{L}_{T \langle \uparrow F \uparrow \uparrow\rangle}^{\vee}$) & $\T{\varphi \vee \psi} \quad \& \quad \F{\C(\varphi)}$ & $\Longrightarrow$ & $\F{\varphi} \quad \& \quad \T{\psi} \quad \& \quad \T{\C(\psi)}$\\
\hline	  
	  ($\textsl{L}_{T \langle F \uparrow\uparrow \uparrow\rangle}^{\vee}$) & $\T{\varphi \vee \psi} \quad \& \quad \F{\varphi}$ & $\Longrightarrow$ & $\T{\psi} \quad \& \quad \T{\C(\psi)}$\\
\hline	  
	  ($\textsl{L}_{T \langle T \uparrow\uparrow \uparrow\rangle}^{\vee}$) & $\T{\varphi \vee \psi} \quad \& \quad \T{\varphi}$ & $\Longrightarrow$ & $\T{\C(\varphi)}$\\
\hline	  
	  ($\textsl{L}_{T \langle T \uparrow T \uparrow\rangle}^{\vee}$) & $\T{\varphi \vee \psi} \quad \& \quad \T{\varphi} \quad \& \quad \T{\psi}$ & $\Longrightarrow$ & $\T{\C(\varphi)} \quad \& \quad \T{\C(\psi)}$\\
\hline	  
	  ($\textsl{L}_{F \langle \uparrow\uparrow \uparrow \uparrow\rangle}^{\C\vee}$) & $\F{\C(\varphi \vee \psi)}$ & $\Longrightarrow$ & $\F{\varphi} \quad \& \quad \F{\C(\varphi)} \quad \& \quad \F{\psi} \quad \& \quad \F{\C(\psi)}$\\
\hline
	  ($\textsl{L}_{T \langle \uparrow\uparrow \uparrow F \rangle}^{\C\vee}$) & $\T{\C(\varphi \vee \psi)} \quad \& \quad \F{\C(\psi)}$ & $\Longrightarrow$ & $\T{\C(\varphi)} \quad \& \quad \F{\psi}$\\
\hline
	  ($\textsl{L}_{T \langle \uparrow\uparrow T \uparrow \rangle}^{\C\vee}$) & $\T{\C(\varphi \vee \psi)} \quad \& \quad \T{\psi}$ & $\Longrightarrow$ & $\T{\C(\psi)}$\\
\hline
	  ($\textsl{L}_{T \langle \uparrow F \uparrow \uparrow\rangle}^{\C\vee}$) & $\T{\C(\varphi \vee \psi)} \quad \& \quad \F{\C(\varphi)}$ & $\Longrightarrow$ & $\F{\varphi} \quad \& \quad \T{\C(\psi)}$\\
\hline
	  ($\textsl{L}_{T \langle T \uparrow\uparrow \uparrow\rangle}^{\C\vee}$) & $\T{\C(\varphi \vee \psi)} \quad \& \quad \T{\varphi}$ & $\Longrightarrow$ & $\T{\C(\varphi)}$\\
\hline
	  ($\textsl{L}_{T \langle T \uparrow T \uparrow\rangle}^{\C\vee}$) & $\T{\C(\varphi \vee \psi)} \quad \& \quad \T{\varphi} \quad \& \quad \T{\psi}$ & $\Longrightarrow$ & $\T{\C(\varphi)} \quad \& \quad \T{\C(\psi)}$\\
\hline
	  ($\textsl{L}_{F \langle \uparrow\uparrow \uparrow F \rangle}^{\wedge}$) & $\F{\varphi \wedge \psi} \quad \& \quad \F{\C(\psi)}$ & $\Longrightarrow$ & $\F{\psi}$\\
\hline
	  ($\textsl{L}_{F \langle \uparrow \uparrow T \uparrow \rangle}^{\wedge}$) & $\F{\varphi \wedge \psi} \quad \& \quad \T{\psi}$ & $\Longrightarrow$ & $\F{\varphi} \quad \& \quad \T{\C(\psi)}$\\
\hline
	  ($\textsl{L}_{F \langle \uparrow F \uparrow \uparrow \rangle}^{\wedge}$) & $\F{\varphi \wedge \psi} \quad \& \quad \F{\C(\varphi)}$ & $\Longrightarrow$ & $\F{\varphi}$\\
\hline 
	  ($\textsl{L}_{F \langle \uparrow F \uparrow F \rangle}^{\wedge}$) & $\F{\varphi \wedge \psi} \quad \& \quad \F{\C(\varphi)} \quad \& \quad \F{\C(\psi)}$ & $\Longrightarrow$ & $\F{\varphi} \quad \& \quad \F{\psi}$\\
\hline
	  ($\textsl{L}_{F \langle T \uparrow \uparrow \uparrow \rangle}^{\wedge}$) & $\F{\varphi \wedge \psi} \quad \& \quad \T{\varphi}$ & $\Longrightarrow$ & $\T{\C(\varphi)} \quad \& \quad \F{\psi}$\\
\hline
	  ($\textsl{L}_{T \langle \uparrow \uparrow \uparrow \uparrow \rangle}^{\wedge}$) & $\T{\varphi \wedge \psi}$ & $\Longrightarrow$ & $\T{\varphi} \quad \& \quad \T{\C(\varphi)} \quad \& \quad \T{\psi} \quad \& \quad \T{\C(\psi)}$\\
\hline
	  ($\textsl{L}_{F \langle \uparrow\uparrow \uparrow F\rangle}^{\C\wedge}$) & $\F{\C(\varphi \wedge \psi)} \quad \& \quad \F{\C(\psi)}$ & $\Longrightarrow$ & $\F{\psi}$\\
\hline
	  ($\textsl{L}_{F \langle \uparrow\uparrow \uparrow T\rangle}^{\C\wedge}$) & $\F{\C(\varphi \wedge \psi)} \quad \& \quad \T{\C(\psi)}$ & $\Longrightarrow$ & $\F{\varphi} \quad \& \quad \F{\C(\varphi)}$\\
\hline
	  ($\textsl{L}_{F \langle \uparrow\uparrow T \uparrow\rangle}^{\C\wedge}$) & $\F{\C(\varphi \wedge \psi)} \quad \& \quad \T{\psi}$ & $\Longrightarrow$ & $\F{\varphi} \quad \& \quad \F{\C(\varphi)} \quad \& \quad \T{\C(\psi)}$\\
\hline
	  ($\textsl{L}_{F \langle \uparrow F \uparrow \uparrow\rangle}^{\C\wedge}$) & $\F{\C(\varphi \wedge \psi)} \quad \& \quad \F{\C(\varphi)}$ & $\Longrightarrow$ & $\F{\varphi}$\\
\hline
	  ($\textsl{L}_{F \langle \uparrow F \uparrow F\rangle}^{\C\wedge}$) & $\F{\C(\varphi \wedge \psi)} \quad \& \quad \F{\C(\varphi)} \quad \& \quad \F{\C(\psi)}$ & $\Longrightarrow$ & $\F{\varphi} \quad \& \quad \F{\psi}$\\
\hline
	  ($\textsl{L}_{F \langle \uparrow T \uparrow \uparrow\rangle}^{\C\wedge}$) & $\F{\C(\varphi \wedge \psi)} \quad \& \quad \T{\C(\varphi)}$ & $\Longrightarrow$ & $\F{\psi} \quad \& \quad \F{\C(\psi)}$\\
\hline
	  ($\textsl{L}_{F \langle T \uparrow\uparrow \uparrow\rangle}^{\C\wedge}$) & $\F{\C(\varphi \wedge \psi)} \quad \& \quad \T{\varphi}$ & $\Longrightarrow$ & $\T{\C(\varphi)} \quad \& \quad \F{\psi} \quad \& \quad \F{\C(\psi)}$\\
\hline
	  ($\textsl{L}_{T \langle \uparrow \uparrow \uparrow \uparrow\rangle}^{\C\wedge}$) & $\T{\C(\varphi \wedge \psi)}$ & $\Longrightarrow$ & $\T{\C(\varphi)} \quad \& \quad \T{\C(\psi)}$\\
\hline
\end{tabular}
}
\end{center}
\caption{The streamlined $\textsl{L}$-statements in $\mathcal{B}_{\mathcal{T}}^{\textsf{lin}}(\textsl{\L$_3$},\langle \mathsf{id},\lambda p.(\neg p\supset p)\rangle)$.}
\label{linbiv:L3}
\end{table}

\subsection{Cut-based systems} \label{subsec:cut-based-rules}

\noindent
We still miss the basic ingredient of cut-based tableau systems, namely a statement capturing the classical principle of excluded middle (this was carefully discussed under the appellation `Principle of Bivalence' in~\cite{dag:monograph:90}):

\begin{equation}
{\quad\Longrightarrow\quad} \Fs \varphi \; \mid\mid \; \Ts \varphi. \tag{$\textsl{CUT\/}$}\label{eq:cut}
\end{equation}

\begin{definition}\label{tableau-cutbased}
The \textsl{classic-like cut-based tableau system ${\mathcal{T}_{\mathsf{cut}}}(\mathcal{L},\overline{\theta})$ associated to $\mathcal{L}$ (and~$\overline{\theta}$)} is composed of the rule $\mathcal{R}(\textsl{CUT\/})$, the rules $\mathcal{R}(S)$ for $S\in\mathcal{B}_{\mathcal{T}}^{\textsf{lin}}(\mathcal{L},\overline{\theta})$ and $\mathcal{R}(\textsl{ABS\/})$.
In such system, fixed a given branch of a given tableau, and given some formula~$\varphi$, an application of $\cut$ over~$\varphi$ in that branch is called \textsl{analytic} in case~$\varphi$ is a generalized subformula of some formula already occurring in that very branch.
\end{definition}

\begin{example}[A cut-based tableau system for $\lthree$]
A cut-based tableau system ${\mathcal{T}_{\mathsf{cut}}}(\lthree,\overline{\theta})$, for $\overline{\theta}=\langle \mathsf{id},\theta\rangle$, where $\theta=\lambda p.(\neg p\supset p)$, consists of the rules $\mathcal{R}(\textsl{CUT\/})$, $\mathcal{R}(\textsl{ABS\/})$, $\mathcal{R}(\textsl{U}\langle T,F\rangle)$ and a rule $\mathcal{R}(\textsl{L}_{X\overline{\overline{Y}}}^{\theta_r{\odot}})$ for each statement $\textsl{L}_{X\overline{\overline{Y}}}^{\theta_r{\odot}}$ in Table~\ref{linbiv:L3}.
\noindent
An example of a cut-based derivation in this system may be found in Fig.~\ref{fat3}.
\end{example}

It is worth remarking that for Classical Logic we obtain precisely the KE system of~\cite{dag:mon:taming,dag:handbook:99}.

We will now check soundness and completeness of our cut-based tableau system ${\mathcal{T}_{\mathsf{cut}}}(\mathcal{L},\overline{\theta})$. 
In particular, the completeness proof will show that it is possible to restrict the use of the rule $\cut$ to analytic applications only.

\begin{proposition}[Soundness]\label{sound-cut-based}
If an $n$-valued valuation 
satisfies some initial root set of classically-labeled
formulas, then it satisfies all the formulas in some open branch
of any tableau proof that originates from that root set in the system ${\mathcal{T}_{\mathsf{cut}}}(\mathcal{L},\overline{\theta})$.
\end{proposition}
\begin{proof}

Consider an $n$-valued valuation
$\val:\mathbb{S}\longrightarrow\mathbb{V}$ in $\mathsf{Sem}$. We know from Prop.~\ref{ok-cutbased} that such a valuation
satisfies all the linear bivalent statements associated to $\mathcal{L}$. Clearly, $\val$ also satisfies all instances of $\textsl{ABS\/}$ and $\textsl{CUT\/}$. It follows that if $\val$ satisfies the premises of a ${\mathcal{T}_{\mathsf{cut}}}(\mathcal{L},\overline{\theta})$ rule then $\val$ must also satisfy one of the branches in its conclusion. In particular, this means that~$\val$ cannot satisfy the premises of any closure rule.

By iterating the argument above, we conclude that if $\val$ satisfies an initial root set, then there exists a branch of the tableau proof where all the formulas are satisfied by $\val$. Since $\val$ cannot satisfy the premises of a closure rule, such a branch must be open.
\qed
\end{proof}

\begin{proposition}[Completeness]\label{complete-cut-based}
From every open branch of an exhausted tableau derived in the system ${\mathcal{T}_{\mathsf{cut}}}(\mathcal{L},\overline{\theta})$, where we allow only analytic applications of $\cut$, one may extract a valuation in $\mathsf{Sem}$ satisfying its root set.
\end{proposition}
\begin{proof}
The following argument is similar to the one used the proof of Prop.~\ref{complete}.
Given an open branch of an exhausted tableau, we can consider an assignment $\val:\mathcal{A}\longrightarrow\mathcal{V}\in\mathsf{Sem}$ such that, for every $p\in\mathcal{A}$, the binary print
$\overline{\theta}({\eval(p)})=\langle X_r\rangle_{r=0}^s$ agrees with the information available in that branch, i.e.,~$\XX{r}{c}\theta_r(p)$ does not occur in the branch, for each $0\leq r\leq s$.
In Prop.~\ref{complete}, we proved that such an assignment always exists for tableaux derived in ${\mathcal{T}}(\mathcal{L},\overline{\theta})$, by observing that otherwise a closure rule could have been applied, thus contradicting the fact that the branch is open and exhausted. Since each closure rule of ${\mathcal{T}}(\mathcal{L},\overline{\theta})$ is also a rule of ${\mathcal{T}_{\mathsf{cut}}}(\mathcal{L},\overline{\theta})$, the same proof applies here.
\noindent
We need to show that $\val^\eval:\mathbb{S}\longrightarrow\mathbb{V}\in\mathsf{Sem}$, the homomorphic extension in $\mathsf{Sem}$ of the assignment~$\val$, satisfies all the signed formulas in the branch, and consequently also the root set. We proceed by induction on the formula complexity. For the base case, we refer the reader again to the proof of Prop.~\ref{complete}, which applies here without any modification. Now consider the induction step. Let $\X{}\theta_r{\odot}(\varphi_1,\dots,\varphi_k)$ be a proper $\theta_r{\odot}$-formula appearing in the branch, where $0\leq r\leq s$ and ${\odot}\in\Sigma_k$ for $k\neq 0$. 
As the branch is exhausted, all the immediate generalized subformulas of $\X{}\theta_r{\odot}(\varphi_1,\dots,\varphi_k)$ also occur in the branch, i.e.,~for $1 \le i \le k$, $0 \le t \le s$, either $\F{\theta_t(\varphi_i)}$ occurs in the branch or $\T{\theta_t(\varphi_i)}$ occurs in the branch (otherwise an analytic application of $\cut$ would be possible). Let us denote such formulas by $\X{ti}{\theta_t(\varphi_i)}$ and
let $\overline{\overline{Y}} = \langle\langle X_{ti} \rangle_{t=0}^s\rangle_{i=1}^k$ be the corresponding vector of (total) binary prints.
We observe that the right-hand side of $\textsl{L\/}_{X\overline{\overline{Y}}}^{\theta_r {\odot}}$ cannot be $\divideontimes$, otherwise an application of $\reg{\textsl{L\/}_{X\overline{\overline{Y}}}^{\theta_r {\odot}}}$ would close the branch. It follows that the right-hand side of $\textsl{L\/}_{X^c\overline{\overline{Y}}}^{\theta_r {\odot}}$ is necessarily $\divideontimes$. As~$\val^\eval$ satisfies $\textsl{L\/}_{X^c\overline{\overline{Y}}}^{\theta_r {\odot}}$, it must be the case that $\val^\eval$ does not satisfy its left-hand side.  
Clearly, we have $\mathsf{cplx}(\theta_t(\varphi_i))<\mathsf{cplx}(\theta_r{\odot}(\varphi_1,\dots,\varphi_k))$. Therefore, by induction hypothesis, $\val^\eval$ satisfies all the formulas $\X{ti}{\theta_t(\varphi_i)}$. As $\val^\eval$ does not satisfy the left-hand side of $\textsl{L\/}_{X^c\overline{\overline{Y}}}^{\theta_r {\odot}}$, we conclude that $\val^\eval(\theta_r{\odot}(\varphi_1,\dots,\varphi_k))\neq X^c$; thus~$\val^\eval$ satisfies $\X{}\theta_r{\odot}(\varphi_1,\dots,\varphi_k)$.
\qed
\end{proof}

As a lesson to be learned from the previous proof, one might now propose:

\begin{definition}[Analytic proof-strategy for cut-based systems]\label{analytic-cut-based}
When developing a tableau in ${\mathcal{T}_{\mathsf{cut}}}(\mathcal{L},\overline{\theta})$, first:
\begin{itemize}
\item apply a closure rule, if possible; or else 
\item use $\mathsf{cplx}$ to choose the most complex $\varphi$ such that $\X{}{\varphi}$ appears in an open branch of the tableau, where~$\varphi$ is an analyzable proper $\theta_r{\odot}$-formula to which $\textsl{L}$-rules of the form $\mathcal{R}(\textsl{L\/}_{X,\overline{\overline{Y}}}^{\theta_r{\odot}})$ have not yet been applied, and then:
\begin{itemize}
  \item apply $\cut$ to all the immediate generalized subformulas of $\varphi$;
  \item on each branch that develops from that, apply all the $\textsl{L}$-rules $\mathcal{R}(\textsl{L\/}_{X,\overline{\overline{Y}}}^{\theta_r{\odot}})$ that happen to be applicable.
\end{itemize}
\end{itemize}
\end{definition}

As in the cut-free case, also in the cut-based approach an \textsl{analytic tableau} is a tableau containing only analytic branches (those that are either closed or such that all applicable rules according to the analytic proof-strategy have already been applied to them).  The notion of analyticity for a cut-based tableau system extends the concept used in the cut-free case by commanding the exclusive use of \textsl{analytic cuts}, that is, cuts involving generalized subformulas of the formulas occurring in a given branch, following the proof-strategy explained above.  With that in mind, we can now prove that:

\begin{proposition}
Finite analytic tableaux exist in ${\mathcal{T}_{\mathsf{cut}}}(\mathcal{L},\overline{\theta})$ for any given finite root set.
\end{proposition}
\begin{proof}
As in the corresponding proof of Prop.~\ref{prop:analytic} for branching tableaux, we may associate the triple $\langle i,j,k\rangle\in\mathbb{N}\times\mathbb{N}\times\mathbb{N}$ to each tableau, where~$i$ is the maximum complexity of a formula occurring in an open branch of the tableau whose corresponding $\textsl{L}$-rules have not yet been applied, $j>0$ is the number of formulas in open branches of the tableau that have complexity~$i$ and whose corresponding $\textsl{L}$-rules have not yet been applied, and $k$ is the number of open branches of the tableau. 
After each application of a closure rule or after each sequence of applications corresponding to the second option in Def.~\ref{analytic-cut-based}, the procedure leads to a tableau whose associated triple $\langle i',j',k'\rangle$ is such that either $i'<i$, or $i'=i$ but $j'<j$, or $i'=i$ and $j'=j$ but $k'<k$. 

In the application of closure rules, the same arguments used in Prop.~\ref{prop:analytic} go through. 
Otherwise, note that after applying $\cut$ on all the immediate subformulas of an analyzable proper $\theta_r{\odot}$-formula and then applying all the corresponding $\textsl{L\/}$-rules, we either get $i'<i$, or $i'=i$ but $j'<j$.
\qed
\end{proof}

The following result may be proved as in Prop.~\ref{complete-cut-based}.

\begin{proposition}[Completeness by analyticity]\label{analytic-complete-cut-based}
From every open branch of an analytic tableau of ${\mathcal{T}_{\mathsf{cut}}}(\mathcal{L},\overline{\theta})$ one may extract a valuation in $\mathsf{Sem}$ satisfying its root set.
\end{proposition}

\begin{corollary}\label{corol:analytic}
Let $\mathcal{L}$ be a finite-valued logic separated by $\overline{\theta}$.
The analytic tableau development for ${\mathcal{T}_{\mathsf{cut}}}(\mathcal{L},\overline{\theta})$ constitutes a decision procedure for $\mathcal{L}$.
\end{corollary}

\subsection{Proof complexity}\label{Complexity}
\noindent
We already know from Coroll.~\ref{corol:analytic} that cut-analyticity guarantees that the cut-based tableau system ${\mathcal{T}_{\mathsf{cut}}}(\mathcal{L},\overline{\theta})$ may be used as a decision procedure for $\mathcal{L}$. 
Since finite-valued logics are known to be decidable by the `brute force' truth-table method, it would seem interesting to compare the computational complexity of the two methods.
As in the case of the KE system for Classical Logic (see~\cite{dag:monograph:90}), it is expectable that our cut-based tableaux for finite-valued logics fare significantly better than conventional tableaux in terms of proof complexity, and in general not worse than the truth-table method.
We adapt from~\cite{dag:handbook:99} the definition of some typical complexity measures to be used below.

\begin{definition}
  \label{def:complexity-measures}
 The \textsl{size} of a tableau $\pi$, denoted by ${\mid}\pi{\mid}$, is the total number of labeled formulas occurring in~$\pi$.
The \textsl{$\lambda$-complexity} of a tableau $\pi$, denoted by $\lambda(\pi)$, is the number of nodes in $\pi$. 
The \textsl{$\rho$-complexity} of a tableau $\pi$, denoted by $\rho(\pi)$, is the maximum number of labeled formulas in a node of $\pi$.
\end{definition}

Clearly, the following relation holds in general: ${\mid} \pi {\mid} \le \lambda(\pi\cdot \rho(\pi))$. 
Note that in the case of a tableau~$\pi$ developed within ${\mathcal{T}_{\mathsf{cut}}}(\mathcal{L},\overline{\theta})$, the $\rho$-complexity of~$\pi$ is bounded by $\rho(\pi) \le r(s+1)$, where $s+1$ is the length of the separating sequence~$\overline{\theta}$ and~$\mathrm{r}$ is the maximum arity of any connective of~$\mathcal{L}$.

The following result shows that the cut-based tableau systems from Def.~\ref{tableau-cutbased} can polynomially simulate (p\textsl{-simulate}) the truth-table method. We use $\textsf{sz}(\varphi)$ to denote the size of the set $\textsf{sbf}(\varphi)$.

\begin{proposition}
  \label{th:truth-tables-simulation}
Given a valid labeled formula $\X{}\varphi(p_1,\dots,p_k)$ of $\mathcal{L}$ there is a closed tableau~$\pi$ of $\XX{}{c}\varphi$ in ${\mathcal{T}_{\mathsf{cut}}}(\mathcal{L},\overline{\theta})$ with $\lambda(\pi)=O(\mathsf{sz}(\varphi)\cdot (s+1)\cdot 2^{k (s+1)})$.
\end{proposition}
\begin{proof}
Here we follow a very simple procedure, different from the one described in Def.~\ref{analytic-cut-based}. First we apply $\mathcal{R}(\textsl{CUT\/})$ to all the basic proper generalized subformulas of $\varphi$. 
This will generate a tree with $2^{k(s+1)}$ branches. 
Then, for each such branch, we proceed by applying $\mathcal{R}(\textsl{CUT\/})$ to an immediate generalized subformula~$\varphi_i$ of~$\varphi$ such that all of its immediate generalized subformulas already occur in the branch. By construction, such a $\varphi_i$ exists. 
We note that at least one of the two branches thereby generated gives rise to a contradiction and may be closed by applying at most one $\textsl{L\/}$-rule and one closure rule. 
Indeed, by the definition of the system, either the system contains an $\textsl{L\/}$-rule for $\varphi_i$ whose application gives rise to a contradiction on one of the proper generalized subformulas of $\varphi_i$, which we close by means of $\abs$, or, as a trivial case, $\varphi_i$ is not satisfiable by any vector of partial binary prints and we can apply an $\textsl{L\/}$-closure rule, that is, either $\Fs\varphi_i \Longrightarrow \divideontimes$ or $\Ts\varphi_i \Longrightarrow \divideontimes$.
If one of the branches does not close, we can reiterate on it the same procedure, by applying $\mathcal{R}(\textsl{CUT\/})$ to a further proper generalized subformula of $\varphi$ such that all its immediate proper generalized subformulas are in the branch. 
We conclude by noticing that all the initial $2^{k(s+1)}$ branches may be closed by following the above described procedure, i.e., by applying $\mathcal{R}(\textsl{CUT\/})$ to at most all the proper generalized subformulas of $\varphi$, and so linearly in $\textsf{sz}(\varphi)\cdot(s+1)$.\qed
\end{proof}

It is worth noting, here, that the latter result shows that cut-based tableaux are able to p-simulate the truth-table method. Indeed, in general, an $n$-valued truth-table for $\varphi(p_1,\dots,p_k)$ will have $n^k$ rows and $\textsf{sz}(\varphi)$ columns, each entry containing a value in $\mathcal{V}_n$ represented by $\log_2(n)$ bits. But we have also seen in Remark~\ref{sepbounds} that, in optimal cases, the number of necessary separating formulas is $s+1=\log_2(n)$, which renders precisely the $\lambda$-complexity obtained in Prop.~\ref{th:truth-tables-simulation}.

We can further show that ${\mathcal{T}_{\mathsf{cut}}}(\mathcal{L},\overline{\theta})$ is never worse than ${\mathcal{T}}(\mathcal{L},\overline{\theta})$. Intuitively, we must be able to reproduce efficiently in ${\mathcal{T}_{\mathsf{cut}}}(\mathcal{L},\overline{\theta})$ any tableau developed within ${\mathcal{T}}(\mathcal{L},\overline{\theta})$.

\begin{proposition}\label{thm:tabsim}
For every proof $\pi$ in the system ${\mathcal{T}}(\mathcal{L},\overline{\theta})$, there exists a proof $\pi^{\mathsf{cut}}$ with the same root in the system ${\mathcal{T}_{\mathsf{cut}}}(\mathcal{L},\overline{\theta})$ such that ${\mid} \pi^{\mathsf{cut}} {\mid} \le {\mid} \pi {\mid}$.
\end{proposition}
\begin{proof}
It is enough to show that each branching rule of ${\mathcal{T}}(\mathcal{L},\overline{\theta})$ may be efficiently derived in the cut-based system; the nonbranching rules of ${\mathcal{T}}(\mathcal{L},\overline{\theta})$ are already primitive rules of ${\mathcal{T}_{\mathsf{cut}}}(\mathcal{L},\overline{\theta})$.
Let us consider an arbitrary such a branching rule $\mathcal{R}({\textsl{B}_X^{\theta_r{\odot}}})$:
$$
{\X{}\theta_r({\odot}(\varphi_1,\dots,\varphi_k))} {\quad\Longrightarrow\quad} {\mid\mid}_{\overline{Z}\in R_X^{\theta_r{\odot}}}
{\textsl{V}(\varphi_1,\dots,\varphi_k\,;\overline{Z})}
$$

Starting with root $\X{}\theta_r({\odot}(\varphi_1,\dots,\varphi_k))$, in ${\mathcal{T}_{\mathsf{cut}}}(\mathcal{L},\overline{\theta})$ we can follow a procedure consisting in: (i) applying linear elimination rules of the form $\mathcal{R}(\textsl{L\/}_{X \overline{\overline{Y}}}^{\theta_r{\odot}})$ whenever possible; (ii)~if there is no $\overline{\overline{Y}}$ for which the rule $\mathcal{R}(\textsl{L}_{X \overline{\overline{Y}}}^{\theta_r{\odot}})$ may be applied, then there exist $1\le i \le k$ and $0\le t \le s$ such that both $\F{\theta_t(\varphi_i)}$ and $\T{\theta_t(\varphi_i)}$ are present in (at least one branch of) the conclusion of $\mathcal{R}({\textsl{B}_X^{\theta_r{\odot}}})$; then we apply $\mathcal{R}(\textsl{CUT\/})$ on $\theta_t(\varphi_i)$ and repeat the procedure. It is easy to see that, by construction, the amount of information in the simulating tree is not bigger than the one produced by the given rule, i.e.,~each formula in such a simulating tree also occurs in at least one branch of the rule $\mathcal{R}({\textsl{B}_X^{\theta_r{\odot}}})$.
\qed
\end{proof}

The decision procedure proposed in Def.~\ref{analytic-cut-based} is based on an analytic proof strategy that guarantees termination. However, in general there might be better heuristics for guiding the development of a tableau. For example, the canonical procedure given in~\cite{dag:monograph:90} for the KE system for classical logic is, in essence, a generalization of the procedure we adopt in the proof of Theorem~\ref{thm:tabsim}, where we apply linear rules as long as possible and use $\cut$ on some proper generalized subformula only when no other rule is applicable.

\begin{example}[Slim tableaux for fat formulas]\label{slim}
Recalling the fat formulas defined in Ex.~\ref{fat-tableau}, and adopting the optimal proof strategy established above, Fig.~\ref{fat3} depicts a slim closed tableau for $T{:}\Phi_3$ in the system ${\mathcal{T}_{\mathsf{cut}}}(\textrm{\L}_3,\overline{\theta})$. The label \ding{172} denotes seven consecutive applications of $\mathcal{R}(\textsl{L}_{T \langle \uparrow\uparrow\uparrow\uparrow\rangle}^{\land})$, \ding{173} an analytic application of  $\mathcal{R}(\textsl{CUT})$, and~\ding{174} the closure of the branch using 
$\mathcal{R}(\textsl{L}_{T \langle \uparrow\uparrow\uparrow\uparrow\rangle}^{\neg})$ and $\mathcal{R}(\textsl{CUT})$. Labels \ding{175} and \ding{176} correspond respectively to four and to two applications of $\mathcal{R}(\textsl{L}_{T \langle \uparrow\uparrow\rangle}^{\lor})$ .
\end{example}

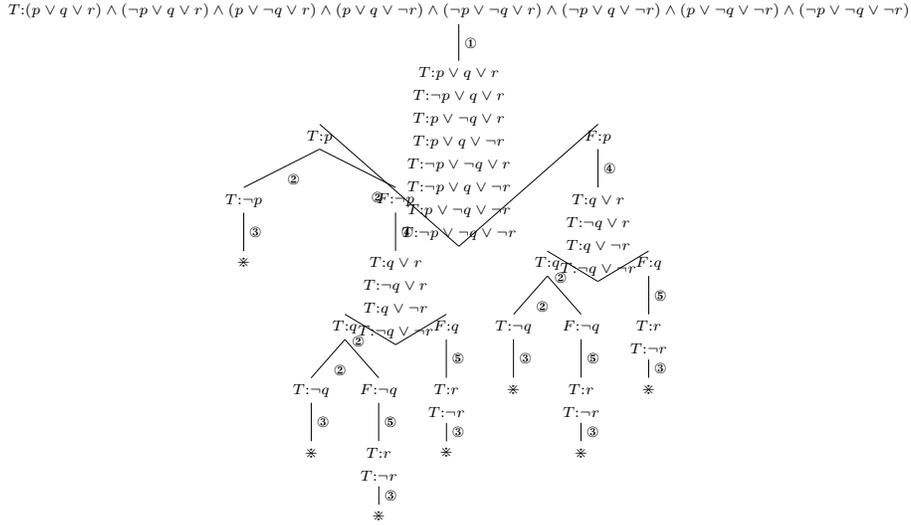
\begin{figure}[ht]\scriptsize
\begin{center}
\scalebox{.8}{
\begin{tikzpicture}[sibling distance=1.14em]
\hspace{-5mm}
\tikzset{every tree node/.style={align=center,anchor=north}}
\tikzset{level 2/.style={level distance=100pt}}
\tikzset{level 4/.style={level distance=60pt}}
\tikzset{level 5/.style={level distance=60pt}}
 \Tree[.{$\Ts (p\vee q\vee r)\wedge(\neg p\vee q\vee r)\wedge(p\vee \neg q\vee r)\wedge(p\vee q\vee \neg r)\wedge(\neg p\vee \neg q\vee r)\wedge(\neg p\vee q\vee \neg r)\wedge(p\vee \neg q\vee \neg r)\wedge(\neg p\vee \neg q\vee \neg r)$}
\edge node[auto=left]{\ding{172}};        
[.{$\Ts p\vee q\vee r$ \\[1mm] $\Ts \neg p\vee q\vee r$ \\[1mm] $\Ts p\vee \neg q\vee r$ \\[1mm] $\Ts p\vee q\vee \neg r$ \\[1mm] $\Ts \neg p\vee \neg q\vee r$ \\[1mm] $\Ts\neg p\vee q\vee \neg r$ \\[1mm] $\Ts p\vee \neg q\vee \neg r$ \\[1mm] $\Ts\neg p\vee \neg q\vee \neg r$}
        \edge node[auto=left]{\ding{173}}; 
        [.{$\Ts p$}
          \edge node[auto=left]{\ding{173}}; 
          [.{$\Ts \neg p$}
            \edge node[auto=left]{\ding{174}};
            [.$\divideontimes$
            ]
          ]
          [.{$\Fs \neg p$}
            \edge node[auto=left]{\ding{175}};
            [.{$\Ts q\vee r$\\[1mm] $\Ts \neg q\vee r$\\[1mm] $\Ts q\vee \neg r$\\[1mm] 
              $\Ts \neg q\vee \neg r$}
              \edge node[auto=left]{\ding{173}}; 
              [.{$\Ts q$}
                \edge node[auto=left]{\ding{173}};
                [.{$\Ts \neg q$}
                  \edge node[auto=left]{\ding{174}};
                  [.{$\divideontimes$}
                  ]
                ]
                [.{$\Fs \neg q$}
                  \edge node[auto=left]{\ding{176}};
                  [.{$\Ts r$\\[1mm] $\Ts \neg r$}
                    \edge node[auto=left]{\ding{174}};
                    [.{$\divideontimes$}
                    ]
                  ]
                ]
              ]
              [.{$\Fs q$}
                \edge node[auto=left]{\ding{176}};
                [.{$\Ts r$\\[1mm] $\Ts \neg r$}
                  \edge node[auto=left]{\ding{174}};
                  [.{$\divideontimes$}
                  ]
                ]
              ]
            ]
          ]
        ]
        [.{$\Fs p$}
          \edge node[auto=left]{\ding{175}};
          [.{$\Ts q\vee r$\\[1mm]$\Ts \neg q\vee r$\\[1mm]$\Ts q\vee \neg r$\\[1mm] 
            $\Ts \neg q\vee \neg r$}
            \edge node[auto=left]{\ding{173}};
            [.{$\Ts q$}
              \edge node[auto=left]{\ding{173}};
              [.{$\Ts \neg q$}
                \edge node[auto=left]{\ding{174}};
                [.{$\divideontimes$}
                ]
              ]
              [.{$\Fs \neg q$}
                \edge node[auto=left]{\ding{176}};
                [.{$\Ts r$\\[1mm] $\Ts \neg r$}
                  \edge node[auto=left]{\ding{174}};
                  [.{$\divideontimes$}
                  ]
                ]
              ]
            ]
            [.{$\Fs q$}
              \edge node[auto=left]{\ding{176}};
              [.{$\Ts r$\\[1mm] $\Ts \neg r$}
                \edge node[auto=left]{\ding{174}};
                [.{$\divideontimes$}
                ]
              ]
            ]
          ]
        ]
       ]
      ] %
\end{tikzpicture}}
\end{center}
\caption{A cut-based closed tableau for $\T{\Phi_3}$.}\label{fat3}	
\end{figure}

\section{Final remarks} \label{conc}

\noindent
The literature of the area abounds with approaches to the study of finite-valued logics based on providing general recipes for producing application-tailored proof systems that could serve as alternative, in supplying decision procedures, to the inefficient truth-table computation.  Here we have described fresh approaches to that study that aim both at being generic and at being efficient.  Our approaches are roughly based on exploiting the logical two-valuedness of the meta-theory of finite-valued logics and on describing their truth-tables in a uniform classic-like fashion, and alongside that quest we expose the non-obvious computational content from the so-called `Suszko's Thesis'.
Writing every single logic with the help of an adequate bivalent semantics or an appropriate classically-labeled (two-signed) tableau system has the obvious advantage of making it easier to compare some given logic to another.  In particular, on what concerns the comparison of different logics, once there is some agreement concerning the language of these logics, one might use our classic-like tableaux to check whether a rule of a certain logic is derivable in another logic (cf.~\cite{mar:men:tfaae4fvl}), and the task of concocting convenient proof tactics that allow for the automation of reasoning within these logics may indeed be easily implemented (cf.~\cite{mar:09d}).  Another advantage of setting up a classic-like framework for a given logic lies in the possibility of dualizing any rule or operator from this logic simply by exchanging truth for falsity, and vice-versa (cf.~\cite{mar:V2V}).  
From the proof-theoretical perspective, differently from the path trodden on early predecessors of the present study, such as~\cite{cal:car:con:mar:humbug:05}, in the present paper we have first presented canonical \textit{cut-free} tableau systems (as in~\cite{ccal:mar:09a}), and have presented the underlying results in full detail, fixing earlier shortcomings of our own approach.
Another great advantage of the present study was the detailed presentation also of an alternative approach based on analytic cut-based tableau systems that allow in general for an exponential speed-up on what concerns proof complexity --- more precisely, that allow for the p(olynomial)-simulation of the brute force truth-tabular procedure.
It might be useful to further extend our complexity-oriented study in order to account for the very cost of the axiom-extraction mechanisms, and even to extend the customary studies on proof complexity in order to measure the apparently non-negligible cost of unifying with long rule premises in the context of large collections of axioms/rules.
We shall leave such extensions, however, as matter for future research.

The received approach to the subject of representation and automation of reasoning in finite-valued logics, in standard references such as~\cite{DBLP:books/el/RV01/BaazFS01,hah:AMVL:HPL}, based on the so-called `signed logic', employs labeled proof formalisms known as `sets-as-signs', which introduce in the language of a given genuinely $n$-valued logic~$\mathcal{L}$ syntactic resources to deal with collections of signs representing the~$n$ truth-values of~$\mathcal{L}$. However, sets-as-signs tableaux seem to enjoy a narrower range of applicability than classic-like tableaux, and in particular their use in logic comparison or dualization is far from obvious.

In labeled deductive systems (cf.~\cite{Gab:LDS:96}) the role of internalizing important semantic information at the syntactical level is routinely played by the use of labels.  A similar goal is often attained by the use of negation in non-signed tableaux for classical and for several non-classical logics.  For instance, in standard references such as~\cite{smu:FOL}, a labeled bivalent statement such as
$$\Fs \varphi\land\psi \quad\Longrightarrow\quad \Fs \varphi \mid\mid \Fs \psi$$
is often replaced by a non-labeled statement like
$$\neg(\varphi\land\psi) \quad\Longrightarrow\quad \neg(\varphi) \mid\mid \neg(\psi).$$
The second statement above clearly goes counter the canonical subformula property, and~$\neg$ in this case obviously plays the role of a separator.  In our current approach we simultaneously utilize labels and separators, in an approach that presupposes generalized notions of subformula, formula complexity and analyticity.  While other recent approaches (cf.~\cite{baa:lah:zam:ijcar12}) have been based on extending the classes of rules that might be called `canonical' in order to accommodate larger sets of labels while insisting on the usual notion of analyticity, our own approach guarantees effectiveness by extending instead the reach of analyticity within a 2-signed labeled environment.

A comment is due here also on our use of the \textit{cut rule}.  On the one hand, in a many-valued setting, different realizations of cut are definable, all corresponding to the fact that a formula can only be assigned a single truth-value in a given interpretation.  While in the sets-as-signs approach, one obvious version of the cut rule will typically consist in expanding the tree with as many branches as the number of truth-values, it is worth noting that in our classic-like framework cut will always be binary branching.
On the other hand, on what concerns cut-based tableaux, some initial advances toward extending them from classical logic to finite-valued logics were sketched in~\cite{dag:monograph:90}, and in~\cite[Chapter 6.1]{hah:book:94} the sets-as-signs approach is claimed (without proof) to produce, in terms of proof complexity, the same improvements obtained in the classical case (p-simulation of truth tables).
To the best of our knowledge, however, the proof of such claim and in fact the first full-fledged approaches to the matter have been done in~\cite{ccal:mar:wollic12}, having the present paper as a sequel.

There are many directions along which the current line of research may be pursued.  For instance, as it has been remarked above, the general axiom/rule extraction mechanisms that we propounded produces statements that may often be streamlined into contracted forms (using standard tools of classical logic at the metalinguistic level).  Instead of first extracting statements in a long form and only simplifying them later, however, one may also propose mechanisms for extracting equivalent sets of rules already in contracted form.  The task of optimizing the rules produced by our mechanism is worth investigating, but we leave it to a future opportunity.  Another interesting line of research concerns other proof formalisms.  While we have chosen to concentrate on tableaux for their relatively unbureaucratic proof theory, choosing sequent systems instead would straightforwardly require us to read closure rules as sequent axioms, and rewrite our $B$-statements contrapositively, rearranging their new left-hand sides in conjunctive normal form (contrast this with Remark~\ref{DNF}).  As sequent systems allow in general for a more flexible meta-theory and a possibly wider application as a logical framework, it would seem appealing to venture an independent study of them.  An adaptation of our algorithms in order to output natural deduction systems, as witnessed by~\cite{eng:hae:per:ND4FVL}, may of course also be the subject of investigation.  The study of other proof formalisms such as reasoning mechanisms based on satisfiability checking or resolution, as it has been done for the sets-as-signs approach in~\cite{FRDP01}, would seem equally welcome.

While one might think that the present study is too limited in the sense of being applicable \textit{only} to finite-valued logics (while many important non-classical logics are known to be infinite-valued), it should here be observed that the main results from Section~\ref{TableauExtraction} also apply when the bivalent statements are not obtained through the method explained in Subsection~\ref{sec:classiclike}, as soon as these statements are based on a generalized notion of compositionality analogous to the one studied in Subsection~\ref{sec:effectiveness}.  
In that case, deductive formalisms based on a generalized form of analyticity would naturally ensue.
Having reached the current milestone, in future work we intend to explore extensions of our present mechanisms to cover other classes of non-classical logics, in particular those defined by genuinely infinite-valued logics, by nondeterministic semantics (cf.~\cite{avr:bn:kon:genFVS}) and by other semantics that presuppose broadening the notion of truth-functionality.

\subsection*{\bf Acknowledgment}

An early draft of the present paper circulated 
for some time 
under the title ``A uniform classic-like analytic deductive formalism for finite-valued logics", containing most details about the extraction of bivalent semantics and cut-free tableaux for finite-valued logics. After the third author joined the other two, the original draft was fully rewritten and extended by the cut-based approach. 
The authors are indebted to Carlos Silva for his very careful reading of several versions of this paper.  Other useful comments were contributed by two anonymous referees, to which the authors are much obliged.

\section*{References} 

 \providecommand{\url}[1]{#1} \newcommand{\noopsort}[1]{}

\end{document}